\documentclass[]{article}
\usepackage{filecontents}
\usepackage [english] {babel}
\usepackage{amsmath,amsfonts,amssymb,amsthm,epsfig,epstopdf,titling,array}
\usepackage [utf8]{inputenc}
\usepackage[T1]{fontenc}
\usepackage{csquotes}
\usepackage{textcomp}
\usepackage{geometry}
\usepackage{bm}
\usepackage[margin=10pt, font=scriptsize, labelfont=bf]{caption}
\usepackage{subcaption}
\usepackage{braket}
\usepackage{environ}
\usepackage{verbatim}
\usepackage{dsfont}
\usepackage{tikz}
\usetikzlibrary{shapes}
\usetikzlibrary{positioning}
\usepackage[colorinlistoftodos]{todonotes}
\usepackage{newunicodechar}
\usepackage{authblk}

\usepackage[colorlinks]{hyperref}
\usepackage{xr}
\newcommand{\xref}[1]{\ref*{PartI-#1}}

\newcommand{\xcref}[1]{\cref*{PartI-#1}} %
\usepackage[capitalize]{cleveref}

\usepackage[style=alphabetic,url=false,doi=false,date=year,giveninits=true,isbn=false,sorting=nyt]{biblatex}
\AtEveryBibitem{\clearfield{issue} \clearfield{number}}
\DeclareFieldFormat
  [article,inbook,incollection,inproceedings,patent,thesis,unpublished]
  {titlecase:title}{\MakeSentenceCase*{#1}}

\newunicodechar{ȩ}{\c{e}}

\newcommand{\orcid}[1]{\href{https://orcid.org/#1}{\textcolor[HTML]{A6CE39}{\aiOrcid}}}

\usetikzlibrary{trees}
\usetikzlibrary{calc}
\usetikzlibrary{decorations.pathmorphing}
\usetikzlibrary{decorations.markings}
\usetikzlibrary{shapes}
\usetikzlibrary{fit}

\makeatletter
\newsavebox{\measure@tikzpicture}
\NewEnviron{scaletikz}[1]{%
  \def\tikz@width{#1}%
  \begin{lrbox}{\measure@tikzpicture}%
  \BODY
  \end{lrbox}%
  \pgfmathparse{#1/\wd\measure@tikzpicture}%
  \BODY
}
\makeatother

\theoremstyle{plain}

\newtheorem{theorem}{Theorem}
\numberwithin{theorem}{section}
\newtheorem{lemma}[theorem]{Lemma}

\newtheorem{proposition}[theorem]{Proposition}

\newtheorem{remark}[theorem]{Remark}
\newtheorem{Remark}[theorem]{Remark}
\newtheorem{corollary}[theorem]{Corollary}

\numberwithin{equation}{section}
\numberwithin{equation}{subsection}
\renewcommand*{\theequation}{%
  \ifnum\value{subsection}=0 %
    \thesection
  \else
    \thesubsection
  \fi
  .\arabic{equation}%
}

\theoremstyle{definition}

\definecolor{darkGreen}{rgb}{0,0.5,0}
\hypersetup{
    colorlinks,
    citecolor=darkGreen,
    filecolor=red,
    linkcolor=blue}

\newcommand{\condset}[2]{\left\{ {#1} \middle| {#2} \right\}}

\DeclareMathOperator{\diam}{diam}

\newcommand{\ul}[1]{{\ensuremath{\underline{#1}}}}
\newcommand{\lis}[1]{{\ensuremath{\overline{#1}}}}

\newcommand{\cc}{\ensuremath{ {\rm c}}}
\newcommand{\dd}{\,\text{\rm d}}             %

\newcommand{\scal}{{\ensuremath{{\rm scal}}}}
\newcommand{\per}{{\ensuremath{{\rm per}}}}
\newcommand{\free}{{\ensuremath{{\rm free}}}}

\newcommand{\B}{\ensuremath{\textup{B}} }
\newcommand{\D}{\ensuremath{\textup{D}} }
\newcommand{\E}{\ensuremath{\textup{E}} }

\newcommand{\M}{\ensuremath{\textup{M}} }
\renewcommand{\P}{\ensuremath{\textup{P}} }
\newcommand{\R}{\ensuremath{\textup{R}} }

\newcommand{\f}{\ensuremath{\textup{f}}}
\newcommand{\p}{\ensuremath{\textup{p}}}

\newcommand{\cA}{\ensuremath{\mathcal{A}}}
\newcommand{\cB}{\ensuremath{\mathcal{B}}}
\newcommand{\cC}{\ensuremath{\mathcal{C}}}
\newcommand{\cD}{\ensuremath{\mathcal{D}}}
\newcommand{\cF}{\ensuremath{\mathcal{F}}}

\newcommand{\cH}{\ensuremath{\mathcal{H}}}
\newcommand{\cL}{\ensuremath{\mathcal{L}}}
\newcommand{\cO}{\ensuremath{\mathcal{O}}}
\newcommand{\cP}{\ensuremath{\mathcal{P}}}
\newcommand{\cQ}{\ensuremath{\mathcal{Q}}}
\newcommand{\cR}{\ensuremath{\mathcal{R}}}
\newcommand{\cS}{\ensuremath{\mathcal{S}}}
\newcommand{\cT}{\ensuremath{\mathcal{T}}}
\newcommand{\cV}{\ensuremath{\mathcal{V}}}

\newcommand{\cW}{\ensuremath{\mathcal{W}}}
\newcommand{\cM}{\ensuremath{\mathcal{M}}}

\newcommand{\calN}{\ensuremath{\mathcal{N}}}

\newcommand{\tcL}{\ensuremath{\widetilde \cL}}
\newcommand{\tcR}{\ensuremath{\widetilde \cR}}

\newcommand{\bbR}{\ensuremath{\mathbb{R}}}

\newcommand{\bC}{\ensuremath{\mathbb{C}}}

\newcommand{\bR}{\ensuremath{\mathbb{R}}}

\newcommand{\bZ}{\ensuremath{\mathbb{Z}}}

\newcommand{\fB}{\ensuremath{\mathfrak{B}}}
\newcommand{\fG}{\ensuremath{\mathfrak{G}}}

\newcommand{\fg}{\ensuremath{\mathfrak{g}}}

\newcommand{\bs}[1]{{\ensuremath{\boldsymbol{#1}}}}

\newcommand{\wt}[1]{\ensuremath{\widetilde{#1}}}

\newcommand{\piecewise}[1]{\begin{cases} #1 \end{cases}}

\newcommand{\floor}[1]{\left\lfloor {#1} \right\rfloor}

\newcommand{\successor}{\ensuremath{\rhd}}

\newcommand{\allct}{\ul{{\upsilon}}}

\newcommand{\ind}{\mathds{1}}

\DeclareMathOperator{\INT}{INT}

\DeclareMathOperator{\Pf}{Pf}

\newcommand{\media}[1]{ { \left\langle #1 \right\rangle}}

\tikzset{vertex/.style={circle,fill=black,inner sep=2pt},
ctVertex/.style={diamond,fill=black,inner sep=2pt},
bigvertex/.style={circle,fill=black,inner sep=4pt},
E/.append style={fill=white,draw},
probeEP/.style={circle,fill=black,draw,inner sep=2pt,
  prefix after command= {\pgfextra{\tikzset{every pin/.style = {pin edge={decorate,decoration={snake,amplitude=2pt,segment length =4pt}}}}}}
},
bareProbeEP/.style={rectangle,fill=black,draw,inner sep=3pt,
  prefix after command= {\pgfextra{\tikzset{every pin/.style = {pin edge={decorate,decoration={snake,amplitude=2pt,segment length =4pt}}}}}}
},
nuEP/.style={circle,fill=white,draw, inner sep=2pt},
linelabel/.style={sloped,above,very near start, inner sep=1pt,execute at begin node=$\scriptstyle,execute at end node=$},
baseline=(current  bounding  box.center),doubled/.style={double distance= 1pt,line width=1.5pt}
}
\pgfdeclarelayer{background}
\pgfsetlayers{background,main}

\newcommand{\tikzvertex}[1]{\tikz[baseline=default]{ \node [#1] {};}}

\DeclareMathOperator\sech{sech}

\makeatletter\@input{partIrefs.tex}\makeatother
\begin{document}
\title{Non-integrable Ising models in cylindrical geometry: \\ Grassmann representation and infinite volume limit}
\author[1]{Giovanni Antinucci}
\author[2,3]{Alessandro Giuliani}
\author[2,*]{Rafael L. Greenblatt}
\affil[1]{\small{Universit\'e de Gen\`eve, Section de math\'ematiques, 2-4 rue du Li\`evre, 1211 Gen\`eve 4, Switzerland}}
\affil[2]{\small{Universit\`a degli Studi Roma Tre, Dipartimento di Matematica e Fisica, L.go S. L. Murialdo 1, 00146 Roma, Italy}}
\affil[3]{\small{Centro Linceo Interdisciplinare {\it Beniamino Segre}, Accademia Nazionale dei Lincei, Palazzo Corsini, Via della Lungara 10,
00165 Roma, Italy.}}
\affil[*]{\small{Present affiliation: Scuola Internazionale Superiore di Studi Avanzati (SISSA), Mathematics Area, Via Bonomea 265, 34136 Trieste, Italy}}
\date{\today}

\maketitle

\begin{abstract} 
In this paper, meant as a companion to \cite{AGG_part1}, we consider a class of
non-integrable $2D$ Ising models in cylindrical domains, and we discuss two key aspects of the multiscale construction of their scaling limit. In particular, we provide a detailed derivation of the Grassmann representation of 
the model, including a self-contained presentation of the exact solution of the nearest neighbor model in the cylinder. Moreover, we prove precise asymptotic estimates of the fermionic 
Green's function in the cylinder, required for the multiscale analysis of the model. We also review the multiscale construction of the 
effective potentials in the infinite volume limit, in a form suitable for the generalization to finite cylinders. Compared to previous works, we introduce a few important simplifications in the localization
procedure and in the iterative bounds on the kernels of the effective potentials, which are crucial for the adaptation of the construction to domains with boundaries. 
\end{abstract}

\tableofcontents

\section{Introduction}\label{sec:intro}

In this article, which is a companion to \cite{AGG_part1}, we consider a class of non-integrable perturbation of the 2D nearest-neighbor Ising model in cylindrical geometry, and discuss some of the key ingredients required 
in the multiscale construction of the scaling limit of the energy correlations in finite domains. The material presented here generalizes and simplifies the approach proposed by two of the authors in \cite{GGM}, 
where a similar problem in the translationally invariant setting was investigated. As discussed extensively in \cite[Section \xref{sec:intro}]{AGG_part1}, which we refer to for additional motivations and references, 
the methods of \cite{GGM}, as well as of several other related works on the Renormalization Group (RG) construction of the bulk scaling limit of non-integrable lattice models at the critical point, 
are insufficient for controlling the effects of the boundaries at the precision required for the construction of the scaling limit in finite domains. This is a serious obstacle in the program of proving conformal invariance of 
the scaling limit of statistical mechanics models \cite{G21}; the goal would be to prove results comparable to the remarkable ones obtained for the nearest neighbor 2D Ising model \cite{Smi10, CS, CHI}, but for a class of non-integrable models, such as perturbed Ising \cite{ADTW} or dimer models \cite{GMT19} in two 
dimensions, via methods that do not rely on the exact solvability of the microscopic model. In this paper and in its companion \cite{AGG_part1}, we attack this program by constructing the scaling limit of the energy 
correlations of a class of non-integrable perturbations of the standard 2D Ising model in the simplest possible finite domain with boundary, that is, a finite cylinder.

Let us define the setting more precisely. For positive integers $L$ and $M$, with $L$ even, we let $G_\Lambda$ be the discrete cylinder with sides $L$ and $M$ in the horizontal and vertical directions, respectively, 
with periodic boundary conditions in the horizontal direction and open boundary conditions in the vertical directions. We consider $G_\Lambda$ as a graph with vertex set $\Lambda=\mathbb Z_L\times (\mathbb Z\cap [1,M])$, where $\mathbb Z_L=\mathbb Z/L\mathbb Z$ (in the following we shall identify the elements of $\mathbb Z_L$ with $\{1,\ldots, L\}$, unless otherwise stated) and edge set $\fB_\Lambda$ consisting of all pairs of 
the form\footnote{If $z=((z)_1,(z)_2)\in\Lambda$ has horizontal coordinate $(z)_1=L$, we use the convention that $z+\hat e_1\equiv (1,(z)_2)$.}
$\left\{ z, z + \hat e_j \right\}$ for $z \in \Lambda, j \in \left\{ 1,2 \right\}$ and $\hat e_1,\hat e_2$ the unit vectors in the 
two coordinate directions. For $x\in\fB_\Lambda$, we let $j(x)$ be the $j$ for which $x= \left\{ z,z+\hat e_j \right\}$ for some $z\in\Lambda$, so that $j(x)=1$ for a horizontal bond and $j(x)=2$ for a vertical bond.
The model is defined by the Hamiltonian
\begin{equation}
\label{eq:HM}
H_{\Lambda}(\sigma)
=-\sum_{x \in \fB_\Lambda} J_{j(x)} \epsilon_x
-\lambda \sum_{X\subset \Lambda} V(X) \sigma_X,
\end{equation}
where: 
$J_1,J_2$ are two positive constants, representing the couplings in the horizontal and vertical directions, $\epsilon_x=\epsilon_x(\sigma) := \sigma_z \sigma_{z'}$ for 
$x= \left\{ z,z'\right\}$; the spin variable $\sigma$ belongs to $\Omega_\Lambda:=\{\pm 1\}^{\Lambda}$, and $\sigma_X:=\prod_{x\in X}\sigma_x$; $V$ is a 
finite range, translationally invariant, even interaction, obtained by periodizing in the horizontal direction a $\Lambda$-independent, translationally invariant, potential on 
$\mathbb Z^2$; finally, $\lambda$ is the strength of the interaction, which can be 
of either sign and, for most of the discussion below, the reader can think of as being small, compared to $J_1,J_2$, but independent of the system size. In the following, 
we shall refer to model \eqref{eq:HM} with $\lambda\neq0$ as to the `interacting' model, in contrast with the standard nearest-neighbor model, which we will refer to as the 
`non-interacting', one of several terminological conventions motivated by analogy with quantum field theory. 
The Hamiltonian defines a Gibbs measure $\media{\cdot}_{\beta,\Lambda}$ depending on the inverse temperature $\beta > 0$, which assigns to any $F : \Omega_\Lambda \to \bR$ the expectation value
\begin{equation} 
\media{F}_{\beta,\Lambda}:= \frac{\sum_{\sigma\in\Omega_{\Lambda}}e^{-\beta H_{\Lambda}(\sigma)}F(\sigma)}{\sum_{\sigma\in\Omega_{\Lambda}} e^{-\beta H_{\Lambda}(\sigma)}}.\end{equation}
The truncated correlations, or cumulants, of the energy observable $\epsilon_x$, denoted $\media{\epsilon_{x_1};\cdots;\epsilon_{x_n}}_{\beta,\Lambda}$, are given by
\begin{equation} \media{\epsilon_{x_1};\cdots;\epsilon_{x_n}}_{\beta,\Lambda}
:=\frac{\partial^n}{\partial A_1 \cdots\partial A_n} \log\media{e^{A_1\epsilon_{x_1}+\cdots
A_n\epsilon_{x_n}}}_{\beta,\Lambda}\Big|_{A_1=\cdots=A_n=0}.
\label{eq:truncated}\end{equation}
For the formulation of the main result, let us fix once and for all an interaction $V$ with the properties spelled out after \eqref{eq:HM}, and 
assume that $J_1/J_2$ and $L/M$ belong to a compact $K\subset (0,+\infty)$. We let $t_l:=t_l(\beta):=\tanh \beta J_l$, with $l=1,2$, and recall that in the 
non-interacting case, $\lambda=0$, 
the critical temperature $\beta_c(J_1,J_2)$ is the unique solution of $t_2(\beta)=(1-t_1(\beta))/(1+t_1(\beta))$. Note that 
there exists a suitable compact $K'\subset (0,1)$ such that whenever $J_1/J_2 \in K$ and $\beta \in [\tfrac12 \beta_c(J_1,J_2),2 \beta_c(J_1,J_2)]$, 
then $t_1,t_2 \in K'$. From now on, we will think $K,K'$ to be fixed once and for all. Moreover, we parameterize the Gibbs measure in terms of $t_l$ as follows: 
$$\media{\cdot}_{\beta,\Lambda}\equiv \media{\cdot}_{\lambda,t_1,t_2;\Lambda}.$$
Given these premises, we are ready to state the main result proven in \cite{AGG_part1}.
\begin{theorem}
\label{prop:main} Fix $V$ as discussed above. Fix $J_1,J_2$ so that $J_1/J_2$ belongs to the compact $K$ introduced above.  
There exists $\lambda_0>0$ and analytic functions
 $\beta_c(\lambda)$, $t_1^*(\lambda)$, $Z_1(\lambda)$, $Z_2(\lambda)$, defined for $|\lambda| \le \lambda_0$, 
such that, for any finite cylinder $\Lambda$ with $L/M\in K$ and any 
$m$-tuple $\bs x=(x_1,\ldots x_m)$ of distinct elements of $\fB_\Lambda$, with $m_1$ horizontal elements, $m_2$ vertical elements, and $m=m_1+m_2\ge 2$, 
	\begin{equation}
		\media{\epsilon_{x_1}; \dots ; \epsilon_{x_m}}_{\lambda,t_1(\lambda),t_2(\lambda);\Lambda}
		=
		\big(Z_1(\lambda)\big)^{m_1}\big(Z_2(\lambda)\big)^{m_2}
		\media{\epsilon_{x_1}; \dots ; \epsilon_{x_m}}_{0,t^*_1(\lambda),t^*_2(\lambda);\Lambda}
		+ R_\Lambda(\bs x), 
		\label{eq:corr_main_statement}
	\end{equation}
where $t_1(\lambda):=\tanh(\beta_c(\lambda)J_1)$, $t_2(\lambda):=\tanh(\beta_c(\lambda)J_2)$ and $t_2^*(\lambda):=(1-t_1^*(\lambda))/(1+t_1^*(\lambda))$. 
Moreover, denoting by $\delta(\bs x)$ the tree distance of $\bs x$, i.e., the cardinality of the smallest connected subset of $\fB_\Lambda$ containing the elements of 
$\bs x$, and by $d=d(\bs x)$ the minimal pairwise distance among the midpoints of the edges in $\bs x$ and 
the boundary of $\Lambda$, for all $\theta\in(0,1)$ and $\varepsilon\in(0,1/2)$ and a
suitable $C_{\theta,\varepsilon}>0$, the remainder $R_\Lambda$ can be
bounded as
\begin{equation} |R_\Lambda(\bs x)|\le  C_{\theta,\varepsilon}^m  |\lambda| m!
\frac1{d^{m+\theta}}\left(\frac{d}{\delta(\bs x)}\right)^{2-2\varepsilon}\;.\label{10b}\end{equation}
\end{theorem}
As a corollary of this theorem, one readily obtain the existence and explicit structure of the scaling limit for the `energy sector' of the interacting model, with quantitative 
estimates on the speed of convergence; see \cite[Corollary \xref{cor:main}]{AGG_part1} and Appendix \ref{Sec:pfaffian_proof}. 

The proof of Theorem \ref{prop:main} is based on a multiscale analysis of the
generating function of the energy correlations, formulated in the form of a Grassmann (Berezin) integral. While the strategy of this proof is based on the same general ideas used 
in \cite{GGM} in the translationally invariant setting, that is, on the methods of the {\it fermionic constructive RG}, 
the presence of boundaries introduces several technical and 
conceptual difficulties, whose solution requires to adapt, improve and generalize the `standard' RG procedure (e.g., in the definition of the `localization procedure', in 
the way in which the kernels of the `effective potentials' are iteratively bounded and in which the resulting bounds are summed over the label specifications, etc.)
As discussed in 
\cite[Section \xref{sec:intro}]{AGG_part1}, which we refer to for additional details, we expect that
understanding how to implement RG in the presence of boundaries or, more in general, of defects breaking translational invariance, 
will have an impact on several related problems, such as the computation of boundary critical exponents
in models in the Luttinger liquid universality class, the Kondo problem, the Casimir effect, and the phenomenon of many-body localization. 

In this paper we give a full presentation of some of 
the key ingredients required in the proof of our main result, namely:
\begin{enumerate}
\item \label{itin1} exact solution of the nearest neighbor model on the cylinder in its Grassmann 
formulation, including multiscale bounds on the bulk and edge parts of the fermionic Green's function (Section \ref{sec:propagators}); 
\item \label{itin2} reformulation of the generating function of 
energy correlations of the interacting model as a Grassmann integral (Section~\ref{sec:gen}); 
\item \label{itin3} tree expansion and iterative bounds on the kernels of the effective potentials
of the interacting theory in the full plane limit, including the computation and proof of analyticity of the interacting critical temperature 
(Section \ref{sec:renexp}). 
\end{enumerate}
The other ingredients, including most of the novel aspects of the RG construction in finite volume, such as the definition of the localization procedure in finite volume, 
the norm bounds on the edge part of the effective potentials and the asymptotically sharp estimates on the correlation functions in the cylinder, are deferred to \cite{AGG_part1}; 
see the end of \cite[Section \ref{sec:intro}]{AGG_part1} for a detailed summary and roadmap of the proof of Theorem \ref{prop:main}.

\medskip

Before starting the technical presentation, let us anticipate in little more detail the contents of the following sections, thus clarifying the main results of this paper. 

\paragraph{Section \ref{sec:propagators}: exact solution of the model in the cylinder.}

The multiscale construction of the interacting theory in the domain $\Lambda$ 
requires a very fine control of the noninteracting model at the critical point, and, in particular, of the structure of its fermionic Green's function, 
which we call the `propagator'; the propagator is nothing but the inverse of a signed adjacency matrix $A$, 
whose definition we recall in Section~\ref{sec:propagators} below \cite[Chapter~IX]{MW}.
The key properties we need, and we prove in Section~\ref{sec:propagators} below (with some -- important! -- technical aspects of the proofs deferred to Appendices \ref{app:diagonalization},
\ref{sec:proofthm2.3} and \ref{sec:proof_2.9}), 
see in particular Eq.~\eqref{msbedec} and Proposition \ref{thm:g_decomposition} below, 
are the following: 
\begin{itemize}
\item multiscale decomposition of the propagator and 
bulk-edge decomposition of the single-scale propagator; 
\item exponentially decaying
pointwise bounds on the bulk- and edge-parts of the single-scale propagators, 
with optimal dimensional bounds (with respect to the scale index) on their $L^\infty$ norms and on their decay rates;
\item Gram representation\footnote{We say that a matrix $g$ admits a Gram representation, if its elements 
$g_{i,j}$ can be written as the scalar product of two vectors $u_i$ and $v_j$ in a suitable Hilbert space.}
of the bulk- and edge-parts of the single-scale propagators, 
with optimal dimensional bounds (with respect to the scale index) on the norms of the Gram vectors. 
\end{itemize}
In reference with the second item, let us remark that the exponential decay needed (and proved below) for the propagator between two 
points $z,z'\in\Lambda$, is in terms of the `right' distance between $z$ and $z'$, namely: 
the standard Euclidean distance on the cylinder between $z$ and $z'$
in the case of the bulk-part of the single-scale propagator; the Euclidean distance on the cylinder between $z,z'$ and the boundary of $\Lambda$, 
in the case of the edge-part of the single-scale propagator. In particular, the exponential decay of the edge-part of the single-scale propagator in the distance 
of $z,z'$ from the boundary of $\Lambda$ 
is of crucial importance for proving improved dimensional bounds on the finite-size corrections to the thermodynamic and correlation functions of the interacting model, 
which are systematically used in  
the conclusion of the proof of Theorem~\ref{prop:main} in the companion paper~\cite{AGG_part1}, see \cite[Section \xref{sec:correlations}]{AGG_part1}.

The proof we give of these key properties is based on an exact diagonalization of the signed adjacency matrix $A$
in terms of the roots of a set of polynomials (this calculation first appeared in~\cite{Greenblatt14}, 
and a similar calculation for a rectangle appears in~\cite{Hucht17a}). It is unlikely that such an explicit diagonalization can be 
obtained in more general domains than the torus, the straight cylinder or the rectangle. Therefore, in order to generalize Theorem \ref{prop:main} to more general domains, it 
would be desirable to prove the properties summarized in the three items above via a more robust method, not based on an explicit diagonalization of $A$. It remains 
to be seen whether the methods of discrete holomorphicity, which allowed to prove the convergence of the propagator in general domains to an explicit, conformally covariant, 
limiting function \cite{CS}, may allow one to prove the desired properties in general domains. 

\paragraph{Section \ref{sec:gen}: Grassmann representation of the generating function.}

In Section~\ref{sec:gen}, we turn our attention to the generating function for the energy correlations 
\begin{equation}
	Z_\Lambda({\bs A}):=
	\sum_{\sigma \in \Omega_\Lambda} \exp\left( 
		\sum_{x \in \fB_\Lambda} \left[ \beta J_{j(x)} +  A_x \right] \epsilon_x
	+ \beta\lambda \sum_{X \subset \Lambda} V( X)  \sigma_X\right).
	\label{eq:Ising_gen}
\end{equation}
that, if computed at a configuration $\bs A$ such that $A_x$ is equal to $A_i$ for $x=x_i$ and zero otherwise, reduces to the combination 
$\media{e^{A_1\epsilon_{x_1}+\cdots A_n\epsilon_{x_n}}}_{\beta,\Lambda}$ appearing in~\eqref{eq:truncated}, up to an overall multiplicative constant, independent of $\bs A$. 

In Proposition~\ref{thm:Ising_to_Grassman} and Eq.\eqref{eq:Xi_def} (adapting a similar result for the torus in \cite{GGM}), we show that the correlations without repeated bonds are the same as those obtained by replacing $Z_\Lambda(\bs A)$ with a Grassman integral of the form
\begin{equation} 
	\Xi_\Lambda(\bs A)
	:=e^{\cW(\bs A)}
	\int P_c^* (\cD \phi) P_m^* (\cD\xi) 
	e^{\cV^{(1)}(\phi,\xi,\bs A)}
	,	
	\label{eq:ZA_goal}
\end{equation}
where $P_c^*$ and $P_m^*$ are Gaussian Grassmann measures associated with the critical, non-interacting Ising model at parameters $t_1^*, t_2^*:=(1-t_1^*)/(1+t_1^*)$, 
with $t_1^*$ a free parameter. Moreover, $\cW(\bs A)$ is a multlinear function of $\bs A$ and $\cV^{(1)}(\phi,\xi,\bs A)$ 
is a Grassmann polynomial whose coefficients are multilinear functions of $\bs A$, 
both of which are defined in terms of explicit, convergent, expansions. As a corollary of Lemma \ref{it:WE_base}, we additionally 
prove that the `kernels' of $\cW(\bs A)$ and $\cV^{(1)}(\phi,\xi,\bs A)$ (i.e., the coefficients of their expansions in $\bs A,\phi,\xi$, 
thought of as functions of the positions of the components of $\bs A,\phi,\xi$ on the cylinder) can be naturally decomposed into sums of a `bulk' part (equal, essentially, to their 
infinite plane limit restricted to the cylinder, with the appropriate boundary conditions) plus an `edge' part (their boundary corrections), exponentially decaying in the appropriate distances. In particular, the edge part of the kernels decays exponentially (on the lattice scale) away from the boundary, a fact that will play a major role in the 
control of the boundary corrections to the correlation functions in \cite{AGG_part1}. 

Let us remark that, in addition to $t_1^*$ and to the inverse temperature $\beta$, the representation \eqref{eq:ZA_goal} has another free parameter, $Z$ (entering the definition of $\cV^{(1)}$); while, 
for the validity of \eqref{eq:ZA_goal}, these parameters can be chosen arbitrarily in certain intervals, in order for 
this representation to produce a convergent expansion for the critical energy correlations of the interacting model, uniformly in the system size, we will need to fix 
$t_1^*,\beta,Z$ appropriately (a posteriori, they will be fixed uniquely by our construction, see below). 
Parameters of this kind are known as \emph{counterterms} in the RG terminology.

\paragraph{Section \ref{sec:renexp}: the RG expansion for the effective potentials in the full plane limit.}

Equation~\eqref{eq:ZA_goal} is the starting point for a multiscale expansion, which is fully presented in the companion paper~\cite{AGG_part1}, see in particular 
\cite[Section \xref{sec:renexp}]{AGG_part1}, but which we summarize here in 
order to provide the context for Section~\ref{sec:renexp}, where we carry out an auxiliary expansion for the full plane limit of the kernels of the `effective potentials'. Such an 
auxiliary expansion, among other things, fixes the values of $t_1^*,\beta,Z$ which are actually used in Theorem~\ref{prop:main}, see Section \ref{sec:fixed_point} below. 

The goal is to iteratively compute \eqref{eq:ZA_goal} in terms of a sequence of effective potentials, defined as follows: at the first step we let 
\begin{equation}
    e^{\cW^{(0)}(\bs A)+\cV^{(0)}(\phi,\bs A)}\propto 
	\int  P_m^* (\cD\xi) 
	e^{\cV^{(1)}(\phi,\xi,\bs A)}	,
	\label{eq:Vcyl_N}
\end{equation}
where $\propto$ means `up to a multiplicative constant independent of $\bs A$';  
the polynomials $\cW^{(0)},\cV^{(0)}$ are specified uniquely by the normalization $\cW^{(0)}(\bs 0)=\cV^{(0)}(0,\bs A)=0$.

We are left with computing the integral of $e^{\cV^{(0)}(\phi,\bs A)}$ with respect to the Gaussian integration $P_c^*(\cD \phi)$ with propagator $\fg^*_c$. As anticipated above, 
in Section \ref{sec:2.2} we decompose the critical propagator $\fg_c^*$ as $\fg^{(\le h)}+\sum_{j=h+1}^0\fg^{(j)}$, for any $h<0$;
correspondingly, in light of the addition formula for Grassman integrals (see e.g.\ \cite[Proposition~1]{GMT17a}), we introduce the sequences $P^{(\le h)}$ and $P^{(h)}$ of Gaussian Grassmann integrations, whose propagators are 
$\fg^{(\le h)}$ and $\fg^{(h)}$, respectively, and satisfy, for any Grassman function $f$,
\begin{equation}
	\int P^{(\le h)} (\cD \phi) f( \phi)
	=
	\int P^{(\le h - 1)} (\cD \phi)
	P^{(h)}(\cD \varphi)
	f(\phi+\varphi).
\end{equation}
We can then iteratively define $\cV^{(h)}$ and $\cW^{(h)}$
with $\cW^{(h)}(\bs 0 ) = \cV^{(h)}(0, \bs A) \equiv 0$ and
\begin{equation}
	e^{\cW^{(h-1)}(\bs A)+\cV^{(h-1)}(\phi,\bs A)}
	\propto
  \int   P^{(h)} (\cD\varphi)
	e^{\cV^{(h)}(\phi+\varphi,\bs A)}.
	\label{eq:Vcyl_h}
\end{equation}
The iteration continues until the scale $h^*=-{\lfloor}\log_2(\min\{L,M\}){\rfloor}$ is reached, at which point we let 
\begin{equation}
	e^{\cW^{(h^*-1)}(\bs A)}\propto
	\int P^{(\le h^*)} (\cD \phi)
	e^{\cV^{(h^*)}(\phi,\bs A)},
	\label{eq:Z_V_hbar}
\end{equation}
giving
\begin{equation}\label{reprXiA}
	\Xi_\Lambda (\bs A)\propto \exp\left(\cW(\bs A)+\sum_{h=h^*-1}^0\cW^{(h)}(\bs A)\right).
\end{equation}
In order to obtain bounds on the kernels of $\cW^{(h)}(\bs A)$ leading to an expansion for the energy correlations that is 
uniform in the system size, at each step it is necessary to isolate from $\cV^{(h)}$ the contributions that tend to expand (in an appropriate norm) 
under iterations: these, in the RG terminology, are the relevant and marginal terms, which we collect in the so-called {\it local part} of $\cV^{(h)}$, denoted by $\cL\cV^{(h)}$. 
In other words, at each step of the iteration, we rewrite $\cV^{(h)}=\cL\cV^{(h)}+\cR\cV^{(h)}$, where, in our case, $\cL \cV^{(h)}$ includes: three terms that are 
quadratic in the Grassmann variables and independent of $\bs A$, depending on a sequence of $h$-dependent parameters which we denote 
$\allct=\{(\nu_h,\zeta_h,\eta_h)\}_{h\le 1}$ and call the {\it running coupling constants}; and two terms that are quadratic in the Grassmann variables and linear in $\bs A$, 
depending on another sequence of effective parameters, $\{Z_{1,h}, Z_{2,h}\}_{h\le 0}$, called the {\it effective vertex renormalizations}. Moreover, $\cR \cV^{(h)}$ is the so-called 
{\it irrelevant}, or {\it renormalized}, part of the effective potential, which is not the source of any divergence.

Such a decomposition corresponds to a systematic reorganization, or `resummation', of the expansions arising from the multiscale computation of the 
generating function. 
The goal will be to show that, by appropriately choosing the parameters $t_1^*,\beta,Z$, which the right side of \eqref{eq:ZA_goal} depends on (and which are 
related via a simple invertible mapping to the initial values of the running coupling constants, $\nu_0,\zeta_0,\eta_0$), 
the whole sequence $\allct$ remains bounded, uniformly in $h^*$; see Section \ref{sec:fixed_point}. 
Under these conditions, we will be able to show that the resulting expansions for 
multipoint energy correlations are convergent, uniformly in $h^*$. Our estimates are based on writing the quantities involved as sums over terms indexed by 
Gallavotti-Nicol{\`o} (GN) trees~\cite{GN,GM01,Ga}, which emerge naturally from the multiscale procedure; the relevant aspects of the definitions of the GN trees 
will be reviewed in Section~\ref{sec:tree_defs} below. 

In order to obtain $L,M$ independent values of these parameters, we study the iteration in the limit $L,M \to \infty$ in \cref{sec:renexp}; we can also restrict to $\bs A = \bs 0$, 
since this already includes all of the potentially divergent terms.  This would superficially appear to involve a number of complications such as defining an infinite-dimensional 
Grassmann integral, but in fact the multiscale computation of the generating function, when understood as an iteration for the kernels of $\cV^{(h)}$, denoted by $V^{(h)}_\Lambda$, has a perfectly straightforward infinite-volume version, which is stated and analyzed in \cref{sec:renexp}. 
The convergence as $L,M\to\infty$ of the finite volume kernels $V_\Lambda^{(h)}$ to the solution $V_\infty^{(h)}$ of the infinite-volume recursive equations for 
the kernels is one of the main subjects of \cite{AGG_part1}, especially \cite[Section \xref{sec:renexp}]{AGG_part1}.  

\cref{sec:renexp} is a reformulation of \cite[Section~3]{GGM}.  We nonetheless present it at length, partly because the treatment of the propagator on the cylinder in \cref{sec:propagators} imposes a different choice of variables which makes the translation of some statements awkard, but mainly in order to take the opportunity to make a number of technical improvements and simplify some unnecessarily obscure aspects of what is already a complicated argument.

Previously, e.g.\ in \cite{GM01,GM05,GGM,Antinucci.thesis}, the localization operator (and consequently the remainder) was defined in terms of the Fourier transform of the 
functions involved. This has the advantage of providing a simple procedure for parametrizing the local part of the effective potential by a finite number of 
running coupling constants, but is quite difficult to apply to non-translation-invariant systems (in \cite{Antinucci.thesis} this led to a peculiar restriction on the dependence of the 
interaction on the system size). Moreover, it makes the treatment of finite size corrections awkard and leads to a convoluted definition of the derivative operators in the remainder
(see \cite{BM01} and \cite{GM13} for the treatment of finite-size corrections on a finite torus via the `standard' definition of localization operator). 
To deal with this, in \cref{sec:interpolation} we introduce a localization operator defined directly in terms of lattice functions, and write the remainder in terms of discrete derivatives using a lattice interpolation procedure. Such definitions naturally admit finite volume counterparts, discussed in \cite[Section \xref{sec:cylinder_multiscale}]{AGG_part1}. 

The strategy used to estimate the interpolation factors in the above cited works also involves decomposing them into components which can be matched with propagators; this involves a number of complications, since it cannot be done in a strictly iterative fashion (this is the problem discussed in \cite[Section~3.3]{BM01}).  
When we handle this issue in Section \ref{sec:formal_bounds}, see in particular \cref{thm:W8:existence}, we instead show iteratively that the coefficient function of the effective 
interaction satisfies a norm bound morally equivalent to exponential decay in the position variables (associated with exponential decay in the scale-decomposed propagators, see 
\cref{thm:g_decomposition}); this then makes it possible to bound the contribution of the interpolation operator immediately, avoiding technical issues such as the `accumulation 
of derivatives' (see \cite[Section~3.3]{BM01} and \cite[end of Section 8.4]{GM01}) or the proof that the Jacobian associated with the change of variables arising from the 
interpolation procedure is equal to $1$ (see \cite[(3.119)]{BM01}). In these aspects, the strategy used in this paper for iteratively estimating the kernels of the effective potentials 
overlaps with \cite{GMR}, which was developed in parallel with the present work.

\section{The nearest-neighbor model}
\label{sec:propagators}

In this section we review some aspects of the exact solution of the nearest-neighbor model ($\lambda=0$), which will play a central role in the multiscale computation of the generating function for the energy correlations
of the non-integrable model, to be discussed in the following sections. In particular, after having recalled the Grassmann representation of the partition function, we explain how to diagonalize the Grassmann action; next, we 
compute the Grassmann propagator and define its multi-scale decomposition, to be used in the following; finally, we compute the scaling limit of the propagator, with quantitative bounds on the remainder. 

\subsection{Diagonalization of the free action}
\label{sec:diagonalization}

\subsubsection{Introduction to the Grassmann variables and representation}
Let us recall the form of the Hamiltonian $H_\Lambda(\sigma)$ \eqref{eq:HM} in the integrable case $\lambda=0$:
\begin{equation*}
H_{\Lambda}(\sigma)=-\sum_{l=1}^2J_l\sum_{z\in \Lambda}\sigma_{z}\sigma_{z+\hat{e}_l},
\end{equation*}
with the understanding that $\sigma_{z+\hat e_2}=0$ for $z=(z_1,z_2)$ such that\footnote{In this section we shall denote the components of $z\in \Lambda$ by 
$z_1, z_2$. We warn the reader that in the following sections the symbols $z_1,z_2$ will mostly be used, instead, 
for the first two elements of an $n$-ple $\bs z$ in $\Lambda^n$ or in $(\mathbb Z^2)^n$, for which we will use the notation 
$\bs z=(z_1,\ldots, z_n)$.} $z_2=M$ and $\sigma_{z+\hat e_1}=\sigma_{z+(1-L)\hat e_1}$ for $z$ such that $z_1=L$. 

As well known \cite[Chapter~VI.3]{MW}, the partition function at  inverse temperature $\beta>0$, $Z_\Lambda=\sum_{\sigma \in \Omega_\Lambda}e^{- \beta H_\Lambda(\sigma)}$, can be written as a Pfaffian, which
admits the following representation in terms of Grassmann variables, see e.g. \cite{S80a} or \cite[Appendix A1]{GM05}:
\begin{equation}
	\begin{split}
	Z_\Lambda=& 2^{LM} (\cosh \beta J)^{L(2M-1)}\int  \mathcal D \Phi\, e^{\cS_{t_1,t_2}(\Phi)},
\end{split}
\label{eq:Z_pfaffian_grassmann}
\end{equation}
where $\Phi=\{(\lis H_{z},H_{z},\lis V_{z},V_{z})\}_{z\in \Lambda}$ is a collection of $4LM$ Grassmann variables 
(we will also use the notation $\{\Phi_i\}_{i\in \mathcal I}$ for $\mathcal I$ a suitable label set with
$4LM$ elements), 
$\mathcal D \Phi$ denotes the Grassmann `differential', 
$$\mathcal D \Phi=\prod_{z\in\Lambda}d\lis H_{z} dH_{z}  d\lis V_{z} dV_{z},$$ and 
\begin{equation} 
	\cS_{t_1,t_2}(\Phi) := \sum_{z\in\Lambda}(t_1 \lis H_{z} H_{z+\hat e_1}+t_2 \lis V_{z} V_{z+ \hat e_2}+\lis H_{z} H_{z}+\lis V_{z} V_{z}+\lis V_{z} \lis H_{z}+ V_{z}
\lis H_{z}+ H_{z} \lis V_{z}+ V_{z} H_{z}) 
	\label{eq:cS_def}
\end{equation}
where $t_{l} =\tanh \beta J_{l}$ for $l=1,2$, and $H_{(L+1,(z)_2)}$, $V_{((z)_1,M+1)}$ should be interpreted as representing $-H_{(1,(z)_2)}$ and $0$, respectively.
The identification of $H_{(L+1,(z)_2)}$ with $-H_{(1,(z)_2)}$ corresponds to {\it anti-periodic} boundary conditions in the horizontal direction for the Grassmann variables: 
these are the right boundary conditions to consider for a cylinder with an even number of sites in the periodic direction, see 
\cite[Eq.(2.6d)]{MW}. For later reference, we let $E_x=\lis H_z H_{z+\hat e_1}$ for a horizontal edge $x$ 
with endpoints $z,z+\hat e_1$, and $E_x=\lis V_z V_{z+\hat e_2}$ for a vertical edge $x$ 
with endpoints $z,z+\hat e_2$. Sometimes, we will call $\{(\lis H_{z},H_{z})\}_{z\in \Lambda}$ the {\it horizontal} variables, 
and $\{(\lis V_{z},V_{z})\}_{z\in \Lambda}$ the {\it vertical} ones. 

The quadratic form $\cS_{t_1,t_2}(\Phi)$ can be written as $\cS_{t_1,t_2}(\Phi)=\frac12 (\Phi, A\Phi)$ for a suitable $4LM\times 4LM$ anti-symmetric matrix $A$ (here $(\cdot,\cdot)$ indicates the standard scalar product for 
vectors whose components are labelled by indices in $\mathcal I$,  
i.e., $(\Phi, A\Phi)=\sum_{i,j\in \mathcal I}\Phi_iA_{ij}\Phi_j$). In terms of this matrix $A$, \eqref{eq:Z_pfaffian_grassmann} can be rewritten as 
$$
Z_\Lambda
=
2^{LM} 
\left( \cosh \beta J_1 \right)^{LM}
\left( \cosh \beta J_2 \right)^{L(M-1)}
{\rm Pf}A.$$ 
[We recall that the Pfaffian of a
$2n\times 2n$ antisymmetric matrix $A$ is defined as
\begin{equation}
{\rm Pf} A:=\frac1{2^n n!}\sum_\pi (-1)^\pi
A_{\pi(1),\pi(2)}...A_{\pi(2n-1),\pi(2n)}; \label{h1}
\end{equation}
where the sum is over permutations $\pi$ of $(1,\ldots,2n)$, with $(-1)^\pi$ denoting the signature.
One of the properties of the Pfaffian is that $({\rm Pf} A)^2={\rm det}A$.]
For later purpose, we also need to compute the averages of arbitrary monomials in the Grassmann variables $\Phi_i$, with $i\in\mathcal I$. These can all be reduced to the computation of the inverse of $A$, thanks to the `fermionic Wick rule': 
\begin{equation} \langle\Phi_{i_1}\cdots \Phi_{i_m}\rangle :=\frac1{{\rm Pf }A}\int \mathcal D\Phi \,
\Phi_{i_1}\cdots \Phi_{i_m}e^{\frac12(\Phi,A\Phi)}={\rm Pf}G\;,\label{2.2c}
\end{equation}
where, if $m$ is even, $G$ is the $m\times m$ matrix with entries 
\begin{eqnarray}
  \label{eq:12}
G_{jk}=
\langle\Phi_{i_j}\Phi_{i_k}\rangle=-[A^{-1}]_{i_j,i_k}
\end{eqnarray}
(if $m$ is odd, the
r.h.s.\ of (\ref{2.2c}) should be interpreted as $0$). Often $\langle\Phi_{i}\Phi_{j}\rangle$ is referred to as the ($ij$ component of the) {\it propagator} of the Grassmann field $\Phi$, or as the {\it covariance} of $\cD \Phi e^{\cS(\Phi)}$; such a form (with a quadratic function in the exponent) is known as a {\it Grassmann Gaussian measure}. 

In the following sections, we shall explain how to compute the Pfaffian of $A$ and its inverse $A^{-1}$, via a block diagonalization procedure. 

\subsubsection{Diagonalization of $\mathcal S_{t_1,t_2}$}
\paragraph{Horizontal direction diagonalization and Schur reduction.}
\label{section_vertical_diagonalization}
By exploiting the periodic boundary conditions in the horizontal direction, we can block diagonalize the Grassmann action by performing a Fourier transform in the same direction: 
for each $z_2\in\{1,2,\ldots,M\}$ we define\begin{equation}
\begin{split}
	&H_{z_2}(k_1)=\sum_{z_1=1}^L e^{ik_1 z_1} 	H_{(z_1, z_2)}, \hspace{4mm} {\lis H}_{z_2}(k_1)=	\sum_{z_1=1}^L e^{ik_1 z_1} \lis H_{(z_1,z_2)},\\
	&V_{z_2}(k_1)=\sum_{z_1=1}^L e^{ik_1 z_1} 	V_{(z_1,z_2)}, \hspace{4mm} {\lis V}_{z_2}(k_1)=\sum_{z_1=1}^L e^{ik_1 z_1} \lis V_{(z_1,z_2)},
\end{split}\label{eq:fourier_transform}
\end{equation}
with $k_1\in\mathcal D_{L}$, where 
\begin{equation}
	\mathcal D_L:=\left\{\frac{\pi (2m-1)}{L}: \ m= -\frac{L}{2}+1,\cdots, \frac{L}2\right\},
\label{momentum_space_discrete_D_L}
\end{equation}
in terms of which 
\begin{equation} 
	\begin{split}
		\mathcal S_{t_1,t_2}(\Phi)
		&= \frac1L\sum_{\substack{k_1\in \mathcal D_L \\ z_2=1,\ldots,M}} \Big[
			(1+t_1e^{-ik_1}) \lis H_{z_2}(-k_1) H_{z_2}(k_1)+
			\lis V_{z_2}(-k_1) V_{z_2}(k_1)	+t_2 \lis V_{z_2}(-k_1) V_{z_2+1}(k_1)
			\\ &
			+\lis V_{z_2}(-k_1) \lis H_{z_2}(k_1)
			+V_{z_2}(-k_1) \lis H_{z_2}(k_1)
			+ H_{z_2}(-k_1) \lis V_{z_2}(k_1)
			+ V_{z_2}(-k_1) H_{z_2}(k_1)
		\Big].
	\end{split}\label{eq:2.1.8}
\end{equation}
[Note that as a consequence of the convention that $V_{(z_1,M+1)}=0$, $V_{M+1}(k_1)$ should also be interpreted as $0$.]
The terms in the second line of \eqref{eq:2.1.8}, mixing the horizontal with the vertical variables, 
can be eliminated by a linear transformation corresponding to a Schur reduction of the coefficient matrix (cf.~\cite[p.~120]{MW}):
\begin{equation} 
\begin{split}
&\begin{bmatrix}\lis H_{z_2}(k_1) \\ H_{z_2}(k_1)\end{bmatrix}=
\begin{bmatrix} \xi_{+,z_2}(k_1) \\ \xi_{-,z_2}(k_1)\end{bmatrix}+
\begin{bmatrix} 
	(1+t_1 e^{ik_1})^{-1} &  - (1+t_1 e^{ik_1})^{-1} \\ 
	(1+t_1 e^{-ik_1})^{-1} & (1+t_1 e^{-ik_1})^{-1} 
\end{bmatrix} 
\begin{bmatrix} \phi_{+,z_2}(k_1) \\ \phi_{-,z_2}(k_1) \end{bmatrix},\\  
&\begin{bmatrix}\lis V_{z_2}(k_1) \\ V_{z_2}(k_1)\end{bmatrix}=
\begin{bmatrix} \phi_{+,z_2}(k_1) \\ \phi_{-,z_2}(k_1)\end{bmatrix}
;
\label{eq:sc}
\end{split}
\end{equation}
Defining a related set of Grassmann variables on $\Lambda$ by
$\phi_{\omega,z}=\frac1{L}\sum_{k_1\in \mathcal D_L}e^{-ik_1 z_1} \phi_{\omega,z_2}(k_1)$ and analogously for $\xi_{\omega,z}$, we then obtain
\begin{equation} \begin{bmatrix}\lis H_{z} \\ H_{z}\end{bmatrix}=
\begin{bmatrix} \xi_{+,z} \\ \xi_{-,z}\end{bmatrix}+\sum_{y=1}^L\begin{bmatrix} s_+(z_1-y) & -s_+(z_1-y) \\ s_-(z_1-y) & s_-(z_1-y) \end{bmatrix} \begin{bmatrix} \phi_{+,(y,z_2)} \\ \phi_{-,(y,z_2)}\end{bmatrix}, \qquad  
\begin{bmatrix}\lis V_{z} \\ V_{z}\end{bmatrix}=
\begin{bmatrix} \phi_{+,z} \\ \phi_{-,z}\end{bmatrix},\label{eq:sc_2}
\end{equation}
where $s_\pm(z_1):=\frac1{L}\sum_{k_1\in \mathcal D_L}\frac{e^{-ik_1z_1}}{1+t_1e^{\pm ik_1}}$. By the Poisson summation formula,  
\begin{equation} \label{spmpoisson}
s_\pm(z_1)=\sum_{n\in\mathbb Z}(-1)^n s_{\infty,\pm}(z_1+nL), \quad \text{with}\quad s_{\infty,\pm}(z)=\int_{-\pi}^\pi \frac{dk_1}{2\pi} \frac{e^{-ik_1z_1}}{1+t_1e^{\pm i k_1}}.
\end{equation}
It is straightforward to check, via a complex shift of the path of integration over $k_1$,
that $s_{\infty,\pm}$ (and, therefore, $s_\pm$) decays exponentially in $z_1$; more precisely, 
$|s_{\infty,\pm}(z_1)|\le e^{-\alpha |z_1|}(1-t_1e^\alpha)^{-1}$ 
for any $\alpha\in[0,-\log t_1)$,
and $|s_\pm (z_1) - s_{\infty,\pm} (z_1)| \le  e^{-\alpha L}(1-t_1e^\alpha)^{-1}$
whenever $|z_1| \le L/2$.

\medskip

In terms of the new variables, the Grassmann action reads $\cS_{t_1,t_2}(\Phi)=\cS_m(\xi)+\cS_c(\phi)$ (the labels `$m$' and `$c$' stand for `massive' and `critical', for reasons that will become clear soon), where
\begin{align} 
	\cS_m(\xi)
	=& 
	\frac1L\sum_{k_1\in \mathcal D_L} \sum_{z_2=1}^M (1+t_1e^{-ik_1}) \xi_{+,z_2}(-k_1) \xi_{-,z_2}(k_1),
	\label{eq:2.1.12}
	\\
	\cS_c(\phi)
	=& 
	\frac1L\sum_{k_1\in \mathcal D_L} \sum_{z_2=1}^M \Big[-b(k_1) \phi_{+,z_2}(-k_1) \phi_{-,z_2}(k_1)
	+t_2  \phi_{+,z_2}(-k_1) \phi_{-,z_2+1}(k_1) 
	\nonumber\\ & \qquad 
	- \frac{i}2\Delta(k_1) \phi_{+,z_2}(-k_1) \phi_{+,z_2}(k_1)+\frac{i}2\Delta (k_1) \phi_{-,z_2}(-k_1) \phi_{-,z_2}(k_1)\Big],
	\label{eq:2.1.13}
\end{align}
with 
\begin{equation}
\Delta(k_1) := \frac{2 t_1 \sin k_1}{|1+t_1e^{ik_1}|^2} ,\qquad 
b(k_1) := \frac{1-t_1^2 }{|1+t_1e^{ik_1}|^2}, \label{eq:b_def}
\end{equation}
where, as a consequence of the convention used above for $V_z$, the term in $\cS_c(\phi)$ involving $\phi_{M+1,-}(k_1)$ should be interpreted as being equal to zero. 
Since $\cS_m$ and $\cS_c$ involve independent sets of Grassmann variables, the Gaussian integral appearing in the partition function factors into a product of two integrals, and the propagators associated with the two terms can be calculated separately.

\paragraph{The `massive' propagator} The calculations for $\cS_m$ are trivial.  Let the antisymmetric matrix $A_m$ be defined by $\cS_m(\xi)=\frac12(\xi,A_m\xi)$. Recall that $\cS_m$ was defined in \eqref{eq:2.1.12}, from which 
$${\rm Pf}A_m=\prod_{k_1\in \mathcal D_L} (1+t_1e^{ik_1})^M,$$
and the propagator is given by the appropriate entry of $A_m^{-1}$, which in the form used in~\eqref{eq:2.1.12} is block-diagonal with $2 \times 2$ blocks such that
$$\langle \xi_{\omega,z_2}(k_1)\xi_{\omega',z_2'}(k'_1)\rangle
=
L\delta_{z_2,z_2'}\delta_{\omega,-\omega'}\delta_{k_1,-k'_1}\frac{\omega}{1+t_1e^{i\omega k_1}}.$$
Therefore, going back to $x$-space, 
\begin{equation}\label{eq:pm} 
\langle \xi_{\omega,z}\xi_{\omega',z'}\rangle=\omega\,\delta_{\omega,-\omega'}  s_\omega(z_1-z_1')\, \delta_{z_2,z_2'},\end{equation}
where $s_\omega(z_1)$ was defined right after \eqref{eq:sc_2}. For later reference, the matrix formed by the elements in \eqref{eq:pm} will be referred to as the massive propagator, and denoted by 
\begin{equation}{\mathfrak g}_{m} (z, z') =\delta_{z_2,z_2'} 
\begin{bmatrix} 0 & s_+(z_1-z_1') \\
-s_-(z_1-z_1') & 0\end{bmatrix}. \label{eq:propmassive}\end{equation}
Recalling that $s_\pm$ decays exponentially, see comments after \eqref{spmpoisson}, we see that $\mathfrak{g}_m(z, z')$ decays exponentially as well, and so 
corresponds to a massive field in the language of quantum field theory. 

\paragraph{The `critical' propagator} 
The antisymmetric matrix $A_c$
defined by $\cS_c(\phi) = \tfrac12 (\phi, A_c \phi)$
can be placed into an explicit block-diagonal form by an ansatz resembling a Fourier sine transform with shifted frequencies; this involves a lengthy but elementary calculation which is detailed in Appendix~\ref{app:diagonalization}.
Here we simply state the result for the critical case
\begin{equation}
  t_1 t_2 + t_1 + t_2 = 1 \Leftrightarrow t_2= \frac{1-t_1}{1+t_1} \Leftrightarrow t_1= \frac{1-t_2}{1+t_2},
  \label{eq:anisotropic_critical_condition}
\end{equation}
which is the only case of relevance for the present work\footnote{The case $t_2 < (1-t_1)/(1+t_1)$, corresponding to the paramagnetic phase can be considered in the same way, while the case
$t_2 > (1-t_1)/(1+t_1)$ is somewhat more complicated since one of the frequencies may be complex \cite{Greenblatt14}.}.
In this case we obtain
\begin{equation}
		\begin{bmatrix}
		\langle \phi_{+,z} \phi_{+,z'} \rangle &
		\langle \phi_{+,z} \phi_{-,z'} \rangle \\
		\langle \phi_{-,z} \phi_{+,z'} \rangle &
		\langle \phi_{-,z} \phi_{-,z'} \rangle 
		\end{bmatrix}
		=  
		\begin{bmatrix}
			g_{++}(z, z') &
			g_{+-}(z, z') \\
			g_{-+}(z, z') &
			g_{--}(z, z') 
		\end{bmatrix}
		=
		{\mathfrak g}_\cc (z, z')
		\label{eq:g_cyl_matrix}
\end{equation}
where
\begin{equation}
	\begin{split}
		{\mathfrak g}_\cc (z, z') :=  &	\frac{1}{L}
		\sum_{k_1\in \mathcal D_L}
		\sum_{k_2 \in \mathcal{Q}_M(k_1)}
		\frac{1}{2N_M(k_1, k_2)}
		e^{-ik_1(z_1-z_1')}
		\\ & \times
		\left\{ e^{-ik_2(z_2-z_2')} \hat{\mathfrak{g}}(k_1, k_2)
		- e^{-ik_2(z_2+z_2')}
	        \begin{bmatrix}
		  \hat{g}_{++}(k_1, k_2) &
		 \hat{g}_{+-}(k_1,-k_2) \\
		  \hat{g}_{-+}(k_1, k_2) & e^{2ik_2(M+1)} 
		  \hat{g}_{--}(k_1, k_2)
		\end{bmatrix}
	      \right\}
	\end{split}
  \label{eq:g_finite_cyl}
\end{equation}
with
\begin{equation}
	\begin{split}
		\hat{\mathfrak g} (k_1, k_2)
		:=&
		\begin{bmatrix}
			\hat {g}_{++} (k_1, k_2) & 
			\hat {g}_{+-}(k_1, k_2) \\
			\hat {g}_{-+}(k_1, k_2) &
			\hat {g}_{--}(k_1, k_2) 
		\end{bmatrix}
		\\
		:= &
		\frac{1}{D(k_1, k_2)}
		 \begin{bmatrix}
			 {-}2 i t_1 \sin k_1 &
			- (1- t_1^2) (1 - B(k_1) e^{-ik_2}) \\
			 (1- t_1^2) (1 - B(k_1) e^{ik_2}) &
			{+} 2 i t_1 \sin k_1 
		\end{bmatrix}
		\label{eq:ghat}
	\end{split}
\end{equation}
where 
\begin{align}
D(k_1, k_2) &:= 2(1 - t_2)^2(1 - \cos k_1)+2(1 - t_1)^2(1 - \cos k_2), \label{defDk1k2}\\
B(k_1)&:= t_2\frac{|1+t_1 e^{ik_1}|^2}{1-t_1^2}, 
\label{eq:B_def}
\end{align}
$\cQ_M(k_1)$ is the set of solutions of {the following equation, thought of as an equation for $k_2$ at $k_1$ fixed}, in the interval $({-\pi,\pi})$:
\begin{equation}
\sin k_2(M+1)=B(k_1)\sin k_2M,
\label{eq:q_condition}
\end{equation}
and
\begin{equation}
	\label{eq:N_M_def}
	N_M(k_1, k_2)=\frac{\frac{d}{dk_2}\left(B(k_1)\sin k_2M-\sin k_2(M+1) \right)}{B(k_1)\cos k_2M-\cos k_2(M+1)}.
\end{equation}

\begin{remark}
	From the above formula it is immediately clear that
	\begin{equation}
		\hat g_{++} (k_1, k_2) = \hat g_{++}(k_1,-k_2) = - \hat g_{++}(-k_1, k_2) = \hat g_{--}(-k_1, k_2)
	\end{equation}
	and 
	\begin{equation}
		\hat g_{+-} (k_1, k_2) = \hat g_{+-} (-k_1, k_2) = - \hat g_{-+} (k_1,-k_2);
	\end{equation}
	furthermore, Equation~\eqref{eq:q_condition} is equivalent to 
	\begin{equation}
		\hat g_{+ -}(k_1, k_2)
		=
		-
		e^{-2ik_2(M+1)} 
		\hat g_{- +} (k_1, k_2)
		,
		\label{eq:ghat_cyl_symmetry}
	\end{equation}
	which therefore holds for all $q \in \cQ_M(k_1)$. Moreover, $N_M(k_1, k_2) = N_M(-k_1, k_2)=N_M(k_1,-k_2)$.
	As we will see in Section~\ref{sec:prop_sym} below, these relationships are closely related to the symmetries of the Ising model on a cylindrical lattice.
	\label{rem:ghat_sym}
\end{remark}

\begin{remark}
	The definition~\eqref{eq:g_finite_cyl} can be extended to all $z,z' \in \bbR^2$;
	then in particular, using the relationships listed in the previous remark, we have
	\begin{equation}
		g_{++} \left( \left( z_1,0 \right), z' \right)
=		g_{++} \left(z, \left( z_1',0 \right) \right)
		=
		g_{+-}  \left( \left( z_1,0 \right), z' \right)
		=
		g_{-+} \left(z, \left( z_1',0 \right) \right)
		=
		0,
		\label{eq:g_cyl_boundary1}
	\end{equation}
	and 
	\begin{equation}
		g_{+-}\left(z, \left( z_1',M+1 \right) \right)
		=
		g_{-+} \left( \left( z_1,M+1 \right), z' \right)
		=
		g_{--}\left( \left( z_1,M+1 \right), z' \right)
		=
		g_{--} \left(z, \left( z_1',M+1 \right) \right)
		=
		0
		,
		\label{eq:g_cyl_boundary2}
	\end{equation}
	for all $z, z'$.\label{rem:g_cyl_boundary}
\end{remark}

\subsection{The critical propagator: multiscale decomposition and decay bounds}
\label{sec:2.2}

In this section, we decompose $\mathfrak{g}_{\cc}(z, z')$ into a sum of terms satisfying bounds which are the main inputs of the multiscale expansion. 

\subsubsection{Multiscale and bulk/edge decompositions}
Let
\begin{equation}
f_\eta(k_1, k_2):= D(k_1, k_2) 
e^{-\eta D(k_1, k_2)},
  \label{eq:fa_def}
\end{equation}
with $D(k_1, k_2)$ as in \eqref{defDk1k2}.   
Note that 
\begin{equation}
  \int_0^\infty f_\eta(k_1, k_2) \dd \eta
  =
  1
  \label{eq:fa_unit}
\end{equation}
as long as $k_1$ and $k_2$ are not both integer multiples of $2\pi$ (and so whenever $k_2 \in \cQ_M(k_1)$). Comparing with Equation~\eqref{eq:ghat} we see immediately that $f_\eta(k_1, k_2) \hat \fg(k_1, k_2)$ is an entire function of both $k_1$ and $k_2$.
The reader may find it helpful in what follows to bear in mind that, for large $\eta$, $f_\eta(k_1, k_2)$ is peaked in a region where $D(k_1, k_2)$ is of the order $\eta^{-1}$, and so $k_1, k_2$ are of order $\eta^{-1/2}$. 

Thanks to \eqref{eq:fa_unit}, $f_\eta$ induces the following multi-scale decomposition 
of $\mathfrak g_\cc$ defined in \eqref{eq:g_finite_cyl}. Let $h^*:=-{\lfloor}\log_2(\min\{L,M\}){\rfloor}$; then, for any $h^* \le h< 0$, 
\begin{equation}\mathfrak g_c(z, z')=\mathfrak g^{(\le h)}(z, z')+\sum_{j=h+1}^0\mathfrak g^{(j)}(z, z'),
\label{eq:gcylmultis}\end{equation}
where
\begin{gather}
  \mathfrak{g}^{(0)} (z, z')
  :=
  \int_{0}^1  \mathfrak{g}^{[\eta]}(z, z')
  \ \dd \eta,
  \label{eq:gN_def}
  \\
  \mathfrak{g}^{(h)} (z, z')
  :=
  \int_{2^{-2h-2}}^{2^{-2h}} 
  \mathfrak{g}^{[\eta]}(z, z')
  \ \dd \eta \quad \text{ for }h^* < h < 0,
  \label{eq:gh_def}
  \\
  \mathfrak{g}^{(\le h)} (z, z')
  :=
 \int_{2^{-2h-2}}^\infty 
  \mathfrak{g}^{[\eta]}(z, z')\ \dd \eta,
  \label{eq:ghbar_def}
\end{gather}
and
\begin{equation}	
  \begin{split}
    {\mathfrak g}^{[\eta]} (z, z') 
    :=
    &
    \frac{1}{L}
    \sum_{k_1 \in \mathcal D_L}
    \sum_{k_2 \in \mathcal{Q}_M(k_1)}
    \frac{1}{2N_M(k_1, k_2)}
    e^{-ik_1(z_1-z_1')}
    f_\eta(k_1, k_2)
    \\ & \quad \times
    \left\{ e^{-ik_2(z_2-z_2')} \hat{\mathfrak{g}}(k_1, k_2)
    - e^{-ik_2(z_2+z_2')}
    \begin{bmatrix}
      \hat{g}_{++}(k_1, k_2) &
      \hat{g}_{+-}(k_1,-k_2) \\
      \hat{g}_{-+}(k_1, k_2) & e^{2ik_2(M+1)} 
      \hat{g}_{--}(k_1, k_2)
    \end{bmatrix}
  \right\}.
\end{split}
\label{eq:ga_def}
\end{equation}
Note that the single-scale propagators preserve the cancellations at the boundary spelled out in Remark \ref{rem:g_cyl_boundary} above, namely, denoting the components of $\mathfrak g^{(h)}$ 
by $g^{(h)}_{\omega\omega'}$, with $\omega,\omega'\in\{\pm\}$, in analogy with \eqref{eq:g_cyl_matrix}, 
	\begin{eqnarray}
		\label{eq:g_cyl_boundary1h}
&&	\hskip-.3truecm	
g_{++}^{(h)} \left( \left( z_1,0 \right), z' \right)
=		g_{++}^{(h)} \left( z, \left( z_1',0 \right) \right)
		=
		g_{+-}^{(h)}  \left( \left( z_1,0 \right), z' \right)
		=
		g_{-+}^{(h)} \left(z, \left( z_1',0 \right) \right)
		=
		0,\\
&& \hskip-.3truecm
g_{+-}^{(h)}\left(z, \left( z_1',M+1 \right) \right)
		=
		g_{-+}^{(h)} \left( \left( z_1,M+1 \right), z' \right)
		=
		g_{--}^{(h)}\left( \left( z_1,M+1 \right), z' \right)
		=
		g_{--}^{(h)} \left(z, \left( z_1',M+1 \right) \right)
		=
		0
		,
		\nonumber\end{eqnarray}
and analogously for $g_{\omega\omega'}^{(\le h)}$. 
Note also that, taking $L,M\to\infty$, the cutoff propagator ${\mathfrak g}_\cc^{[\eta]}(z, z')$  
tends to its infinite-plane counterpart, provided that $z, z'$  are chosen `well inside the cylinder'; in particular, if $z_{L,M}:=(L/2,\lfloor M/2\rfloor)$, then
\begin{equation} \begin{split}\lim_{L,M\to\infty} {\mathfrak g}^{[\eta]} (z_{L,M}+z, z_{L,M}+ z') &=
	\int\limits_{[-\pi,\pi]^2} \frac{\dd k_1\dd k_2}{\left( 2 \pi \right)^2}\,
	e^{-i[k_1(z_1-z_1')+k_2(z_2-z_2')]} 
	f_\eta (k_1,k_2)
	\hat{\mathfrak{g}}(k_1, k_2)\\
&=: \mathfrak{g}_\infty^{[\eta]}(z- z').\end{split}\label{ginftyeta} \end{equation}
For later purposes, we need to decompose the cutoff propagator ${\mathfrak g}^{[\eta]}$ into a `bulk' part which is minimally sensitive to the size and shape of the cylinder, plus a remainder which we call the `edge' part. The bulk part is simply chosen to be 
the restriction of $\mathfrak{g}_\infty^{[\eta]}$ to the cylinder, with the appropriate (anti-periodic) boundary conditions in the horizontal direction: 
\begin{equation} 
	\fg^{[\eta]}_{\B} (z, z') 
	:=	s_L \left( z_1 - z_1' \right)
	\fg_\infty^{[\eta]}\left( \per_L\left( z_1 - z_1' \right), z_2 - z_2' \right)
	\label{eq:gaB_def}\end{equation}
where, recalling that $z_1,z_1'\in\{1,\ldots, L\}$,
\begin{equation}
s_L(z_1-z_1')
	:=
	\piecewise{+1, & |z_1-z_1'|<L/2\\
		\phantom{+}0, & |z_1-z_1'|=L/2 \\
-1, & |z_1-z_1'|>L/2 }
\end{equation}
and
\begin{equation}
	\per_L (z_1)
	:=
	z_1- L \floor{\frac{z_1}{L} + \frac12}.\label{perL}
\end{equation}
The edge part is, by definition, the difference between the full cutoff propagator and its bulk part: 
\begin{equation} \fg^{[\eta]}_{\E} (z, z') 
	:=\fg^{[\eta]} (z, z') 
-\fg^{[\eta]}_{\B} (z,z').\end{equation}
Using these expressions, we define $\fg^{(h)}_\B$ and $\fg^{(h)}_\E$ via the analogues of \eqref{eq:gN_def}-\eqref{eq:gh_def}, with the subscript $\cc$ replaced by $\B$ and 
$\E$, respectively. As a consequence, for any $h^*\le h< 0$, 
\begin{equation}\mathfrak g_{\cc}(z, z')=\mathfrak g^{(\le h)}(z, z')+\sum_{j=h+1}^0(\mathfrak g_{\B}^{(j)}(z, z')+\mathfrak g_{\E}^{(j)}(z, z')).\label{msbedec}\end{equation}
As already observed in Remark \ref{rem:g_cyl_boundary}, 
all the functions involved in this identity can be naturally extended to all $x, y\in\mathbb R^2$ (and, therefore, in particular, to all $z,z'\in\mathbb Z^2$), by interpreting the right side of \eqref{eq:ga_def}, etc., as a function on $\mathbb R^2\times\mathbb R^2$. 

\subsubsection{Decay bounds and Gram decomposition: statement of the main results}\label{sec:2.2.2}

Given the multi-scale and bulk-edge decomposition \eqref{msbedec}, we now intend to prove suitable decay bounds for the single-scale bulk and edge propagators, as well as to show the existence of 
an inner product representation (`Gram representation') thereof. These will 
be of crucial importance in the non-perturbative multi-scale bounds on the partition and generating functions, discussed in the rest of this work, and they are summarized
in the following proposition. 

Given a function $f:\left(\mathbb Z^2\right)^2\to \mathbb C$ (or a function $f:\Lambda^2\to \mathbb C$ extendable to 
$(\mathbb Z^2)^2$, in the sense explained after \eqref{msbedec}), 
we let $\partial_{1,j}$ be the discrete derivative in direction $j$ with respect to the first argument, defined by
 $\partial_{1,j} f (z, z')	:=f(z + \hat{e}_j, z') - f(z, z')$, with $\hat{e}_j$ the $j$-th Euclidean basis vector; an analogous definition holds for $\partial_{2,j}$. 
\begin{proposition}\label{thm:g_decomposition}
There exists constants $c,C$ such that, for any integer $h$ in $[h^*+1,0]$, any $\bs r=(r_{1,1}, r_{1,2}, r_{2,1}, r_{2,2})\in\mathbb Z_+^4$, and any $z, z'\in\Lambda^2$, 
	\begin{enumerate}
			\item \begin{equation}
					\| \bs\partial^{\bs r}{\mathfrak g}^{(h)} (z, z')\|\le 
					C^{1+|\bs r|_1}\bs r! 
					2^{(1+|\bs r|_1)h} e^{-c 2^h \| z- z'\|_1}
					\label{eq:gBh_bounds}
				\end{equation}
				where: the matrix norm in the left side (recall that $\bs\partial^{\bs r}{\mathfrak g}^{(h)} (x, y)$ is a $2\times2$ matrix) is 
				the max norm, i.e., the maximum over the matrix elements, 				
$\bs\partial^{\bs r}:=\prod_{i,j=1}^2 \partial_{i,j}^{r_{i,j}}$, 
                                $\bs r!=\prod_{i,j=1}^2 r_{i,j}!$,  and 
				$\|z\|_1=|\per_L(z_1)|+|z_2|$, see \eqref{perL},
				\label{it:gBh_bounds}
		\end{enumerate}
Moreover, if $z, z'\in\Lambda$ are such that $|\per_L(z_1-z_1')|< L/2-|r_{1,1}|-|r_{2,1}|$, 
\begin{enumerate}
		  \setcounter{enumi}{1}	
		  \item 				\begin{equation}
					\| \bs\partial^{\bs r} {\mathfrak g}^{(h)}_{\E} (z, z')\|\le 
					C^{1+|\bs r|_1} 	\bs r!
					2^{(1+|\bs r|_1)h} e^{-c 2^h {d_\E(z, z')}}
					\label{eq:gEh_bounds}
				\end{equation}
				where $d_\E(z, z') :=\min\{ |\per_L(z_1 - z_1')| + \min\{z_2+z_2',2(M+1)-z_2-z_2'\}, L-|\per_L(z_1-z_1')|+|z_2-z_2'|\}$.
				\label{it:gEh_bounds}\
\end{enumerate}
Finally, there exists a Hilbert space $\cH_{LM}$ with inner product $(\cdot,\cdot)$ including elements 
	$\gamma^{(h)}_{\omega,\bs s, z}$, $\tilde{\gamma}^{(h)}_{\omega, \bs s, z}$,
	$\gamma^{(\le h)}_{\omega,\bs s, z}$, $\tilde{\gamma}^{(\le h)}_{\omega, \bs s, z}$  
	(for $\bs s=(s_1,s_2)\in\mathbb Z_+^2$, $x\in\Lambda$) 
	such that whenever $h^* \le h \le 0$, 		
		\begin{enumerate}
		  \setcounter{enumi}{2}
			\item 
				$\bs \partial^{(\bs s,\bs s')}g_{\omega\omega'}^{(h)} (z, z') \equiv \left( \tilde \gamma^{(h)}_{\omega,\bs s, z},\gamma^{(h)}_{\omega', \bs s', z'} \right) $
				and  
				$\bs \partial^{(\bs s,\bs s')}g_{\omega\omega'}^{(\le h)} (z, z') \equiv \left( \tilde \gamma^{(\le h)}_{\omega,\bs s, z},\gamma^{(\le h)}_{\omega', \bs s', z'} \right) $,
				and
\label{it:gh_gram}
\item $\left|\gamma^{(h)}_{\omega,\bs s, z}\right|^2 , \ 
	\left|\tilde \gamma^{(h)}_{\omega, \bs s, z}\right|^2 , \ 
	\left|\gamma^{(\le h)}_{\omega,\bs s, z}\right|^2 , \ 
	\left|\tilde \gamma^{(\le h)}_{\omega, \bs s, z}\right|^2  
	 \le C^{1+|\bs s|_1}\bs s! 2^{h(1+2|\bs s|_1)}$ where $|\cdot|$ is the norm generated by the inner product 
$(\cdot,\cdot)$.
\label{it:gh_gram_bounds}
\end{enumerate}
\end{proposition}
Combining points~\ref{it:gh_gram} and~\ref{it:gh_gram_bounds}, we see that
\begin{corollary}
  \label{cor:g_le_h_bounds}
  For all $z, \ z' \in \Lambda$, $\bs r \in \bZ_+^4$,  and $h^* \le h \le 0$, 
		\begin{equation}
			\| \bs\partial^{\bs r}  {\mathfrak g}^{(\le h)} (z, z')\|	\le 
			C^{1+|\bs r|_1} \bs r! 
			2^{(1+|\bs r|_1)h}.
			\label{eq:g_le_h_bounds}
		\end{equation}		
\end{corollary}

\begin{remark} Since $\fg^{(h)}_\B=\fg^{(h)}-\fg^{(h)}_\E$, from items \ref{it:gBh_bounds} and \ref{it:gEh_bounds} it follows that 
  the bulk, single-scale, propagator $\fg^{(h)}_\B$ satisfies the same estimate as \eqref{eq:gBh_bounds} for all $x, y\in\Lambda$ allowed in Item~\ref{it:gEh_bounds}, i.e. whenever 
$|\per_L(z_1-z_1')|< L/2-|r_{1,1}|-|r_{2,1}|$. 
The latter restriction on $z, z'$ just comes from the requirement that the discrete derivatives 
do not act on the discontinuous functions $s_L$ and $\per_L$ entering the definition of $\fg^{(h)}_\B$ (and, therefore, of $\fg^{(h)}_\E$); in fact, 
one can easily check from the proof that, if $\bs r=\bs 0$, then \eqref{eq:gEh_bounds} is valid for all $z, z'\in\Lambda$, without further restrictions. 
\end{remark}

\begin{remark}\label{rem:gramfgE}
All the estimates stated in the proposition are uniform in $L,M$, therefore, they remain valid for the $L,M\to\infty$ limit of the propagators. 
In particular, $\fg_\infty^{(h)}(z- z')$ and $\fg^{(\le h)}_\infty(z- z')$ satisfy the same estimates as 
\eqref{eq:gBh_bounds} and \eqref{eq:g_le_h_bounds}, respectively. 
Similarly, the Gram representation stated in items \ref{it:gh_gram} and \ref{it:gh_gram_bounds} is also valid for $\fg_\infty^{(h)}$. 
Therefore, if $|z_1-z_1'|\le L/2-|s_1|-|s'_1|$, also $\bs\partial^{(\bs s,\bs s')}\fg_\E^{(h)}(\bs x,\bs y)=\bs\partial^{(\bs s,\bs s')}\fg^{(h)}(z, z')-
\bs\partial^{(\bs s,\bs s')}\fg_\infty^{(h)}(z- z')$ admits a Gram representation, with qualitatively the same Gram bounds (and, of course, 
such a representation can be extended by anti-periodicity to all $z, z'$ such that $|\per_L(z_1-z_1')|\le L/2-|s_1|-|s'_1|$).
\end{remark}

\begin{remark}
	\label{rem:g_massive_rescaled}
	If we rename the massive propagator in \eqref{eq:propmassive} as $\fg_m(z, z') =: \fg^{(1)}(z, z')$ and then use it to define $\fg_\infty^{(1)}$, $\fg_\B^{(1)}$, and $\fg_\E^{(1)}$ in the same way that we did above with $\fg^{[\eta]}$,
it is straightforward to check that the estimates 
in items \ref{it:gBh_bounds} and \ref{it:gEh_bounds} of the proposition remain valid for $h=1$. Similarly, the reader can check that the proof of items \ref{it:gh_gram} and 
\ref{it:gh_gram_bounds} given below can be straightforwardly applied to the case $h=1$, as well.
\end{remark}

\begin{remark}
  \label{rem:g_analytic}
  Again since $D(k_1, k_2)$ is exactly the denominator in the definition~\eqref{eq:ghat},
  it is easily seen that $f_\eta(k_1, k_2) \hat{\fg}(k_1, k_2)$ is an entire function of $t_1$.
  All of the bounds in Proposition~\ref{thm:g_decomposition} are obtained by writing the relevant quantities as absolutely convergent integrals or sums in $k_1$, $k_2$, and $\eta$; since these bounds are locally uniform in $t_1$ as long as it is bounded away from $0$ and $1$, we also see that all of the propagators are analytic functions of $t_1$ with all other arguments held fixed.
\end{remark}

The proof of Proposition \ref{thm:g_decomposition} is in Appendix \ref{sec:proofthm2.3}.

\subsection{Asymptotic behavior of the critical propagator}
\label{sec:asymp_prop}

Although our main result, \cref{prop:main}, involves correlation functions for a finite lattice, since the continuum limit of the energy correlation functions of the non-interacting model is well understood \cite{Hon.thesis,CCK}, as a result we also obtain a characterization of the scaling limit \cite[\xcref{cor:main}]{AGG_part1}.
For completeness, we give here a description of the scaling limit of the critical propagator, from which the non-interacting energy correlation functions are easily calculated. 

We rescale the lattice as follows: fix two positive constants $\ell_1,\ell_2$ (no condition 
on the ratio $\ell_1/\ell_2$), and let $L=2\lfloor a^{-1}\ell_1/2\rfloor$, $M=\lfloor a^{-1}\ell_2\rfloor$ for $a>0$ the lattice mesh.
Let $\Lambda^a:=a\Lambda$ and 
$\Lambda_{\ell_1,\ell_2}$ the continuum cylinder. We also let $\|\cdot\|$ indicate the Euclidean distance on the cylinder $\Lambda_{\ell_1,\ell_2}$. 
Given $z, z'\in\Lambda_{\ell_1,\ell_2}$, we let 
\begin{equation}
\fg_{\cc, a}(z, z')= a^{-1}\fg_\cc(\lfloor a^{-1} z\rfloor, \lfloor a^{-1} z'\rfloor). \label{eq:g_cyl_scaling_finite_a}
\end{equation}
The main result of this section concerns the limiting behavior of $\fg_{\cc, a}$ as $a\to 0$. 

\begin{proposition}\label{thm:g_scaling}
	Given $\ell_1,\ell_2>0$, there exist $C,c > 0$ such that for all $x, y\in \Lambda_{\ell_1,\ell_2}$ such that $ z\neq z'$ and $a>0$ for which $a(\min\{\ell_1,\ell_2,\|z- z'\|\})^{-1} \le c$,
\begin{equation}
\fg_{\cc, a}(z, z')= \fg_\scal(z, z')+ R_a(z, z'),\end{equation}
and $\|R_a(z, z')\|\le Ca \big(\min\{\ell_1,\ell_2,\|z- z'\|\}\big)^{-2}$, where 
\begin{eqnarray} 
  \label{eq:g_scal_cyl_def}
  &&
\fg_\scal(z, z')=\sum_{\bs n\in\mathbb Z^2}(-1)^{\bs n} \Biggl\{ \fg^{\scal}_\infty(z_1-z_1'+n_1\ell_1,z_2-z_2'+2n_2\ell_2)\\
&&+
\begin{bmatrix} 
-g^{\scal}_{1}(z_1-z_1'+n_1\ell_1,z_2+z_2'+2n_2\ell_2) & g^{\scal}_{2}(z_1-z_1'+n_1\ell_1,z_2+z_2'+2n_2\ell_2) \\
-g^{\scal}_{2}(z_1-z_1'+n_1\ell_1,z_2+z_2'+2n_2\ell_2) & g^{\scal}_{1}(z_1-z_1'+n_1\ell_1,z_2+z_2'+2(n_2-1)\ell_2) \end{bmatrix}
\Biggr\},\nonumber\end{eqnarray}
and where, letting 
\begin{equation}
g^{\scal}(z_1,z_2):=\frac{1}{t_2(1-t_2)} \iint_{\mathbb R^2}\frac{\dd k_1\dd k_2}{(2\pi)^2} e^{-ik_1z_1-ik_2z_2}\frac{-i k_1}{k_1^2+k_2^2}
=-\frac1{{2\pi} t_2(1-t_2)}\frac{z_1}{z_1^2+z_2^2},
\label{eq:g_scale_scalar}
\end{equation}
we denoted $g_1^{\scal}(z_1,z_2):=g^{\scal}(\frac{z_1}{1-t_2},\frac{z_2}{1-t_1})$, $g_2^{\scal}(z_1,z_2):=g^{\scal}(\frac{z_2}{1-t_1},\frac{z_1}{1-t_2})$, and
\begin{equation}
  \fg_\infty^\scal(z_1,z_2)
    :=
	\begin{bmatrix}
    g_1^\scal (z_1,z_2) & g_2^\scal (z_1,z_2) \\
    g_2^\scal (z_1,z_2) & - g_1^\scal (z_1,z_2)
	\end{bmatrix}
  .
\end{equation}
\end{proposition}

The proof of \cref{thm:g_scaling} is given in \cref{sec:proof_2.9}. 
It is easy to see from the definition of $\fg_\cc^{\scal}$ that its entries vanish for $z_2 = 0$ and/or $z_2 = \ell_2$ and/or $z_2' = 0$ and/or $z_2' = \ell_2$ in a 
fashion analogous to the one discussed in Remark~\ref{rem:g_cyl_boundary}.

\subsection{Symmetries of the propagator}
\label{sec:prop_sym}

Note that the action $\cS_{t_1,t_2}(\Phi)$ of Equation~\eqref{eq:cS_def} is unchanged by the substitutions
\begin{equation}
	\lis H_{z} \to i H_{\theta_1 z}, \quad
	H_{z} \to i \lis H_{\theta_1 z}, \quad
	\lis V_{z} \to i \lis V_{\theta_1 z}, \quad
	V_{z} \to  - i V_{\theta_1 z}
	\label{eq:symm_horizontal_HV}
\end{equation}
with $\theta_1 (z_1, z_2) := (L+1-z_1, z_2)$,
or
\begin{equation}
	\lis H_{z} \to -i \lis H_{\theta_2 z}, \quad
	H_{z} \to  i  H_{\theta_2 z}, \quad
	\lis V_{z} \to i V_{\theta_2 z}, \quad
	V_{z} \to i \lis V_{\theta_2 z}
	\label{eq:symm_vertical_HV}
\end{equation}
where $\theta_2 (z_1, z_2) := (z_1, M+1 - z_2)$.
These transformations, of course, correspond to the reflection symmetries of the Ising model on a cylinder. 
In terms of the $\phi,\xi$ variables, it is easy to see from Equation~\eqref{eq:sc_2} that these substitutions are equivalent to 
\begin{equation}
	\phi_{\pm,z} \to \Theta_1 \phi_{\pm,z} := \pm i \phi_{\pm, \theta_1 z}, \quad
	\xi_{ \pm,z} \to \Theta_1 \xi_{\pm,z} := i \xi_{\mp,\theta_1 z}
	\label{eq:symm_horizontal_phixi}
\end{equation}
and
\begin{equation}
	\phi_{\pm,z} \to \Theta_2  \phi_{\pm,z} :=  i \phi_{\mp,\theta_2 z}, \quad
	\xi_{ \pm, z} \to \Theta_2 \xi_{ \pm, z} :=  \mp i \xi_{\pm,\theta_2 z}
	.
	\label{eq:symm_vertical_phixi}
\end{equation}
With a little more notation we can write this more compactly:
letting $\phi_{\omega,z}$ denote $\phi_{+,x},\phi_{-,x},\xi_{+,x},\xi_{-,x}$ for $\omega=1,-1,+i,-i$, respectively, and letting $\theta_j \omega := (-1)^{j+1} \lis \omega$ for $\omega \in \bC$, Equations~\eqref{eq:symm_horizontal_phixi} and~\eqref{eq:symm_vertical_phixi} can be combined into
\begin{equation}
	\Theta_j \phi_{\omega,z}
	:=
	i \alpha_{j,\omega}
	\phi_{\theta_j \omega, \theta_j x }
	\label{eq:symm_master_phixi}
\end{equation}
where $\alpha_{j,\omega}$ is $-1$ if ($j=1$ and $\omega = -1$) or ($j=2$ and $\omega = i$), and 1 otherwise.
Since these transformations act on the vector $\Psi$ as orthogonal matrices,
this is equivalent to the symmetry of the coefficient matrix $A$ (and therefore its inverse) under the associated similarity transform, and since $\fg_\cc$ is just a diagonal block of $A^{-1}$ we have
\begin{equation}
		\fg_\cc (z, z')
	=
		\begin{bmatrix}
		-g_{\cc; + +}  
			&
			g_{\cc; + -}  
			\\
			g_{\cc; - +}  
			&
		-g_{\cc; - -}  
		\end{bmatrix}
	\left( \theta_1 z, \theta_1 z' \right)
	\label{eq:symm_horizontal_g}
\end{equation}
and
\begin{equation}
	\fg_\cc  (z, z')
		=
	-
		\begin{bmatrix}
			g_{\cc; - -} 
			&
		g_{\cc; -+}  
			\\
		g_{\cc; +-}  
			&
			g_{\cc; + +} 
		\end{bmatrix}
	\left( \theta_2  z, \theta_2 z' \right)
	.
	\label{eq:symm_vertical_g}
\end{equation}
These relationships can also be recovered from Equation~\eqref{eq:g_finite_cyl}, using the observations on $\hat \fg_\infty$ in Remark~\ref{rem:ghat_sym}.
This latter point is helpful because, since $f_\eta$ is even in both $k_1$ and $k_2$, it also applies to $\fg_\cc^{[\eta]}$.  Taking the appropriate $L,M \to \infty$ limit we also obtain
\begin{equation}
		\fg_\infty^{[\eta]} (z_1,z_2)
	=
	\begin{bmatrix}
			-g^{[\eta]}_{\infty; ++}
		&
			g^{[\eta]}_{\infty; +-}
		\\
			g^{[\eta]}_{\infty; -+}
		&
			-g^{[\eta]}_{\infty; --}
	\end{bmatrix}
		\left( -z_1, z_2 \right)
		=
		-
		\begin{bmatrix}
			g^{[\eta]}_{\infty; --}
			&
			g^{[\eta]}_{\infty; -+}
			\\
			g^{[\eta]}_{\infty; +-}
			&
			g^{[\eta]}_{\infty; ++}
		\end{bmatrix}
		\left( z_1, - z_2 \right)
		,
		\label{eq:ga8_symmetries}
\end{equation}
	which also implies that $\fg_\B^{[\eta]}$ has the symmetries \eqref{eq:symm_horizontal_g} and~\eqref{eq:symm_vertical_g}.
	Applying the differences and integrals in the relevant definitions we see that
\begin{lemma}
  $\fg^{(h)}_\cc$, $\fg^{(\le h)}_\cc$, $\fg^{(h)}_\B$, and $\fg^{(h)}_\E$ 
  all have the symmetries \eqref{eq:symm_horizontal_g} and~\eqref{eq:symm_vertical_g} for any $h_* \le h \le 0$,
  and $\fg_\infty^{(h)}$ has the symmetries~\eqref{eq:ga8_symmetries} for any $h \le 0$.
	\label{eq:ga_symmetries}
\end{lemma}
For $\fg^{(1)}_\cc \equiv \fg_m$ 
we have
	\begin{equation}
  \fg_m( z, z')
		=
  -
		\begin{bmatrix}
		g_{m; --}  
			&
		g_{m; -+}  
			\\
		g_{m; +-}  
			&
		g_{m; ++}  
		\end{bmatrix}
	\left( \theta_1 z, \theta_1 z' \right)
		=
		\begin{bmatrix}
		-g_{m; + +}  
			&
		g_{m; + -}  
			\\
		g_{m; - +}  
			&
		-g_{m; - -}  
		\end{bmatrix}
	\left( \theta_2 z, \theta_2 z' \right)
  \label{eq:g_m_symmetries}
	\end{equation}
which similarly extends to $\fg^{(1)}_\infty$, $\fg^{(1)}_\B$, and $\fg_\E^{(1)}$.

\section{Grassmann representation of the generating function}\label{sec:gen}

In this section we rewrite the generating function of the energy correlations for the Ising model \eqref{eq:HM} with finite range interactions as an {\it 
interacting} Grassmann integral, and we set the stage for the multiscale integration thereof, to be discussed in the following sections.
The estimates in this and in the following section are uniform for $J_1/J_2, L/M\in K$ and $t_1,t_2\in K'$, but may depend upon the choice of $K,K'$, 
with $K,K'$ the compact sets introduced before the statement of Theorem \ref{prop:main}. 
As anticipated there, we will think of $K,K'$ as being fixed once and for all and, for simplicity, we will not track the 
dependence upon these sets in the constants $C,C', \ldots, c,c', \ldots, \kappa, \kappa', \ldots$, appearing below. Unless otherwise stated, 
the values of these constants may change from line to line. 

\medskip

Our goal is to show that the generating function \eqref{eq:Ising_gen} of the energy correlations can be replaced, for the purpose of computing multipoint 
energy correlations at distinct edges, by the Grassmann generating function \eqref{eq:ZA_goal}. 
Our main representation result for the Taylor coefficients at $\bs A=\bs 0$ of the logarithm of $Z_\Lambda({\bs A})$, analogous to \cite[Proposition 1]{GGM}, is the following. 
\begin{proposition}
	For any translation invariant interaction $V$ of finite range, there exists $\lambda_0=\lambda_0 (V)$ 
	such that, for any $|\lambda| \le \lambda_0 (V)$, 
	$$ 
	\left. \frac{\partial}{\partial A_{x_1}} \dots \frac{\partial}{\partial A_{x_n}} \log Z_\Lambda({\bs A}) \right|_{\bs A = \bs 0} =  
	\left. \frac{\partial}{\partial A_{x_1}} \dots \frac{\partial}{\partial A_{x_n}} \log \wt \Xi_\Lambda(\bs A) \right|_{\bs A = \bs 0} 
	$$
	as long as $ n \ge 2$ and the $x_j$ are distinct, where
\begin{equation}
	\wt \Xi_\Lambda ({\bs A}):= e^{\mathcal W(\bs A)}\int \cD \Phi \; e^{\cS_{t_1,t_2}(\Phi) 
		+ \cV(\Phi, \bs A)}
	\label{eq:Ising2Grassman}
			\end{equation}
where $\cS_{t_1,t_2}$ was defined in \eqref{eq:cS_def} and, recalling that $E_x$ is the Grassmann binomial defined after \eqref{eq:cS_def}:
	\begin{enumerate}
	\item 		\label{it:B_cyl_base}
	\begin{equation}
\begin{split}			\cV (\Phi, \bs A)
			&=
			\sum_{x\in \fB_\Lambda}(1-t_{j(x)}^2)E_x A_x 
			+\sum_{\substack{ X,Y \subset \fB_\Lambda \\ X \neq \emptyset}}
			W_\Lambda^{{\rm int}} (X,Y)\prod_{x \in X} E_x	\prod_{x \in Y} A_x
			\\ &\equiv \cB^{\free}(\Phi,\bs A)+\cV^{\rm int}(\Phi,\bs A)
		\end{split}
		\label{eq:B_expansion_bis}
		\end{equation}
		where, for any $n\in \mathbb N$, $m\in\mathbb N_0$, and suitable positive constants $C,c,\kappa$, 
		\begin{equation}
			\sup_{x_0\in\fB_\Lambda}\sum_{\substack{X,Y\subset \fB_\Lambda:\,  x_0\ni X\\ |X|=n,\, |Y|=m}}|W^{\rm int}_\Lambda(X,Y)| e^{c \delta( X \cup Y)}	\le
			C^{m+n}|\lambda|^{\max (1,\kappa (m+n))}
			\label{eq:B_base_decay}
		\end{equation}
		and $\delta(X)$, for $X \subset \fB_\Lambda$, denotes the size of the smallest $Z \supset X$ which is the edge set of a connected subgraph of $\fG_\Lambda$.
\item 	\begin{equation}
			\cW (\bs A)
			=\sum_{\substack{ Y \subset \fB_\Lambda \\ |Y|\ge 2}}
			w_\Lambda(Y)	\prod_{x \in Y} A_x
			\label{eq:B_expansion_tris}
		\end{equation}
		where, for any $m\in\mathbb N$, and the same $C,c,\kappa$ as above, 
		\begin{equation}
			\sup_{x_0\in\fB_\Lambda}\sum_{\substack{Y\subset \fB_\Lambda:\\  x_0\ni Y,\, |Y|=m}}|w_\Lambda(Y)| e^{c \delta(Y)}\le
			C^{m}|\lambda|^{\max (1,\kappa m)}.
			\label{eq:B_base_decay2}
		\end{equation}
		\label{it:B_cyl_base2}
	\item $W^{\rm int}_\Lambda, w_\Lambda$, considered as functions of $\lambda$, $t_1$, and $t_2$, can be analytically continued to any complex $\lambda, t_1, t_2$ such that $|\lambda| \le \lambda_0$ and $|t_1|, |t_2|\in K'$, with 
	$K'$ the same compact set introduced before the statement of Theorem \ref{prop:main}, and the analytic continuations satisfies the same bounds above.
	\end{enumerate}
	\label{thm:Ising_to_Grassman}
\end{proposition}
\begin{proof}
The proof is basically the same as \cite[Proposition~1]{GGM}, so we refer to that for the details. 
Note that the restriction in~\cite{GGM} to a pair interaction is unimportant, since any even interaction of the form $V(X)\sigma_X$ with $X\subset \Lambda$ can always be written 
as a product of factors of $\epsilon_x$, in analogy with the rewriting $\sigma_z\sigma_z'=\frac12(U_{z,z'}+D_{z,z'})$ discussed after \cite[Equation~(2.8)]{GGM}; the only 
difference in the current setting is that the `strings' graphically associated with $U_{z,z'}$ and $D_{z,z'}$, see~\cite[Figure~3]{GGM}, are replaced by other figures, whose specific 
shape depends on $V(X)$ and that one should use $t_1$ or $t_2$ in place of $t$ as appropriate\footnote{We take the occasion to 
point out that \cite[Figure~4]{GGM} contains a mistake: the string $S_2$ depicted there is not allowed by the conventions of \cite{GGM}: the shape $S_2$ can only be obtained 
as the union of two appropriate strings.}. Note that the set of strings associated with a pair interaction, or the set of more general figures associated with a generic even 
interaction, is, or can be chosen to be, invariant under the basic symmetries of the model, namely horizontal translations, and horizontal and vertical reflections; therefore, in the 
following, we shall assume that such a graphical representation is invariant under these symmetries. 

By proceeding as in \cite{GGM} we get the analogue of \cite[Eq.(2.20)]{GGM}, namely 
	\begin{equation}
		\mathcal W(\bs A)+\cV(\Phi,{\bs A})	=
		\sum_{\Gamma \in C_\Lambda}
		\varphi^T (\Gamma)
		\prod_{\gamma \in \Gamma}
		\zeta_G (\gamma)
		=
		\sum_{X,Y \subset \fB_\Lambda}W_\Lambda (X,Y)\prod_{x \in X} E_x\prod_{x \in Y} A_x\label{sa2.20}
	\end{equation}
	where $C_\Lambda$ is the set of multipolygons in $\Lambda$, $\varphi^T$ is the Mayer's coefficient, and $\zeta_G$ is the activity 
	of the polygon $\gamma$, which is a polynomial in the $E_x,A_x$ for edges $x$ in $\gamma$ (for more details about the notation and more precise definitions, we refer to \cite{GGM}).
The terms with $X=\emptyset$ contribute to $\cW(\bs A)$ (that is, we let $w_\Lambda(Y) := W_\Lambda (\emptyset,Y)$), 
while those with $X\neq \emptyset$ contribute to $\cV(\Phi,{\bs A})$
(note that, for the purpose 
of computing the derivatives of $\log\wt\Xi_\Lambda(\bs A)$ of order 2 or more, the terms with $|X|=0$ and $|Y|=0,1$ can be dropped from the definition of $\cW(\bs A)$, 
and we do so). The explicit computation of the term independent of $\lambda$, which has $X=Y\in\fB$, 
leads to the decomposition in \cref{eq:B_expansion_bis}. The bounds \cref{eq:B_base_decay,eq:B_base_decay2} follow directly 
from the bounds in \cite{GGM}, see e.g. \cite[Eq.(2.25)]{GGM} and following discussion. 

Finally, the analyticity property (which was used implicitly in \cite{GGM}) follows by noting that we have defined all of the quantities of interest as uniformly absolutely convergent sums of terms which are themselves analytic functions of $\lambda, t_1, t_2$ as long as the absolute values of these parameters belong to the appropriate intervals. \end{proof}

With a view towards the analysis of finite size effects in \cite{AGG_part1} (and, in particular, towards the claims done in \cite[Section \xref{sec:2.2}]{AGG_part1}
after the statement of \cite[Proposition \xref{prop:repr}]{AGG_part1}), it is convenient to decompose the kernel $W^{\rm int}_\Lambda$ 
of $\cV^{\rm int}(\Phi,\bs A)$ into a `bulk' plus an `edge' part. 
This requires a bit of notation. Note that any subset $X$ of $\Lambda$
with horizontal diameter smaller than $L/2$ can be identified (non uniquely, of course) with a subset of $\mathbb Z^2$ with the same diameter and `shape' as $X$; we call $X_\infty\subset \mathbb Z^2$ one of these arbitrarily chosen representatives\footnote{
For instance, 
given $X=\{z_1,\ldots, z_n\}$, recalling that $(z_i)_1\in\{1,2,\ldots, L\}$ and $(z_i)_2\in\{1,2,\ldots, M\}$, we can let $X_\infty=\{y_1,\ldots,y_n\}$ be the set of points such that: (1) 
the vertical coordinates are the same as those of $\bs z$, i.e., $(y_i)_2=(z_i)_2$, $\forall i=1,\ldots, n$; (2)
the horizontal coordinate of $y_1$ is the same as $z_1$, i.e., $(y_1)_1=(z_1)_1$; (3) all the other horizontal coordinates are the same modulo $L$, i.e., $(y_i)_1=(z_i)_1$ mod $L$, $\forall i=2,\ldots, n$; (4) 
the specific values of $(y_i)_2$ for $i\ge 2$ are chosen in such a way that the horizontal distances between the corresponding pairs in $X$ and $X_\infty$ 
are the same, if measured on the cylinder $\Lambda$ or on $\mathbb Z^2$, respectively.}
of $X$, and we shall use an analogous convention for the subsets of $\fB_\Lambda$ with horizontal diameter smaller than $L/2$. 

\begin{lemma} 
	\label{lem:WB+WE}
	Under the same assumptions of Proposition \ref{thm:Ising_to_Grassman}, the kernel $W^{\rm int}_\Lambda$ of $\cV^{\rm int}(\Phi, \bs A)$ can be decomposed as
	\begin{eqnarray} W^{\rm int}_\Lambda(X,Y)&=&W_\infty^{\rm int}(X_\infty,Y_\infty)\, \mathds 1(\diam_1(X\cup Y)\le L/3)+W_\E^{\rm int}(X,Y) \label{WB+WE}\\ &=:& W_\B^{\rm int}(X,Y)+W_\E^{\rm int}(X,Y),\nonumber \end{eqnarray}
where: $\diam_1$ is the horizontal diameter on the cylinder $\Lambda$; $X_\infty, Y_\infty
\subset \mathfrak B:=\mathfrak B_{\mathbb Z^2}$ are two representatives of $X,Y$, respectively, such that $X_\infty\cup Y_\infty$ is a representative of $X\cup Y$, in the sense defined before the statement of the lemma; 
$W_\infty^{\rm int}$ is a function, independent of $L,M$, invariant under translations and under reflections about either coordinate axis, which satisfies the same 
	weighted $L^1$ bound \eqref{eq:B_base_decay} as $W_\Lambda^{\rm int}$. Moreover, for any $n\in\mathbb N$ and $m\in\mathbb N_0$, $W_\E^{\rm int}$ satisfies 
		\begin{equation}
			\frac1{L}\sum_{\substack{X,Y\subset \fB_\Lambda\\ |X|=n,\, |Y|=m}}|W_{\E}^{\rm int}(X,Y)| e^{c\delta_\E(X\cup Y)}\le C^{m+n} |\lambda|^{\max (1,\kappa(m+n))}, 
			\label{eq:WE_base_decay}
		\end{equation}
with the same $C,c,\kappa$ as in Proposition \ref{thm:Ising_to_Grassman}, where $\delta_\E(X)$ is the cardinality of the smallest connected subset of $\fB_\Lambda$ which 
includes $X$ and either touches the boundary of the cylinder\footnote{We say that $X\subset \fB_\Lambda$ `touches the boundary of the cylinder $\Lambda$', if
at least one of the edges in $X$ has an endpoint whose vertical coordinate is either equal to $1$ or to $M$.}, or its horizontal diameter is larger than $L/3$. \label{it:WE_base}
\end{lemma}

\begin{proof}
In order to obtain the decomposition \eqref{WB+WE}, let 
\begin{equation}\label{eq:WZ2} W_{\infty}^{\rm int}(X_\infty,Y_\infty):=\lim_{L,M\to\infty}W_\Lambda^{\rm int}(X_\infty+z_{L,M},Y_\infty+z_{L,M}),\end{equation} 
where $z_{L,M}=(L/2,\lfloor M/2\rfloor)$ and $X_\infty+z_{L,M}$ is the translate of $X_\infty$ by $z_{L,M}$;
note that this limit is well defined thanks to the 
fact that $W_\Lambda^{\rm int}$ can be expressed in terms of a sum like \cite[Eq.(2.21)]{GGM}, which is exponentially convergent, see \cite[Eq.(2.25)]{GGM}. The kernel 
$W^{\rm int}_\infty$ satisfies the analogue of \eqref{sa2.20}, that is 
\begin{equation}\label{dumbo}\sum_{\Gamma \in C_\infty}
		\varphi^T (\Gamma)
		\big(\prod_{\gamma \in \Gamma}
		\zeta_G (\gamma)-\prod_{\gamma \in \Gamma}\zeta_G(\gamma)\big|_{\lambda=0}\big)
		=	\sum_{X,Y \subset \fB}W_\infty^{\rm int}(X,Y)\prod_{x \in X} E_x\prod_{x \in Y} A_x,\end{equation}
 with $C_\infty$ the set of multipolygons on $\bZ^2$, and the activity $\zeta_G(\gamma)$ the same as the one in \eqref{sa2.20}, provided that $\gamma$ is considered 
 now as a polygon in $\bZ^2$, rather than in $\Lambda$ (note that such identification is possible as long as $\gamma$ does not wrap around the cylinder). 
 Moreover, $W^{\rm int}_\infty$ is translation invariant, and, letting
	\begin{equation}
		W^{\rm int}_\E (X,Y) := W^{\rm int}_\Lambda (X,Y) - \ind(\diam_1(X \cup Y)\le L/3) W^{\rm int}_\infty (X_\infty,Y_\infty),
	\end{equation}
the contribution to the first term in the right side of all multipolygons in $C_\Lambda$ with horizontal diameter $\le L/3$ cancels with their counterparts in $C_\infty$ from 
the second term in the right side. Each of the remaining multipolygons either comes from the first term in the right side and involves a multipolygon in $C_\Lambda$ 
whose support has horizontal diameter larger than $L/3$, or comes from the second term in the right side and involves a multipolygon in $C_\infty$ whose support contains 
a set $Z_\infty$ which is the representative (in the sense explained before the statement of the lemma) of a connected subset $Z_\Lambda$ of $\fB_\Lambda$ that contains 
$X \cup Y$ and touches the boundary of $\Lambda$; in either case the number of edges in the support of such a multipolygon is at least $\delta_\E(X \cup Y)$, from 
which the bound \eqref{eq:WE_base_decay} follows.
\end{proof}

In the following, we will wish to work in the $\phi,\xi$ variables introduces in \cref{section_vertical_diagonalization}; applying the change of variables \eqref{eq:sc_2}, with some abuse of notation we rewrite $\cV(\Phi,\bs A)$ in \eqref{eq:B_expansion_bis} as:
\begin{equation}
		\begin{split}
			\cV(\phi,\xi,\bs A)
			=&
			\sum_{x \in \fB_\Lambda}
			\sum_{(\bs \omega, \bs z) \in  \cO^2 \times \Lambda^2}
			B^\free_\Lambda((\bs \omega, \bs z), x)
			\phi(\bs \omega, \bs z)
			A_x
			\\ & \quad
			+
			\sum_{\substack{n\in2\mathbb N\\ m\in\mathbb N_0}}
			\sum_{(\bs \omega, \bs z) \in  \cO^n \times \Lambda^n}
			\sum_{\bs x\in \fB_\Lambda^m} 
			W^{\rm int}_\Lambda((\bs \omega,\bs z),\bs x) \phi(\bs\omega,\bs z) \bs A(\bs x)
			\\ =:& \ 
			\cB^{\free}(\phi,\xi,\bs A)+\cV^{\rm int}(\phi,\xi,\bs A),
		\end{split}
	\label{expanV1}
\end{equation}
where $\mathbb N$ and $\mathbb N_0$ are the sets of positive ad non-negative integers, respectively,
$\cO := \left\{ 1,-1,i,-i \right\}$, and, for $\bs \omega \in \cO^n$, $\bs z \in \Lambda$, we denote 
\begin{equation}
	\phi(\bs \omega, \bs z)
	=
	\prod_{j=1}^n
	\phi_{\omega_j, z_j}
\end{equation}
with $\phi_{\pm i,z} = \xi_{\pm, z}$, and similarly $\bs A(\bs x) := \prod_{j=1}^m A_{x_j}$ (for $\bs x=\emptyset$, we interperet $\bs A(\emptyset)=1$). 
The decay properties of $s_{\pm}$ noted after \cref{spmpoisson} together with  \cref{eq:B_base_decay} imply a similar decay property for the new coefficients:
\begin{equation}
	\sup_{\bs \omega \in \cO^2}
	\sup_{x \in \fB_\Lambda}
	\sum_{\bs z \in \Lambda^2}
	e^{c  \delta(\bs z, x)}
	\big|B_\Lambda^{{\free}}( (\bs\omega,\bs z),x)\big|\le C, 
	\label{eq:W_free_bound}
\end{equation}
and
\begin{equation}
\begin{split} & 	\sup_{\bs\omega\in\cO^n}\sup_{z_1\in\Lambda}\sum_{z_2,\ldots, z_{n} \in \Lambda}
	e^{c  \delta(\bs z)}
	\big|W_\Lambda^{{\rm int}}( (\bs\omega,\bs z),\emptyset)\big|\le C^{n}|\lambda|^{\max (1, \kappa n)}\\
	&
	\sup_{\bs\omega\in\cO^n}\sup_{x_1\in\fB_\Lambda}\sum_{x_2,\ldots, x_{m} \in \fB_\Lambda}\sum_{\bs z\in\Lambda^n}	
	e^{c  \delta(\bs z,\bs x)}
	\big|W_\Lambda^{{\rm int}}( (\bs\omega,\bs z),\bs x)\big|\le C^{n+m}|\lambda|^{\max (1, \kappa(n+m))}.
	\end{split}\label{eq:WL1bound}
\end{equation}
Note that, with this rewriting in terms of the $\phi,\xi$ variables, recalling  that $\cS_{t_1,t_2}(\Phi)=\cS_m(\xi)+\cS_c(\phi)$, see \eqref{eq:2.1.12}-\eqref{eq:2.1.13}, 
and denoting $P_c(\cD \phi):=\cD \phi\, e^{\cS_c(\phi)}/\Pf(A_c)$, $P_m(\cD \xi):=\cD \xi\, e^{\cS_m(\xi)}/\Pf(A_m)$ (here $A_c$ and $A_m$ are the two
$2|\Lambda|\times 2|\Lambda|$ anti-symmetric matrices associated with the Grassmann quadratic forms $\cS_c(\phi)$ and $\cS_m(\xi)$, respectively),
the Grassmann generating function $\widetilde \Xi_\Lambda(\bs A)$ in \eqref{eq:Ising2Grassman} can be rewritten as
\begin{equation}\label{3.16} \wt \Xi_\Lambda(\bs A)\propto e^{\cW(\bs A)}\int P_c(\cD\phi)\int P_m(\cD\xi) e^{\cV(\phi,\xi,\bs A)},\end{equation}
where $\propto$ means `up to a multiplicative constant independent of $\bs A$'. 
In view of these rewritings, Proposition \ref{thm:Ising_to_Grassman} implies  
\cite[Proposition \xref{prop:repr}]{AGG_part1} as an immediate corollary. 

\medskip

Of course, the bulk-edge decomposition of Lemma \ref{it:WE_base} implies an analogous decomposition for the kernel of $\cV^{\rm int}(\phi,\xi,\bs A)$, which reads as follows: 
\begin{equation} 
	\begin{split}
		W_\Lambda^{\rm int}((\bs\omega,\bs z),\bs x) 
		&= \, 
		(-1)^{\alpha(\bs z)}\,\mathds 1(\diam_1(\bs z,\bs x)\le L/3)
		\,W_\infty^{\rm int}((\bs\omega,\bs z_\infty),\bs x_\infty)
		+W_\E^{\rm int}((\bs\omega,\bs z),\bs x)
		\\ 
		& =:
		W_\B^{\rm int}((\bs\omega,\bs z),\bs x)+W_\E^{\rm int}((\bs\omega,\bs z),\bs x),
	\end{split}
	\label{WB+WE.2}
\end{equation}
where,  for any $\bs z$ with $\diam_1(\bs z)\le L/3$, 
\begin{equation}
	\alpha(\bs z)= 	\begin{cases} \#\{z_i\in\bs z\ : \ (z_i)_1\le L/3\}, & \text{if}\quad \max_{z_i,z_j\in\bs z}\{(z_i)_1-(z_j)_1\}\ge 2L/3,\\
	0 \,, & \text{otherwise.} \end{cases}\label{eq:alphaz}
\end{equation}
and 
\begin{equation}\label{eq:WZ2.2} W_{\infty}^{\rm int}((\bs\omega,\bs z),\bs x):=\lim_{L,M\to\infty}W_\Lambda^{\rm int}((\bs \omega,\bs z+z_{L,M}),\bs x+z_{L,M})
.\end{equation}
The factor $(-1)^{\alpha(\bs z)}$ in front of the first term in the right side of \eqref{WB+WE.2}, 
in light of the antiperiodicity of the $\phi,\xi$ fields, guarantees that $W_{\infty}^{\rm int}$ is translation invariant (in both coordinate directions), 
and that both $W_\B^{\rm int}$ and $W_\E^{\rm int}$ are invariant under simultaneous translations of $\bs z$ and $\bs x$ in the horizontal direction, with anti-periodic and periodic boundary conditions in $\bs z$ and $\bs x$, respectively. In terms of this new notation, Eq.\eqref{eq:WE_base_decay} implies that, for any $n\in\mathbb N$ and $m\in\mathbb N_0$, 
		\begin{equation}
			\frac1{L}\sup_{\bs \omega\in\mathcal O^n}\sum_{\substack{\bs z\in\Lambda^n\\ \bs x\in\fB_\Lambda^m}} |W_{\E}^{\rm int}((\bs\omega,\bs z),\bs x)| e^{c\delta_\E(\bs z,\bs x)}\le C^{m+n} |\lambda|^{\max (1,\kappa(m+n))}, 
			\label{eq:WE_base_decayrenzy}
		\end{equation}
with $\delta_\E(\bs z,\bs x)$ is the `edge' tree distance of $(\bs z,\bs x)$, i.e., the cardinality of the smallest connected subset of $\fB_\Lambda$ that includes $\bs x$, touches the 
points of $\bs z$ and either touches the boundary of the cylinder or it has horizontal diameter larger than $L/3$. Of course, $B^\free_\Lambda$ admits a similar bulk-edge decomposition: in analogy with \eqref{WB+WE}, $B^\free_\Lambda=B^\free_\B+B^\free_\E$ with 
		\begin{equation}
			\frac1{L}\sup_{\bs \omega\in\mathcal O^2}\sum_{\substack{\bs z\in\Lambda^2\\ x\in\fB_\Lambda}} |B_{\E}^{\rm free}(\bs\omega,\bs z,x)| e^{c\delta_\E(\bs z,x)}\le C.			\label{eq:WE_base_decayletta}
		\end{equation}
Before concluding this section, let us comment on the connection between \eqref{3.16} and \eqref{eq:ZA_goal}. Fix once and for all a neighbourhood $U\subset \mathbb R$ of $1$ not containing $0$; say, for definiteness, $U:=\{z\in\mathbb R: |z-1|\le 1/2\}$. For any $Z\in U$ and $t_1^*\in K'$, 
we let $t_2^* := (1- t_1^*)/(1+t_1^*)$ and let $\cS^*_c(\phi)=\frac12(\phi,A_c^*\phi)$ (resp.\ $\cS_m^*(\xi)=\frac12(\xi,A^*_m\xi)$) be obtained from $\cS_c$ (resp.\ $\cS_m$) by 
replacing $t_1,t_2$ with $t_1^*,t_2^*$ in \cref{eq:2.1.13} (resp.\ \cref{eq:2.1.12}). We also let $P_c^*(\cD \phi):= \cD \phi e^{\cS_c^*(\phi)}/\Pf(A^*_c)$, 
$P_m^* (\cD \xi):=\cD \xi e^{\cS_m^* (\xi)}/\Pf(A^*_m)$. Given these definitions, in \eqref{3.16} we 
first rescale the $\phi$ and $\xi$ variables by $Z^{-1/2}$, then multiply and divide the Grassmann integrand by $e^{\cS^*_c(\phi)+\cS^*_m(\xi)}$, thus 
getting \begin{equation}
	\wt \Xi (\bs A)
	\propto
	e^{\cW(\bs A)}
	\int P_c^* (\cD \phi) P_m^* (\cD\xi) 
	e^{\cV^{(1)}(\phi,\xi,\bs A)}
	\label{eq:Xi_def}
\end{equation}
with \begin{equation}
	\cV^{(1)}(\phi,\xi,\bs A)
	\begin{aligned}[t]
		& :=
		Z^{-1} \cS_c (\phi)
		-
		\cS_c^* (\phi)
		+
		Z^{-1} \cS_m (\xi)
		-
		\cS_m^* (\xi)
		+
		\cV (Z^{-1/2} \phi, Z^{-1/2}\xi, \bs A)
		\\
		& =:
	 \calN_c (\phi)
		+
		\calN_m (\xi)
		+
		\cV (Z^{-1/2} \phi, Z^{-1/2}\xi, \bs A)
		.
	\end{aligned}
	\label{eq:cV_1_def}
\end{equation}
This proves \eqref{eq:ZA_goal} and puts us in the position of setting the multiscale computation of the sequence of effective potentials, 
whose infinite plane counterparts are constructed and bounded in the next section. For later reference, we note that, in light of \eqref{eq:2.1.12}, \eqref{eq:2.1.13}, 
\eqref{expanV1}, $\cV^{(1)}(\phi,\xi,\bs A)$ can be written as:
\begin{equation}\cV^{(1)}(\phi,\xi,\bs A)= \sum_{\substack{n\in2\mathbb N,\\ m\in\mathbb N_0}}\sum_{(\bs \omega,\bs z)\in\mathcal O^n\times\Lambda^n}\sum_{\bs x\in\fB_\Lambda^m} W^{(1)}_\Lambda((\bs\omega,\bs z),\bs x)\phi(\bs\omega,\bs z)\bs A(\bs x),\label{AZen}\end{equation}
for an appropriate kernel, which inherits its properties from those of $\cS_c$, $\cS_m$ and $\cV$. With no loss of generality, we 
can assume that $W^{(1)}_\Lambda$ is 
anti-symmetric under simultaneous permutations of $\bs \omega$ and $\bs z$, 
symmetric under permutations of $\bs x$, invariant under simultaneous translations of $\bs z$ and $\bs x$ in the horizontal direction (with anti-periodic and periodic boundary conditions in $\bs z$ and $\bs x$, 
respectively), invariant under the reflection symmetries induced by the transformations $A_x\to A_{\theta_l x}$ and $\phi_{\omega,z}\to \Theta_l\phi_{\omega,z}$, see 
\eqref{eq:symm_horizontal_phixi}-\eqref{eq:symm_vertical_phixi}. 
From now on, with some abuse of notation, given $\bs\omega=(\omega_1,\ldots,\omega_n)\in\mathcal O^n$ and $\bs z=(z_1,\ldots, z_n)\in\Lambda^n$, we shall identify the pair $(\bs\omega,\bs z)$ with the $n$-ple $((\omega_1,z_1),\ldots,(\omega_n,z_n))\in(\mathcal O\times\Lambda)^n$. 

\section{The renormalized expansion in the full-plane limit}\label{sec:renexp}

In this section we construct the sequence of effective potentials (see the last part of Section \ref{sec:intro}) in the infinite volume limit, and derive weighted $L^1$ bounds 
for their kernels, in a form appropriate for the subsequent generalization to the finite cylinder, 
discussed in \cite[Section \xref{sec:renexp}]{AGG_part1}. The construction of this section will allow us, in particular, to fix the free parameters $\beta,Z,t_1^*$, which the 
Grassmann integral in the right side of \eqref{eq:Xi_def} depends on, in such a way that the sequence of running coupling constants goes to zero exponentially fast in the 
infrared limit; see Section \ref{sec:fixed_point} below. 

As anticipated at the end of Section \ref{sec:intro}, here we limit ourselves to construct the sequence of effective potentials at $\bs A=\bs 0$, so, for lightness of notation, we 
denote by $\cV^{(h)}(\phi):=\cV^{(h)}(\phi,\bs 0)$ the effective potentials with $h\le 0$ at zero external fields (similarly, we let $\cV^{(1)}(\phi,\xi)=\cV^{(1)}(\phi,\xi,\bs A)$). 
In light of \eqref{eq:Vcyl_N} and \eqref{eq:Vcyl_h}, these effective potentials are iteratively defined via
\begin{equation}
\cV^{(h-1)}(\phi)={\rm const.}+\begin{cases} \log \int P^*_m(\cD\xi) \ e^{\cV^{(1)}(\phi,\xi)} & \text{if $h=1$,}\\
\log \int P^{(h)}(\cD\varphi)\ e^{\cL\cV^{(h)}(\phi+\varphi)+\cR\cV^{(h)}(\phi+\varphi)} & \text{if $h\le 0$,}\end{cases}\label{itT}\end{equation}
where the const. is fixed so that $\cV^{(h)}(0)=0$, for all $h\le 0$, and $\cL\cV^{(h)}+\cR\cV^{(h)}$ is an equivalent 
rewriting of $\cV^{(h)}$, to be defined (in the full plane limit) below. Expanding the exponential and the 
logarithm in the right side of \eqref{itT} allows us to rewrite
\begin{equation}\cV^{(0)}(\phi)={\rm const.}+\sum_{s\ge 1}\frac1{s!} \mathbb E^*_m\big(\underbrace{\cV^{(1)}(\phi,\cdot);\cdots;\cV^{(1)}(\phi,\cdot)}_\text{$s$ times}\big)
\label{itTt1}\end{equation}
and, for $h\le 0$, 
\begin{equation} \cV^{(h-1)}(\phi)={\rm const.}+\sum_{s\ge 1}\frac1{s!} 
\mathbb E^{(h)}\big(\underbrace{\cL\cV^{(h)}(\phi+\cdot)+\cR\cV^{(h)}(\phi+\cdot);\cdots;\cL\cV^{(h)}(\phi+\cdot)+
\cR\cV^{(h)}(\phi+\cdot)}_\text{$s$ times}\big),\label{itTt}\end{equation}
where $\mathbb E^*_m$ (resp. $\mathbb E^{(h)}$) denotes the {\it truncated expectation} \cite[Eq.(4.13)]{GM01} with respect to the Grassmann Gaussian integration $P^*_m$
(resp. $P^{(h)}$). Expanding the effective potentials in terms of their kernels, in analogy with \eqref{AZen}, Eqs.\eqref{itTt1}-\eqref{itTt} allow us to iteratively compute the 
kernels of 
$\cV^{(h)}$, for all $h\le 0$. For instance, at the first step, using \eqref{AZen}, \eqref{itTt1} and the BBFK formula (for Battle, Brydges, Federbush, Kennedy) 
for the Grassmann truncated expectations \cite{B,BrF,AR98,BK87}, we find that, denoting by $V^{(1)}_\Lambda(\Psi)=
W^{(1)}_\Lambda(\Psi,\emptyset)$ with $\Psi=((\omega_1,z_1),\ldots,(\omega_n,z_n))\in \cup_{n\in2\mathbb N}(\cO\times\Lambda)^n=:\cM_{1,\Lambda}$ 
the kernel of $\cV^{(1)}(\phi,\xi)$, the kernel $V^{(0)}_\Lambda$ of $\cV^{(0)}(\phi)$ satisfies, for any 
$\Psi=((\omega_1,z_1),\ldots,(\omega_n,z_n))\in\cup_{n\in2\mathbb N}(\{+,-\}\times\Lambda)^n$, 
\begin{equation}V^{(0)}_\Lambda(\Psi) =\sum_{s=1}^{\infty}\frac{1}{s!}
	\sum_{\Psi_1,\ldots,\Psi_s\in \cM_{1,\Lambda}}^{(\Psi)}\
	\sum_{T \in \cS(\bar\Psi_1,\ldots,\bar\Psi_s)}
	\fG_{T}^{(1)}(\bar\Psi_1,\ldots,\bar\Psi_s)\,
 \alpha(\Psi;\Psi_1,\ldots,\Psi_s)\Big(  \prod_{j=1}^s V^{(1)}_\Lambda(\Psi_j)\Big),\label{eq:BBF0cyl}\end{equation}
where: 
\begin{itemize}
\item the symbol $(\Psi)$ on the second sum means that the sum runs over all ways of representing $\Psi$ as an ordered sum of $s$ (possibly empty) tuples, 
$\Psi'_1+\cdots\Psi'_s=\Psi$, and over all tuples $\cM_{1,\Lambda}\ni\Psi_j\supseteq\Psi'_j$; for each such term in the second sum, we denote by $\bar\Psi_j:=\Psi_j\setminus\Psi'_j$ and by $\alpha(\Psi;\Psi_1,\ldots,\Psi_s)$ the sign of the permutation from $\Psi_1\oplus\cdots\oplus \Psi_s$ to $\Psi\oplus\bar\Psi_1\oplus\cdots\oplus\bar\Psi_s$ (here $\oplus$ indicates concatenation of ordered tuples); 
\item $\cS(\bar\Psi_1,\ldots,\bar\Psi_s)$ denotes the set of all the `spanning trees' on $\bar\Psi_1,\ldots,\bar\Psi_s$, that is, 
of all the sets $T$ of ordered pairs $(f,f')$, with $f \in \bar\Psi_i$, $f' \in \bar\Psi_j$ and $i < j$, whose corresponding graph $G_T=(V,E_T)$, 
with vertex set $V=\{1,\ldots,s\}$ and edge set $E_T=\{(i,j)\in V^2\,:\, \exists (f,f')\in T\ \text{with}\  f\in Q_i, f'\in Q_j\}$, is a tree graph (for $s=1$, we let $\cS(\bar\Psi_1)\equiv\{\emptyset\}$);
\item 
$\fG_{T}^{(1)}(\bar\Psi_1,\ldots,\bar\Psi_s)$ is different from zero only if $\bar\Psi_j\in\cup_{2\mathbb N_0}(\{+i,-i\}\times\Lambda)^n$ for all $j=1,\ldots,s$, and, if $s>1$, only if $\bar\Psi_j\neq\emptyset$ for all $j=1,\ldots,s$; more precisely: if $s=1$ and $\bar\Psi_1=\emptyset$, then 
$\fG_\emptyset^{(1)}(\emptyset)=1$; if $s=1$ and $\bar\Psi_1\neq\emptyset$, then 
$\fG_\emptyset^{(1)}(\bar\Psi_1)={\rm Pf}\big(G^{(1)}_{\bar\Psi_1}\big)$,
where, given a pair $\ell=((\omega,z),(\omega',z'))$ of distinct elements of $\bar\Psi_1$, $\big(G^{(1)}_{\bar\Psi_1}\big)_{\ell}=g^{(1)}_\ell:=
g^{(1)}_{-i\omega,-i\omega'}(z,z')$\footnote{Here $g^{(1)}_{\sigma,\sigma'}(z,z')$ with $\sigma,\sigma'\in\{+,-\}$ are the elements of the $2\times2$ matrix $\fg^{(1)}(z,z')
\equiv\fg^*_m(z,z')$, where $\fg^*_m(z,z')$ the same as in \eqref{eq:propmassive} with $t_1$ replaced by $t_1^*$.}; if $s>1$ and $\bar\Psi_j\neq\emptyset$ for all $j=1,\ldots,s$,
then 
\begin{equation}
\label{etcetcbase}
\fG_T^{(1)}(\bar\Psi_1,\dots,\bar\Psi_s):=	\alpha_T(\bar\Psi_1,\dots,\bar\Psi_s)\left[ \prod_{\ell \in T} g^{(1)}_\ell \right]
		\int P_{\bar\Psi_1,\dots,\bar\Psi_s,T} (\dd \bs t)\, {\rm{Pf}} \big(G^{(1)}_{\bar\Psi_1,\ldots,\bar\Psi_s,T} (\bs t)\big),\end{equation}
where: 
\begin{itemize}
\item $\alpha_T(\bar\Psi_1,\dots,\bar\Psi_n)$ is the sign of the permutation from $\bar\Psi_1\oplus\cdots\oplus \bar\Psi_s$ to $T\oplus(\bar\Psi_1\setminus T)\oplus\cdots\oplus(\bar\Psi_s\setminus T)$;
\item ${\bs t}=\{t_{i,j}\}_{1\le i,j \le s}$, and $P_{\bar\Psi_1,\ldots,\bar\Psi_s,T}(\dd \bs t)$
is a probability measure with support on a set of ${\bs t}$ such that
$t_{i,j}=\bs u_i\cdot\bs u_{j}$ for some family of vectors $\bs u_i=\bs u_i({\bs t})\in \mathbb R^s$ of
unit norm;
\item letting $2q=\sum_{i=1}^s|\bar\Psi_i|$, $G^{(1)}_{\bar\Psi_1,\ldots,\bar\Psi_s,T}({\bs t})$ is an antisymmetric
  $(2q-2s+2)\times (2q-2s+2)$ matrix, whose off-diagonal elements are given by
  $\big(G^{(1)}_{\bar\Psi_1,\ldots,\bar\Psi_s,T} (\bs t)\big)_{f,f'}=t_{i(f),i(f')}g^{(1)}_{\ell(f,f')}$, where $f,
  f'$ are elements of the tuple $(\bar\Psi_1\setminus T)\oplus\cdots\oplus(\bar\Psi_s\setminus T)$, and $i(f)$ is the integer in $\{1,\ldots,s\}$ such that 
  $f$ is an element of $\bar\Psi_i\setminus T$.
\end{itemize}
\end{itemize}
Recalling that $g_\ell^{(1)}$ and $V_\Lambda^{(1)}$ admit infinite volume limits $g_{\ell,\infty}^{(1)}$ and $V_\infty^{(1)}$, respectively, 
in the sense of \eqref{ginftyeta} and \eqref{eq:WZ2.2}, from \eqref{eq:BBF0cyl} it follows 
that $V_\Lambda^{(0)}$ admits an infinite volume limit as well, equal to the `obvious' analogue of the right side of \eqref{eq:BBF0cyl}, namely, the expression obtained from that 
one by replacing: $\cM_{1,\Lambda}$ by $\cM_{1,\infty}:=\cup_{n\in2\mathbb N}(\cO\times\mathbb Z^2)^n$; $\fG_T^{(1)}$ by $\fG_{T,\infty}^{(1)}$ (the latter being defined 
analogously to the former, with $g^{(1)}_\ell$ replaced by $g^{(1)}_{\ell,\infty}$ in all the involved expressions); and $V_\Lambda^{(1)}$ by $V_\infty^{(1)}$. 

Proceeding inductively in $h\le 0$, one finds that \eqref{itTt} implies a representation of the kernel $V_\Lambda^{(h)}$ of $\cV^{(h)}$ analogous to \eqref{eq:BBF0cyl}. 
Also in that case, 
the resulting formula for $V_\Lambda^{(h)}$ admits a natural infinite volume limit. In this way, we obtain a recursive equation for the infinite plane kernels, denoted
$V_\infty^{(h)}$, whose definition and solution is described below. Convergence of the finite volume kernels 
$V_\Lambda^{(h)}$ to their infinite volume counterparts, with optimal bounds on the norm of the finite size corrections, 
is deferred to \cite[Section \xref{sec:renexp}]{AGG_part1}. A key point in the derivation of bounds on the kernels that are uniform in the scale label $h$ is the definition 
of an appropriate action of the $\cL$ and $\cR$ operators, as well as of their infinite volume counterparts, $\cL_\infty$ and $\cR_\infty$. As anticipated above, these operators
allow us to isolate the potentially divergent part of the kernels, $\cL\cV^{(h)}$ (the `local' contributions, parametrized at any given scale, by a finite number of `running coupling constants')
from a remainder $\cR\cV^{(h)}$, which is `dimensionally better behaved' than $\cL\cV^{(h)}$; in order for the remainder to be shown to satisfy `improved dimensional bounds', it is necessary to rewrite it in an appropriate, interpolated, form, involving the action 
of discrete derivatives on the Grassmann fields. 

The plan of the incoming subsections is the following: in Sect.\ref{sec:repeq} we describe the representation of the effective potentials in the infinite volume limit
and introduce the notion 
of equivalent kernels; in Sect.\ref{sec:interpolation} we define the operators $\cL_\infty$ and $\cR_\infty$; in Sect.\ref{sec:tree_defs} we derive the solution to the 
recursive equations for the infinite volume kernels in terms of a tree expansion; in Sect.\ref{sec:formal_bounds}, we use such a tree expansion to derive weighted $L^1$ bounds
on the kernels; importantly, these bounds depend upon the sequence of running coupling constants, and they imply analyticity of the kernels provided such a sequence is uniformly
bounded in the scale label; in Sect.\ref{sec:fixed_point}, as a corollary of the weighted $L^1$ bounds of the previous subsection, 
we prove a fixed point theorem, which allows us to fix the free parameters $Z,\beta,t_1^*$ in such 
a way that the flow of the running coupling constants is, in fact, uniformly bounded in $h$, as desired: even more, the running coupling constants go to zero exponentially 
fast as $h\to-\infty$, a consequence of the irrelevance of the quartic effective interaction in the theory at hand. 

\subsection{Effective potentials and kernels: representation and equivalence}\label{sec:repeq}

In this subsection we define the effective potential in the full plane in terms of equivalence classes of {\it kernels} $V(\Psi)$, namely, of real-valued functions playing the same role
as the coefficient functions $V_\infty^{(1)}(\Psi)$ and $V^{(0)}_\infty(\Psi)$ introduced above. This points of view avoids defining an infinite dimensional Grassmann algebra. 
The equivalence relation among kernels, to be introduced momentarily, 
generalizes the relationships which hold between different ways of writing the coefficients of a given Grassmann polynomial. 

As mentioned above, in order to obtain bounds on the kernels of the effective potentials which are uniform in the scale label, 
we will need to group some of the Grassmann fields into discrete derivatives; we will mainly use the directional derivative $\partial_j \phi_{\omega, z} := \phi_{\omega, z+\hat{e}_j} - \phi_{\omega, z}$ (note that this is the same convention used in Section~\ref{sec:2.2.2}).  We consequently consider kernels which specify when and how this is done, and in particular define the equivalence relationship with this in mind.  

\medskip

Let $\Lambda_\infty :=\bZ^2$, let $\fB$ denote the set of nearest neighbor edges of $\Lambda_\infty$, and let $\cD:=\{D\in\{0,1,2\}^2: \|D\|_1\le 2\}$. 
Let $\cM_\infty=\cup_{n\in2\mathbb N}(\{+,-\}\times\cD\times \Lambda_\infty)^n$ be the set of {\it field multilabels}. 
for some $n\in 2\mathbb N$, such that $\|D_i\|_1\le 2$. 
We can think of any $\Psi=((\omega_1,D_1,z_1),\ldots,(\omega_n,D_n,z_n))\in \cM_\infty$ 
as indexing a formal Grassmann monomial $\phi(\Psi)$ given by
\begin{equation}
\phi(\Psi)=\partial^{D_1}\phi_{\omega_1,z_1}\cdots \partial^{D_n}\phi_{\omega_n,z_n},
\end{equation}
where, denoting $D_i=((D_i)_1,(D_i)_2)\in\cD$, we let 
$$\partial^{D_i}\phi_{\omega_i,z_i}:=\partial_1^{(D_i)_1}\partial_2^{(D_i)_2}\phi_{\omega_i,z_i},$$
with $\partial_1$ and $\partial_2$ the (right) discrete derivatives introduced above.
In the following, with some abuse of notation, any element 
$\Psi\in{\cM_{\infty}}$ of length $|\Psi|=n$ will be denoted indistinctly by $\Psi=((\omega_1,D_1,z_1),\ldots,(\omega_n,D_n,z_n))$ 
or $\Psi=(\bs \omega, \bs D, \bs z)$, with the understanding that $\bs \omega=(\omega_1,\ldots,\omega_n)$, etc. 

We will call a function $V:\cM_\infty\to \bR$ a {\it kernel function}, let $V_n$ denote its restriction to field multilabels of length $n$, {and let 
$V_{n,p}$ be the restriction of $V_n$ to field multilabels with $\|\bs D\|_1=p$.} 
Thinking of such a $V$ as the coefficient function of a formal Grassmann polynomial
\begin{equation} 
\cV(\phi)=\sum_{\Psi\in \cM_\infty}V(\Psi)\phi(\Psi)
\label{eq:forma}
\end{equation}
suggests an equivalence relationship corresponding to manipulations allowed by the anticommutativity of the Grassmann 
variables and by the definition of discrete derivative. 

More precisely, we say that $V$ is equivalent to $V'$, and write $V\sim V'$, if either: 
\begin{enumerate}
	\item $V'$ is obtained from $V$ by permuting the arguments and changing the sign according to the parity of the permutation;
	\item\label{item2} $V'$ is obtained from $V$ by writing out the action of a derivative:
	 that is, there exist $n\in 2\mathbb N$,
	 $i\in\{1,\ldots, n\}$ and $j\in\{1,2\}$ such that, letting $\bs D^+_{i,j}=(D_1,\ldots,D_{i-1}, D_{i}+\hat e_{j}, D_{i+1}, \ldots, D_n)$ and 
	 $\bs z^-_{i,j}=(z_1,\ldots, z_{i-1}, z_i-\hat e_j,z_{i+1}, \ldots,z_n)$,
	  \begin{equation}
	    V'_{n,p}(\bs \omega,\bs D,\bs z)
	    =\begin{cases} 0  & \text{if} \quad (D_i)_j=2,\\
	    V_n(\bs \omega,\bs D^+_{i,j},\bs z^-_{i,j}) - V_n(\bs \omega, \bs D^+_{i,j},\bs z) & \text{if} \quad (D_i)_j=1,\\
	    V_n(\bs \omega,\bs D,\bs z)+V_n(\bs \omega,\bs D^+_{i,j},\bs z^-_{i,j}) - V_n(\bs \omega, \bs D^+_{i,j},\bs z) & \text{if} \quad (D_i)_j=0,\end{cases}
\end{equation}	  
	  while $V_{m}'=V_{m}$ for all $m\in2\mathbb N\setminus\{n\}$; 
	\item $V'$ is obtained from $V$ by adding an arbitrary kernel $V^*$ that is different from zero only for arguments with common repetition, that is, $V^*(\bs \omega, \bs D, \bs z)= 0$ unless there is some $i \neq j$ such that $(\omega_i,D_i,z_i)=(\omega_j,D_j,z_j)$;
\end{enumerate}
or $V'$ is obtained from $V$ by a countable sequence of such elementary operations and of convex combinations thereof. 
Moreover, we assume that the equivalence relation $\sim$ is preserved by linear combinations, i.e., if $V'_\alpha\sim V_\alpha$ for all $\alpha$ in the countable index set $\mathcal I$, then $\sum_{\alpha\in\mathcal I}V_\alpha'\sim \sum_{\alpha\in\mathcal I}V_\alpha$.  We will call the equivalence classes generated by $\sim$ {\it potentials}, and often specify them by formal sums like~\eqref{eq:forma}.

\begin{remark} The operation in item \ref{item2} can be thought of as a form of `integration by parts'. The kernels equivalent to zero, $V\sim 0$,
	correspond to what are known as `null fields' in the literature on conformal field theories.\label{rem:null_fields}
\end{remark}

\subsection{Localization and interpolation}\label{sec:interpolation}

In this section we define the operators $\cL_\infty$ and $\cR_\infty$ acting on kernels indexed by field multilabels in $\cM_\infty$, and show several estimates related to $\cR_\infty$.  We recall that, given a kernel $V$, the symbol $V_{n,p}$ denotes its restriction to field multilabels of length $n$, such that $\|\bs D\|_1=p$. 

\medskip

{\it The operator $\cL_\infty$.} 
First of all, we let 
\begin{equation} 
	\cL_\infty (V_{n,p}):=0, \quad \text{if}\quad 2-\frac{n}2-p<0.
\label{loc1}
\end{equation}
In the RG jargon, the combination $2-\frac{n}2-p$ is the {\it scaling dimension} of $V_{n,p}$, and will reappear below, for example in \cref{lm:W8:scaldim}; in this sense, \eqref{loc1} says that 
the local part of the terms with negative scaling dimension (the irrelevant terms, in the RG jargon) is set equal to zero. 

There are only a few cases for which $2-\frac{n}2-p\ge 0$, namely $(n,p)=(2,0), (2,1), (4,0)$. In these cases, the action of $\cL_\infty$ on 
$V_{n,p}$ is non trivial, and will be defined in terms 
of other basic operators, the first of which is $\tcL$, which is defined as: $\tcL(V_{n,p})=(\tcL V)_{n,p}\equiv \tcL V_{n,p}$, with 
\begin{equation}
	\widetilde \cL V_{n,p} (\bs\omega,\bs D,(z_1,\ldots, z_n))
	:=\left(\prod_{j=2}^n\delta_{z_j,z_1}\right)
	  \sum_{y_2,\ldots, y_n \in \Lambda_\infty}
	  V_{n,p}(\bs \omega,\bs D,(z_1,y_2,\ldots, y_n)).
	\label{eq:cLtilde_kernel_def}
\end{equation}
If $(n,p)=(2,1), (4,0)$, we let 
\begin{equation} \cL_\infty(V_{2,1}):=\cA (\tcL V_{2,1}), \qquad  \cL_\infty V_{4,0} :=\cA (\tcL V_{4,0}), \label{eq:local2140}\end{equation}
where  $\cA$ is the operator that antisymmetrizes with respect to permutations {\it and} symmetrizes with respect to reflections in the horizontal and vertical directions\footnote{The 
reflection transformations in the infinite plane act on the Grassmann fields in the same way as \eqref{eq:symm_master_phixi}, with the difference that 
the reflections $\theta_1 (x_1,x_2)$ and $\theta_2 (x_1,x_2)$ are replaced by their infinite-plane analogues, namely by $\tilde\theta_1(x_1,x_2)=(-x_1,x_2)$ and
$\tilde\theta_2(x_1,x_2)=(x_1,-x_2)$, respectively.}. 
A first important remark, related to the definitions \eqref{eq:local2140}, is that, if $V_{2,1}$ is invariant under translations and under the action of $\cA$, then 
\begin{equation} \cL_\infty(V_{2,1})= c_1F_{\zeta,\infty}+ c_2 F_{\eta,\infty}, \label{1strem} \end{equation}
for some real numbers $c_1,c_2$ and $F_{\zeta,\infty}, F_{\eta,\infty}$ the $\cA$-invariant kernels associated with the potentials 
\begin{equation} \cF_{\zeta,\infty}(\phi):=
\sum_{\omega=\pm}\sum_{z\in\Lambda_\infty} \omega\phi_{\omega,z}\dd_1\phi_{\omega,z}, \qquad 
\cF_{\eta,\infty}(\phi):=\sum_{\omega=\pm}\sum_{z\in\Lambda_\infty}\phi_{\omega,z}\dd_2\phi_{-\omega,z},\nonumber\\
\end{equation}
with $\dd_j$ the symmetric discrete derivatives, acting on the Grassmann fields as $\dd_j\phi_{\omega,z}=\frac12(\partial_j\phi_{\omega,z}+\partial_j\phi_{\omega,z-\hat e_j})$. 
A second and even more important remark is that, 
due to the fact that $\omega$ only assumes two values and that $\tcL V_{4,0}$ is supported on $\bs z$ such that $z_1=z_2=z_3=z_4$, one has 
\begin{equation}  \cL_\infty (V_{4,0})= 0,\label{pauli}\end{equation}
a cancellation that will play an important role in the following. 

In order to define the action of $\cL_\infty$ on $V_{2,0}$, we want to obtain a kernel function equivalent to $V_{2,0}- \tcL V_{2,0}$, denoted by $(\tcR V)_{2,1}$,
which is supported on arguments with an additional derivative. As we will see, the definition of $\tcR$ will also play a central role in the definition 
of the operator $\cR_\infty$ below. We rewrite
\begin{eqnarray}
&& \phantom{=} \sum_{\bs z\in\Lambda_\infty^2}[V_{2,0}(\bs \omega,\bs 0,\bs z)-\tcL V_{2,0}(\bs \omega,\bs 0,\bs z)]
\phi(\bs \omega,\bs 0,\bs z)\label{inte.r}\\
& &= \sum_{\bs z\in\Lambda_\infty^2}V_{2,0}(\bs \omega,\bs 0,\bs z)[\phi(\bs\omega,\bs 0,
\bs z)-\phi(\bs\omega,\bs 0,(z_1,z_1))]=\sum_{\bs z\in\Lambda_\infty^2}V_{2,0}(\bs \omega,\bs 0,\bs z)\phi_{\omega_1,z_1}(\phi_{\omega_2,z_2}-\phi_{\omega_2,z_1})
.\nonumber\end{eqnarray}
We now intend to write the difference $\phi_{\omega_1,z_1}(\phi_{\omega_2,z_2}-\phi_{\omega_2,z_1})$ as 
a sum of terms of the form $\phi_{\omega_1,z_1}\partial^{D'} \phi_{\omega_2,y}$, with $\|D'\|_1=1$, over the sites $y$ on a path from $z_1$ to $z_2$. 
To do this we must specify which path is to be used.

For each $z, z' \in \Lambda_\infty$, let $\gamma(z, z')=(z, z_1, z_2, \ldots, z_n,z')$ be the shortest path obtained by going first horizontally and then vertically from $z$ to $z'$.
Note that $\gamma$ is covariant under the symmetries of the model on the infinite plane, i.e., 
\begin{equation}
	S \gamma (z, z')
	=
	\gamma (S z, S z')
	\label{eq:gamma_symmetry}
\end{equation}
where $S: \Lambda_\infty \to \Lambda_\infty$ is some composition of translations and reflections parallel to the coordinate axes. Given $z,z'$ two distinct sites in $\Lambda_\infty$, let $\INT (z,z')$ be the set of 
$(\sigma,(D_1,D_2),(y_1,y_2))\equiv ( \sigma, \bs D, \bs y)\in\{\pm\}\times\{0,\hat e_1,\hat e_2\}^2\times  \Lambda_\infty^2$ such that: (1) $y_1=z$, (2) $D_1=0$, (3) $y_2, y_2 + D_2\in \gamma (z,z')$, (4)
$\sigma=+$ if $y_2$ precedes $y_2 + D_2$ in the sequence defining $\gamma (z,z')$, and $\sigma=-$ otherwise. In terms of this definition, one can easily check that \eqref{inte.r} can be rewritten as
\begin{equation}  
	\begin{split}
		\eqref{inte.r}=&\sum_{\bs z\in\Lambda_\infty^2}V_{2,0}(\bs \omega,\bs 0,\bs z)\sum_{(\sigma,\bs D,\bs y)\in\INT(\bs z)}\sigma\phi(\bs \omega,\bs D,\bs y)\\
		\equiv&\sum_{\bs y\in\Lambda_\infty^2}\sum_{\bs D}^{(1)}(\tcR V)_{2,1}(\bs \omega,\bs D,\bs y)\phi(\bs \omega,\bs D,\bs y), 
	\end{split}
\label{juve.n}
\end{equation}
where, if $\bs z=(z_1,z_1)$, the sum over $(\sigma,\bs D,\bs y)$ in the first line should be interpreted as being equal to zero (in this case, we let $\INT(z_1,z_1)$ be the empty set). In going from the first to the second line, we exchanged the order of summations over $\bs z$ and $\bs y$; moreover, $\sum_{\bs D}^{(p)}$ denotes the sum 
over the pairs $\bs D=(D_1,D_2)$ such that $\|\bs D\|_1=p$, and 
\begin{equation} (\tcR V)_{2,1}(\bs \omega,\bs D,\bs y):=\sum_{\substack{\sigma, \bs z:\\(\sigma,\bs D,\bs y)\in\INT(\bs z)}}\sigma V_{2,0}(\bs \omega,\bs 0,\bs z).\label{juve.m}\end{equation}
From the previous manipulations, it is clear that $(\tcR V)_{2,1}\sim V_{2,0}-\tcL V_{2,0}$. We are finally ready to define the action of $\cL_\infty$ on $V_{2,0}$: 
\begin{equation} \cL_\infty(V_{2,0}):=\cA(\tcL V_{2,0}+\tcL (\tcR V)_{2,1}).  \label{eq:local20}\end{equation}
Note  that, if $V_{2,0}$ is invariant under translations and under the action of $\cA$, then 
\begin{equation} 
	\cL_\infty(V_{2,0})
	=
	c_0 F_{\nu,\infty}+c_1F_{\zeta,\infty}+ c_2 F_{\eta,\infty},
	\label{2ndrem} 
\end{equation}
for some real numbers $c_0,c_1,c_2$, and $F_{\nu,\infty}$ the $\cA$-invariant kernel associated with the potential 
\begin{equation} \cF_{\nu,\infty}(\phi):=\frac12
\sum_{\omega=\pm}\sum_{z\in\Lambda_\infty}\omega\phi_{\omega,z}\phi_{-\omega,z},\end{equation} 
while we recall that $F_{\zeta,\infty}, F_{\eta,\infty}$ were defined right after \eqref{1strem}. 
Summarizing, in view of \cref{pauli}, 
\begin{equation}
	\label{defcL.1}
	(\cL_\infty V)_{n,p}=\begin{cases} \cA\,  (\tcL V_{2,0}) &\text{if} \ (n,p)=(2,0),\\
	\cA(\tcL V_{2,1}+\tcL(\tcR V)_{2,1}) & \text{if} \ (n,p)=(2,1),\\
	0 & \text{otherwise}.
	\end{cases}
\end{equation}

{\it The operator $\cR_\infty$.} We now want to define an operator $\cR_\infty$ such that $\cR_\infty V\sim V-\cL_\infty V$ for kernels $V$ that are invariant under translations and 
under the action of $\mathcal A$.  First of all, we let 
\begin{equation} \cR_\infty(V_{n,p})=(\cR_\infty V)_{n,p}:=V_{n,p},\quad \text{if:} \quad  n\ge 6,  \quad \text{or} \ n=4\ \text{and} \ p\ge 2, \quad \text{or} \ n=2\ \text{and} \ p\ge 3.\end{equation}
Moreover, we let 
\begin{equation} (\cR_\infty V)_{2,0}=(\cR_\infty V)_{2,1}=(\cR_\infty V)_{4,0}:=0. \end{equation}
The only remaining cases are
$(n,p)=(2,2), (4,1)$. For these values of $(n,p)$, 
$(\cR_\infty V)_{n,p}$ is defined in terms of an interpolation generalizing the definition of $(\tcR V)_{2,1}$ in \eqref{juve.m}. 
As a preparation for the definition, we first introduce $(\tcR V)_{n,p}$ for $(n,p)=(2,2), (4,1)$. For this purpose, we start 
from the analogues of \eqref{inte.r}, \eqref{juve.n} in the case that $(2,0)$ is replaced by $(n,p)=(2,1), (4,0)$: for such values of $(n,p)$ we write
\begin{equation}
	\begin{split}
		&\phantom{=}  \sum_{\bs z\in\Lambda_\infty^n}[V_{n,p}(\bs \omega,\bs D,\bs z)-\tcL V_{n,p}(\bs \omega,\bs D,\bs z)]
		\phi(\bs \omega,\bs D,\bs z)
		\\
		 &= \sum_{\bs z\in\Lambda_\infty^n}V_{n,p}(\bs \omega,\bs D,\bs z)[\phi(\bs\omega,\bs D,
		\bs z)-\phi(\bs\omega,\bs D,(z_1,z_1,\ldots,z_1))]
		\\
		&=\sum_{\bs z\in\Lambda_\infty^n}V_{n,p}(\bs \omega,\bs D,\bs z)\sum_{(\sigma,\bs D',\bs y)\in\INT(\bs z)}\sigma\phi(\bs \omega,\bs D+\bs D',\bs y).
	\end{split}
\label{last.r} 
\end{equation}
In the last expression, if $n=2$, then $\INT(\bs z)$ is the same defined after \eqref{eq:gamma_symmetry}; if $n=4$, then $\INT(\bs z)$ is the set of 
$(\sigma,(D_1,\ldots, D_4),(y_1,\ldots, y_4))\equiv ( \sigma, \bs D, \bs y)\in\{\pm\}\times\{0,\hat e_1,\hat e_2\}^4\times  \Lambda_\infty^4$ such that: either
$y_1=y_2=y_3=z_1$, $D_1=D_2=D_3=0$, and $(\sigma, (0,D_4), (z_1,y_4))\in\INT(z_1,z_4)$; or $y_1=y_2=z_1$, $y_4=z_4$, $D_1=D_2=D_4=0$, and $(\sigma, (0,D_3), (z_1,y_3))\in\INT(z_1,z_3)$;
or $y_1=z_1$, $y_3=z_3$, $y_4=z_4$, $D_1=D_3=D_4=0$, and $(\sigma, (0,D_2), (z_1,y_2))\in\INT(z_1,z_2)$. By summing \eqref{last.r} over $\bs D$ and exchanging the order of summations over $\bs z$ and $\bs y$, 
we find 
\begin{equation}
	\begin{split}
		&\sum_{\bs z\in\Lambda_\infty^n}\sum_{\bs D}^{(p)}V_{n,p}(\bs \omega,\bs D,\bs z)\sum_{(\sigma,\bs D',\bs y)\in\INT(\bs z)}\sigma\phi(\bs \omega,\bs D+\bs D',\bs y)
		\\
		&= \sum_{\bs y\in\Lambda_\infty^n}\sum_{\bs D}^{(p+1)}(\tcR V)_{n,p+1}(\bs \omega,\bs D,\bs y)\phi(\bs \omega,\bs D,\bs y).
	\end{split}
\end{equation}
with
\begin{equation} (\tcR V)_{n,p+1}(\bs \omega,\bs D,\bs y):=\sum_{\substack{\sigma, \bs z, \bs D':\\(\sigma,\bs D',\bs y)\in\INT(\bs z)}}\sigma V_{n,p}(\bs \omega,\bs D-\bs D',\bs z).\label{deftcR}\end{equation}
We are now ready to define: 
\begin{equation}
(\cR_\infty V)_{2,2}:= \cA(V_{2,2}+(\tcR V)_{2,2}+ (\tcR(\tcR V))_{2,2}), \qquad 
(\cR_\infty V)_{4,1}:=\cA (V_{4,1}+(\tcR V)_{4,1}).\end{equation}
Summarizing, 
\begin{equation}
(\cR_\infty V)_{n,p}=\begin{cases} 0 &\text{if} \ (n,p)=(2,0), (2,1), (4,0),\\
\cA(V_{2,2}+(\tcR V)_{2,2}+ (\tcR(\tcR V))_{2,2}) & \text{if} \ (n,p)=(2,2),\\
\cA(V_{4,1}+(\tcR V)_{4,1}) & \text{if} \ (n,p)=(4,1),\\
V_{n,p} & \text{otherwise}\end{cases}\label{defcR}\end{equation}
From the previous manipulations and definitions, it is clear that, if $V$ is invariant under translations and 
under the action of $\mathcal A$, then $V-\cL_\infty V\sim \cR_\infty V$. For later use, given $\bs D=(D_1,\ldots,D_n)$ with $\|\bs D\|_1=p$, 
we let $\cR_\infty V\big|_{\bs D}$ be the restriction of $(\cR V)_{n,p}$ to that specific choice of derivative label. 

\begin{Remark}\label{remLR}
From the definitions of $\cL_\infty$ and $\cR_\infty$, it also follows that, if $V$ is invariant under translations and under the action of $\mathcal A$, then 
$\cL_\infty(\cL_\infty V)=\cL_\infty V$ and $\cR_\infty(\cL_\infty V)=0$, two properties that will play a role in the following. 
\end{Remark}
 
{\it Norm bounds.} Let us conclude this section by a couple of technical estimates, which relate a suitable weighted norm of $\cR_\infty V$ to that of $V$, and will be useful in the following. 
Suppose that $V$ is translationally invariant. Let, for any $\kappa\ge 0$, 
\begin{equation} \|V_{n,p}\|_{(\kappa)}:=\sup_{\bs \omega}\sum_{\substack{\bs z\in\Lambda_{\infty}^n:\\ z_1\ \text{fixed}}}e^{\kappa\delta(\bs z)}\ \sup_{\bs D}^{(p)}\ |V_{n,p}(\bs \omega,\bs D,\bs z)|,\label{eq:weightnorm}\end{equation}
where the label $(p)$ on the sup over $\bs D$ indicates the constraint that $\|\bs D\|_1=p$, and $\delta(\bs z)$ is the tree distance of $\bs z$. 
With these definitions,
\begin{lemma}
	For any positive $\epsilon$,
\begin{eqnarray}
& \| (\tcR V)_{n,p}\|_{(\kappa)}\le (n-1)\epsilon^{-1} \| V_{n,p-1}\|_{(\kappa+\epsilon)} &\text{if} \quad (n,p)=(2,2), (4,1), \label{ek.1}\\
& \| (\tcR(\tcR V))_{2,2}\|_{(\kappa)}\le \epsilon^{-2} \| V_{2,0}\|_{(\kappa+2\epsilon)}. & \label{ek.2} \end{eqnarray}
\label{lem:R_bounds}
\end{lemma}

As a consequence, noting that $\|\cA V_{n,p}\|_{(\kappa)}\le
\|V_{n,p}\|_{(\kappa)}$ and recalling the definitions \eqref{defcR}, we find that 
\begin{align}
\| (\cR_\infty V)_{2,2}\|_{(\kappa)}&\le \|V_{2,2}\|_{(\kappa)}+\epsilon^{-1}\|V_{2,1}\|_{(\kappa+\epsilon)}+\epsilon^{-2}\|V_{2,0}\|_{(\kappa+2\epsilon)},\label{staz.ok0}\\
\| (\cR_\infty V)_{4,1}\|_{(\kappa)}&\le \|V_{4,1}\|_{(\kappa)}+3\epsilon^{-1} \|V_{4,0}\|_{(\kappa+\epsilon)}. \label{staz.ok} \end{align}
In the following, we will use bounds of this kind in order to evaluate the size of the renormalized part of the effective potential on scale $h$. In such a case, both $\kappa$ 
and $\epsilon$ will be chosen of the order $2^{h}$. 

\begin{proof}
	In order to prove \eqref{ek.1}, note that it follows directly from the defintion of $\tcR$ that 
	\begin{equation} \| (\tcR V)_{n,p}\|_{(\kappa)}\le \sup_{\bs\omega}\sum_{\substack{\bs z\in\Lambda_{\infty}^n:\\ z_1\ \text{fixed}}}\sup^{(p)}_{\bs D} \sum_{\substack{\sigma, \bs y, \bs D':\\(\sigma,\bs D',\bs z)\in\INT(\bs y)}}e^{\kappa\delta(\bs z)} |V_{n,p-1}(\bs \omega,\bs D-\bs D',\bs y)|.\label{ekko.1}\end{equation}
	If we now exchange the order of summations over $\bs z$ and $\bs y$, we find
	\begin{equation} 
		\| (\tcR V)_{n,p}\|_{(\kappa)}\le
		\sup_{\bs\omega} \sum_{\substack{\bs y\in\Lambda_{\infty}^n:\\ y_1\ \text{fixed}}}\ \sum_{\substack{\sigma, \bs z, \bs D':\\(\sigma,\bs D',\bs z)\in\INT(\bs y)}}e^{\kappa\delta(\bs z)} \sup^{(p-1)}_{\bs D}|V_{n,p-1}(\bs \omega,\bs D,\bs y)|.\label{ekko.2}
	\end{equation}
	Now note that $\delta(\bs z)\le \delta(\bs y)$ and that $|\INT(\bs y)|\le (n-1)\delta(\bs y)$, so that 
	\begin{equation} 
		\begin{split}
			\| (\tcR V)_{n,p}\|_{(\kappa)}
			\le& 
			(n-1)	\sup_{\bs\omega}
			\sum_{\substack{\bs y\in\Lambda_{\infty}^n:\\ y_1\ \text{fixed}}}\ e^{\kappa\delta(\bs y)} \delta(\bs y)\sup^{(p-1)}_{\bs D}|V_{n,p-1}(\bs \omega,\bs D,\bs y)|
			\\ \le & 
			\frac{n-1}{\epsilon}\|V_{n,p-1}\|_{\kappa+\epsilon},
		\end{split}
		\label{ekko.3}
	\end{equation}
	where in the last step we used the fact that $\delta\le e^{\epsilon\delta}/\epsilon$, for any $\epsilon>0$. A two-step iteration of the bound \eqref{ek.1} proves \eqref{ek.2}. 
\end{proof}
Similar estimates are valid for more general values of $(n,p)$, but Lemma~\ref{lem:R_bounds} includes all the cases which are relevant to the present work. 

\bigskip

{\it The running coupling constants.} At each scale $h\le 0$ we represent the infinite volume effective potential, as arising from the iterative application of 
\eqref{itTt} in the infinite volume limit, in the form \eqref{eq:forma}, namely:
\begin{equation} \cV^{(h)}_\infty(\phi)=\sum_{\Psi\in\cM_\infty}V_\infty^{(h)}(\Psi)\phi(\Psi).\end{equation}
For each $h\le 0$, in order to compute $\cV_\infty^{(h-1)}$ from $\cV^{(h)}_\infty$ via the infinite volume limit of \eqref{itTt}, we decompose $V^{(h)}_\infty\sim
\cL_\infty V^{(h)}_\infty+\cR_\infty V^{(h)}_\infty$. Note that the kernel $\cL_\infty V^{(h)}_\infty$, 
in light of \eqref{1strem}, \eqref{pauli} and \eqref{1strem}, takes the form
\begin{equation}\label{rCC} \cL_\infty V^{(h)}_\infty=2^h\nu_h F_{\nu,\infty}+\zeta_h F_{\zeta,\infty}+\eta_h F_{\eta,\infty}=:\upsilon_h\cdot F_{\infty}, \end{equation}
for three real constants $\nu_h,\zeta_h,\eta_h$, called the running coupling constants at scale $h$. The factor $2^h$ in front of $\nu_h$ is motivated by the fact that 
the $F_{\nu,\infty}=(F_{\nu,\infty})_{2,0}$ has scaling dimension $2-\frac{n}2-p\big|_{(n,p)=(2,0)}=1$, see \eqref{loc1}, see also \eqref{defscaldim} below. 

\subsection{Trees and tree expansions}\label{sec:tree_defs}

In this section, we describe the expansion for the kernels of the effective potentials, as it arises from the iterative application of \cref{itTt}. 
As anticipated above, it is convenient to graphically represent the result of the expansion in terms of GN trees. 
At the first step, recalling \eqref{eq:cV_1_def} and denoting by $N_c(\Psi)$ and $N_m(\Psi)$ the kernels of $\calN_c(\phi)$ and $\calN_m(\xi)$, respectively, we reorganize the 
expression for $V_\infty^{(0)}$ obtained by taking the infinite volume of \cref{eq:BBF0cyl} to obtain, for any $\Psi=(\bs\omega,\bs D,\bs z)\in\cM_\infty$ such that 
$\bs D=\bs 0$ (which we identify with the corresponding element $(\bs \omega,\bs z)$ of $\cup_{n\in2\mathbb N}(\{+,-\}\times\Lambda_\infty)^n$), 
\begin{equation}
	\begin{split}
		V^{(0)}_\infty(\Psi) 
		&= \ 
		N_c (\Psi)+Z^{-|\Psi|/2}V_\infty^{\rm int}(\Psi)
		+
		\sum_{\substack{\Psi_1 \in \cM_{1,\infty}:\\ \Psi_1\supset \Psi}}Z^{-|\Psi_1|/2}V^{\rm int}_\infty(\Psi_1)
		\fG_{\emptyset,\infty}^{(1)}(\Psi_1\setminus\Psi)\\ & 
		+
		\sum_{s=2}^{\infty}\frac{1}{s!}
		\sum_{\Psi_1,\ldots,\Psi_s\in \cM_{1,\infty}}^{(\Psi)}\ 
		\left(  \prod_{j=1}^s \big[N_m(\Psi_j)+Z^{-|\Psi_j|/2}V_\infty^{\rm int}(\Psi_j)\big] \right)
		\alpha(\Psi;\Psi_1,\ldots,\Psi_s)
		\\ & 
		\hphantom{
		+
		\sum_{s=2}^{\infty}\frac{1}{s!}
		\sum_{\Psi_1,\ldots,\Psi_s\in \cM_{1,\infty}}^{(\Psi)}\ 
		}
		\times
		\sum_{T \in \cS(\bar\Psi_1,\ldots,\bar\Psi_s)}
		\fG_{T,\infty}^{(1)}(\bar\Psi_1, \ldots,\bar\Psi_s),
	\end{split}
	\label{eq:BBF0a}
\end{equation}
where $V_\infty^{\rm int}$ is the infinite volume limit of the kernel of $\cV^{\rm int}(\phi,\xi):=\cV^{\rm int}(\phi,\xi,\bs 0)$, and we recall that in the first (resp. last) line, 
the factor $\fG_{T,\infty}^{(1)}(\Psi_1\setminus\Psi)$ (resp. $\fG_{T,\infty}^{(1)}(\bar\Psi_1, \ldots,\bar\Psi_s)$) is different 
from zero only if $\Psi_1\setminus\Psi$ is an element (resp. $\bar\Psi_1,\ldots,\bar\Psi_s$ are elements) of $\cup_{n\in2\mathbb N}(\{+i,-i\}\times\Lambda_\infty)^n$.

Similarly, at the following steps, for any $h\le 0$ and any $\Psi\in\cM_\infty$, we can write $V^{(h-1)}_\infty(\Psi)$, as computed via the infinite volume analogue of 
\eqref{itTt}, as follows: 
\begin{equation}
	\begin{split}
		V^{(h-1)}_\infty(\Psi) 
		&\sim \upsilon_h\cdot F_\infty(\Psi)+\cR_\infty V^{(h)}_\infty(\Psi)+
		\sum_{\substack{\Psi_1 \in \cM_{1,\infty}:\\ \Psi_1\supset \Psi}}\cR_\infty V^{(h)}_\infty(\Psi_1)
		\fG_{\emptyset,\infty}^{(h)}(\Psi_1\setminus\Psi)\\ & 
		+
		\sum_{s=2}^{\infty}\frac{1}{s!}
		\sum_{\Psi_1,\ldots,\Psi_s\in \cM_{1,\infty}}^{(\Psi)}\ 
		\left(  \prod_{j=1}^s \big[\upsilon_h\cdot F_\infty(\Psi_j)+\cR_\infty V_\infty^{(h)}(\Psi_j)\big] \right)
		\alpha(\Psi;\Psi_1,\ldots,\Psi_s)
		\\ & 
		\hphantom{
		+
		\sum_{s=2}^{\infty}\frac{1}{s!}
		\sum_{\Psi_1,\ldots,\Psi_s\in \cM_{1,\infty}}^{(\Psi)}\ 
		}
		\times
		\sum_{T \in \cS(\bar\Psi_1,\ldots,\bar\Psi_s)}
		\fG_{T,\infty}^{(h)}(\bar\Psi_1, \ldots,\bar\Psi_s),
	\end{split}
	\label{eq:BBFh-1ab}
\end{equation}
where $\fG_{T,\infty}^{(h)}(\bar\Psi_1,\ldots,\bar\Psi_s)$ is the infinite volume analogue of the function defined in \eqref{etcetcbase},
differing from it for an important feature (besides the `obvious' one that $\fG_{T,\infty}^{(h)}$ is defined in terms of the infinite plane propagators $g^{(h)}_{\ell,\infty}$ 
rather than those on the cylinder): since now the field multilabels $\Psi_i$ have the form $(\bs \omega_i,\bs D_i,\bs z_i)$, with $\bs D_i$ different from $\bs 0$, in general, 
the infinite plane propagators $g_{\ell,\infty}^{(h)}$, with $\ell=((\omega_i,D_i,z_i),(\omega_j,D_j,z_j))$, entering the 
definition of $\fG_{T,\infty}^{(h)}$ should now be interpreted as $\partial^{D_i}_{z_i}\partial^{D_j}_{z_j}g^{(h)}_{\omega_i\omega_j}(z_i,z_j)$.

We graphically interpret \eqref{eq:BBF0a} as in Figure~\ref{eq:figtree1}.
\begin{figure}
\begin{center}
	\begin{tabular}{rcl}
		$V_\infty^{(0)}$ & $=$ & 
		\tikz[baseline=-2pt]{ \draw (0,0) node {} -- (1,0) node[ctVertex] {} ; \draw (0,-0.5) node {$1$}; \draw (1,-0.5) node {$2$}; }
		+
		\tikz[baseline=-2pt]{ \draw (0,0) node {} -- +(1,0) node[vertex] {};}
		+
		\tikz[baseline=-2pt]{ \draw (0,0) node[vertex] {} -- +(1,0) node[vertex] {};}
		+
		\tikz[baseline=-2pt]{ \draw (0,0) node[vertex] {} -- +(1,0.5) node[vertex] {};
		\draw (0,0) -- (1,-0.5) node[vertex] {};}
		+
		\tikz[baseline=-2pt]{ \draw (0,0) node[vertex] {} -- +(1,0.5) node[vertex] {};
		\draw (0,0) -- (1,-0.5) node[ctVertex] {};}
		$+\cdots$
		\\ 
		\rotatebox{90}{$=$} & & 
		\\
		$\cL_\infty V_\infty^{(0)} + \cR_\infty V_\infty^{(0)}$
		&=&
		\tikz[baseline=-2pt]{\draw (0,0) node[ctVertex] {} ; \draw (0,-0.5) node {1}} 
		+
		\tikz[baseline=-2pt]{\draw (0,0) node[bigvertex] {} ; \draw (0,-0.5) node {1}} 
		\vphantom{
		\tikz[baseline=-2pt]{ \draw (0,0) node[vertex] {} -- +(1,0.5) node[vertex] {};
		\draw (0,0) -- (1,-0.5) node[ctVertex] {};}
		}
	\end{tabular}
\end{center}\caption{Graphical interpretation of \eqref{eq:BBF0a}}
\label{eq:figtree1}
\end{figure}
On the right hand side of the first line we have drawn a series of diagrams consisting of a root at scale $1$, which we will usually denote $v_0$, connected to $s$ other vertices (which we will call {\it endpoints}) at scale $2$ which are of two different types: 
\tikzvertex{ctVertex}, which we call {\it counterterm endpoints} and which represent $N_c$ or $N_m$,
and \tikzvertex{vertex}, called {\it interaction endpoints} and representing $V_\infty^{\rm int}$.
In the first two terms (in which there is no factor $\fG^{(1)}_{T,\infty}$, because it is `trivial', i.e., it equals $\fG^{(1)}_{\emptyset,\infty}(\emptyset)=1$) 
the root is drawn simply as the end of a line segment (we will say it is {\it undotted}); 
in the other terms (including all those with $s>1$) we draw a dot \tikzvertex{vertex} to indicate the presence of a nontrivial $\fG^{(1)}_{T,\infty}$ and $\alpha$ factors 
and additional sums.

In order to iterate the scheme, we decompose $V_\infty^{(0)}$ as $V_\infty^{(0)}\sim\cL_\infty V_\infty^{(0)}+\cR_\infty V_\infty^{(0)}$ and graphically represent
$\cL_\infty V_\infty^{(0)}=\upsilon_0\cdot F_\infty$ by a counterterm vertex \tikzvertex{ctVertex} at scale $1$, and $\cR_\infty V_\infty^{(0)}$ by a big vertex \tikzvertex{bigvertex} at scale $1$, as indicated in the second line of Figure~\ref{eq:figtree1}. 
Next, using the conventions of Figure~\ref{eq:figtree1}, we graphically represent $V_\infty^{(-1)}$, computed by \eqref{eq:BBFh-1ab} with $h=0$, 
as described in Figure~\ref{eq:figtree2}. 
\begin{figure}
\begin{center}
	\begin{tabular}{rl}
		$V^{(-1)}_\infty=$
		&
		\tikz[baseline=-2pt]{\draw (0,0) node {} -- (1,0) node [ctVertex] {};
			\draw (0,-0.5) node {0}; \draw (1,-0.5) node {1};
			}
		+
		\tikz[baseline=-2pt]{\draw (0,0) node {} -- (1,0) node [bigvertex] {};
			\draw (0,-0.5) node {0}; \draw (1,-0.5) node {1};
			}
		+
		\tikz[baseline=-2pt]{\draw (0,0) node [vertex] {} -- (1,0) node [bigvertex] {};
			}
		+
		\tikz[baseline=-2pt]{\draw (0,0) node [vertex] {} -- (1,0.5) node [bigvertex] {};
			\draw (0,0) -- (1,-0.5) node[bigvertex] {};
			}
		$+ \cdots$
		\\
		= &
		\tikz[baseline=-2pt]{\draw (0,0) node {} -- (1,0) node [ctVertex] {};
			\draw (0,-0.5) node {0}; \draw (1,-0.5) node {1};
			}
		+
		\tikz[baseline=-2pt]{\draw (0,0) node {} -- (1,0) node [label=above:{$\cR_\infty$}] {} -- (2,0) node[vertex] {};
			\draw (0,-0.5) node {0}; \draw (1,-0.5) node {1};\draw (2,-0.5) node {2};
			}
		+
		\tikz[baseline=-2pt]{\draw (0,0) node {} -- (1,0) node [vertex,label=above:{$\cR_\infty$}] {} -- (2,0) node[vertex] {};
			\draw (0,-0.5) node {0}; \draw (1,-0.5) node {1};\draw (2,-0.5) node {2};
			}
		+
		\tikz[baseline=-2pt]{\draw (0,0) node {} -- (1,0) node [label=above:{$\cR_\infty$}] {} -- (2,0) node [ctVertex] {} ;
			\draw (0,-0.5) node {0}; \draw (1,-0.5) node {1};\draw (2,-0.5) node {2};
		}
		\\
		\\ & $+ \cdots +$ 
		\tikz[baseline=-2pt]{\draw (0,0) node[vertex] (v0) {} -- (1,-0.5) node[vertex,label=below:{$\cR_\infty$}] (v1) {} -- +(1,0.3) node[vertex] {};
			\draw (v1) -- +(1,-0.2) node[ctVertex] {};
			\draw (v1) -- +(1,-0.7) node[vertex] {};
			\draw (v0) -- +(1,0.5) node[vertex,label=above:{$\cR_\infty$}] (v2) {};
			\draw (v2) -- +(1,0.25) node[vertex] {};
			\draw (v2) -- +(1,-0.25) node[vertex] {};
		}
		$+ \cdots$
	\end{tabular}
\end{center}\caption{Graphical interpretation of \eqref{eq:BBFh-1ab} with $h=0$.}\label{eq:figtree2}
\end{figure}
In passing from the first to the second line of Figure~\ref{eq:figtree2} we expanded the big vertex on scale $1$, which represents $\cR_\infty V_\infty^{(0)}$, 
by using the first line of Figure~\ref{eq:figtree1}, with an additional label $\cR_\infty$ on the vertices on scale $1$, to represent the action of $\cR_\infty$.  

The graphical equations in Figure~\ref{eq:figtree1} and Figure~\ref{eq:figtree2} are the analogues of the graphical equations in \cite[Figures~6-7]{GM01}, which 
contains a more detailed discussion of some aspects of this construction.
By iterating the same kind of graphical equations on lower scales, expanding the big vertices \tikzvertex{bigvertex} until we are left 
with endpoints all of type \tikzvertex{ctVertex} or \tikzvertex{vertex}, we find that $V_\infty^{(h)}$ can be graphically 
expanded in terms of trees of the kind depicted in Figure~\ref{fig:cT_example}, with the understanding that in principle there should be a label $\cR_\infty$ 
at all the intersections of the branches with the vertical lines, with the sole exception of $v_0$; however, by convention, in order not to overwhelm the 
figures, we prefer not to indicate these labels explicitly. We call such trees `GN trees', and denote by $\cT^{(h)}_\infty$, with $h\le 0$, the set of GN trees 
with root $v_0$ on scale $h+1$. We call `vertices' of a GN tree the root $v_0$, all its dotted nodes, and its 
endpoints. 

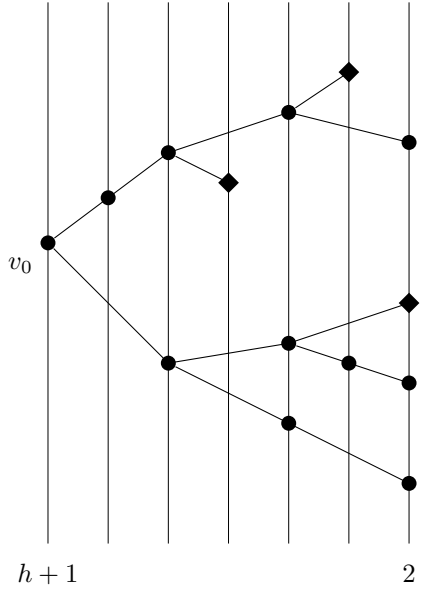
\begin{figure}[t]
	\centering
	  \begin{tikzpicture}[scale=0.8]
      \foreach \x in {3,...,9}
      {
	  \draw[very thin] (\x ,-1) -- (\x , 8);
      }
      \draw (3,-1.5) node {$h+1$};
      \draw (9,-1.5) node {$2$};
      \draw (3,4) node[vertex, label=-135:{$v_0$}] (vbar) {};
      \draw (vbar) -- ++ (1,0.75) node[vertex] {} -- ++(1,0.75) node[vertex] (vt1) {} -- ++ (2,0.67) node[vertex] (vt11) {}
      -- ++(2,-0.5) node[vertex] {};
      \draw (vt11) -- ++(1,0.67) node[ctVertex] (t1e1) {}; 
      \draw (vt1) -- ++(1,-0.5) node[ctVertex] (t1e2) {};
      \draw (vbar) -- (5,2) node[vertex] (v1) {} -- (7,2.33) node[vertex] (vt2) {} -- (9,3) node[ctVertex] {};
      \draw (vt2) -- ++(1,-0.33) node[vertex] (t2e1) {} -- ++(1,-0.33) node[vertex] (t2e2) {};
      \draw (v1) -- ++(2,-1) node[vertex] (t3v1) {} -- +(2,-1) node[vertex] {};
    \end{tikzpicture}
    \caption{Example of a tree in $\cT_\infty^{(h)}$. As explained in the text, one should imagine that a label $\cR$, indicating an action 
    of the $\cR$ operator, is present at all the intersections of the branches with the vertical lines, with the only exception of $v_0$. In order 
    not to overwhelm the figures, these labels are left implicit.}
	\label{fig:cT_example}
\end{figure}

We introduce some conventions and observations about the set of GN trees:
\begin{itemize}
\item The root $v_0$ is the unique leftmost vertex of the tree. Its degree (number of incident edges) must be at least 1, i.e., $v_0$ cannot be an endpoint. 
It may or may not be dotted; in order for $v_0$ not to be dotted, its degree must be $1$.
  \item Vertices, other than the root, with exactly one successor, are called `trivial'. 
  \item Interaction endpoints \tikzvertex{vertex} can only be on scale $2$. Counterterm endpoints \tikzvertex{ctVertex} can be on all scales $\le 2$; if 
  such an endpoint is on a scale $h<2$, then it must be connected
to a non-trivial vertex on scale $h-1$. [The reason is the following: if this were not case, then there would be an $\cR_\infty$ operator acting on the value of the 
endpoint, but this would annihilate it, because a \tikzvertex{ctVertex} endpoint on scale $h<2$ corresponds to $\cL_\infty V_\infty^{(h)}$, and the 
definitions of $\cL_\infty,\cR_\infty$ are such that $\cR_\infty(\cL_\infty V_\infty^{(h)})\sim 0$, see Remark \ref{remLR}.]
\end{itemize}

In terms of these trees, we shall write the expansion for $V^{(h)}_\infty=V_\infty^{(h)}[\allct]$, thought of as a function of $\allct:=\{(\nu_h,\zeta_h,\eta_h)\}_{h\le 0}$, as 
\begin{equation}\label{eq:exptree.1}
V^{(h)}_\infty[\allct]\sim \sum_{\tau\in\cT^{(h)}_\infty}W_\infty[\allct;\tau].
\end{equation}
In order to write
$W_\infty[\allct;\tau]$ more explicitly, we need to specify some additional notations and conventions about GN trees.
Let $\tau\in\cT_\infty:=\cup_{h\le 0}\cT^{(h)}_\infty$ and $v_0=v_0(\tau)$ be its root. Then:
\begin{itemize}
\item We let $V(\tau)$ be the set of vertices, $V_{e}(\tau)\subset V(\tau)$ the set of endpoints, and $V_0(\tau):=V(\tau)\setminus V_e(\tau)$.
We also let $V'(\tau):=V(\tau)\setminus\{v_0\}$ and $V'_0(\tau):=V_0(\tau)\setminus\{v_0\}$.
\item Given $v\in V(\tau)$, we let $h_v$ be its scale. 
  \item $v \ge w$ or `$v$ is a successor of $w$' means that the (unique) path from $v$ to $v_0$ passes through $w$.
  Obviously, $v>w$ means that $v$ is a successor of $w$ and $v\neq w$.   
  \item `$v$ is an immediate successor of $w$', denoted $v \successor w$, means that $v \ge w$, $v \neq w$, and $v$ and $w$ are directly connected.
	  For any $v \in \tau$, $S_v$ is the set of $w \in \tau$ such that $w \successor v$.
  \item For any $v>v_0$, we denote by $v'$ the unique vertex such that $v \rhd v'$.  
  \item Subtrees: for each $v \in V_0(\tau)$, let $\tau_v \in \cT^{(h_v-1)}_\infty$ denote the subtree consisting of the vertices with $w \ge v$.
\end{itemize}
Next,  we need to attach 
labels to their vertices, in order to distinguish the various contributions to the kernels arising from the different choices of the 
sets $\Psi_i$, etc, in \eqref{eq:BBF0a}, \eqref{eq:BBFh-1ab}, also keeping track of the order in which they appear. In particular, with each $v\in V(\tau)$ we associate a set $P_v$ of {\it field labels},
sometimes called the set of {\it external fields},
whose elements carry two informations: their position within an ordered list which they belong to, and their $\omega$ index; more precisely, 
the family $\ul P=\{P_v\}_{v\in V(\tau)}$ is characterized by the following properties, which correspond to properties of the iteration of the kernel: 
\begin{itemize}
	\item $|P_v|$ is always even and positive. 
		If $v$ is a \tikzvertex{ctVertex} endpoint, then $|P_v| = 2$.
	\item If $v$ is an endpoint of $\tau$, then $P_v$ has the form $\left\{ (j,1,\omega_1),\dots,(j,2n,\omega_{2n}) \right\}$, 
	where $j$ is the position of $v$ in the ordered list of endpoints, and $\omega_l\in\{+,-,i,-i\}$, if $h_v=2$, while $\omega_l\in\{+,-\}$, if 
	$h_v<2$. Given $f=(j,l,\omega_l)$, we let $o(f)=(j,l)$ and $\omega(f):=\omega_l$.
	\item If $v$ is not an endpoint, $P_v \subset \bigcup_{w \in S_v} P_w$.
	\item If $v\in V_0(\tau)$, we let $Q_v:= \left( \bigcup_{w \in S_v} P_w \right) \setminus P_v$ be the set of {\it contracted fields}. If $v$ is dotted, then we require 
$|Q_v| \ge 2 $ and $|Q_v| \ge 2 (|S_v| - 1)$; and conversely $Q_v$ is empty if and only if $v=v_0$ and $v_0$ is not dotted. 
	\item If $h_v = 1$ and $v$ is not an endpoint, then $Q_v = \bigcup_{w \in S_v} \condset{f \in P_w\, }{\, \omega(f) \in\{+i,-i\}}$ (all and only massive fields are integrated on scale $1$).\end{itemize}

For $\tau \in \cT_\infty$, we denote by $\cP(\tau)$ be the set of allowed $\ul P=\{P_v\}_{v\in V(\tau)}$. We also denote by $\bs \omega_{v}$ the tuple 
of components $\omega(f)$, with $f\in P_v$, and by $\bs \omega_v\big|_Q$  the restriction of $\bs\omega_v$ to any subset $Q\subseteq P_v$.
Note that the definitions 
imply that for $v,w \in \tau$ such that  neither $v \geq w$ or $v \leq w$ (for example when $v' = w'$ but $v \neq w$), $P_v$ and $P_w$ are disjoint,
as are $Q_v$ and $Q_w$.

Next, given $\ul P\in \cP(\tau)$, for all $v\in V_0(\tau)$ we define sets $T_v$, 
$$T_v = \left\{ \left( f_1,f_2 \right),\dots, \left(f_{2 |S_v| - 3}, f_{2 |S_v| - 2} \right) \right\}\subset Q_v^2,$$
called {\it spanning trees} associated with $v$, characterized by the following properties: 
if $w(f)$ denote the (unique) $w \in S_v$ for which $f \in P_w$, then $(f,f')\in T_v\Rightarrow$ $w(f)\neq w(f')$ and $o(f)<o(f')$; moreover, 
$\left\{ \left\{ w(f_1), w(f_2) \right\} , \dots,  \left\{  w(f_{2 |S_v| - 3}), w(f_{2 |S_v| - 2})\right\}\right\}$
is the edge set of a tree with vertex set $S_v$. We denote by $\cS(\tau,\ul P)$ the set of allowed $\ul T=\{T_v\}_{v\in V_0(\tau)}$.

Finally, for each $v\in V(\tau)$, we denote by $D_v$ a map $D_v : P_v \to \cD=\{D\in\{0,1,2\}^2: \|D\|_1\le 2\}$; the reader should think that a derivative operator 
$\partial^{D_v(f)}$ acts on the field labelled $f$. We denote by $\cD(\tau, \ul P)$ the set of families of maps $\ul D=\{D_v\}_{v\in V(\tau)}$. We also denote 
by $\bs D_v$ the tuple of components $D_v(f)$, with $f\in P_v$, and by $\bs D_v\big|_{Q}$ the restriction of $\bs D_v$ to any subset $Q\subseteq P_v$. 
Additionally, if a map $z: P_v\to \Lambda_\infty$ is assigned, we denote 
by $\bs z_v$ the tuple of components $z(f)$, with $f\in P_v$, and by $\Psi_v=\Psi(P_v):=(\bs\omega_v,\bs D_v,\bs z_v)$ the field multilabel associated with 
$\bs\omega_v,\bs D_v,\bs z_v$; moreover, if $v\in V_0(\tau)$ and also the maps $z: P_w\to \Lambda_\infty$, for all $w\in S_v$, are assigned, for each $w\in S_v$ 
we denote by 
$\bar\Psi_w=\Psi(P_w\setminus P_v)=(\bs \omega_w\big|_{P_w\setminus P_v}, \bs D_w\big|_{P_w\setminus P_v}, \bs z_w\big|_{P_w\setminus P_v})$
the restriction of $\Psi_w$ to $P_w\setminus P_v$ (here $\bs z_w\big|_Q$ is the restriction of $\bs z_w$ to the subset $Q\subset P_w$).  

\medskip

In terms of these definitions, we write $W_\infty[\allct;\tau]$ in the right side of \eqref{eq:exptree.1} as
\begin{equation}\label{Wtauexp}
	W_\infty[\allct;\tau]
	= 
	\sum_{\ul P\in\cP(\tau)}\sum_{\ul T\in\cS(\tau,\ul P)}\sum_{\ul D\in \cD(\tau,\ul P)}
	W_\infty[\allct;\tau,\ul P,\ul T, \ul D], 
\end{equation}
where $W_\infty[\allct;\tau,\ul P,\ul T, \ul D]$ is the translationally invariant kernel inductively defined as follows: 
letting $\bs D_{v_0}':=\bigoplus_{v\in S_{v_0}}\bs D_v\big|_{P_{v_0}}$ for $h_{v_0}<1$ and $\bs D_{v_0}':={\bf 0}$ for $h_{v_0}=1$, 
\begin{equation}
	\begin{split}
		&
		W_\infty [\allct; \tau, \ul P, \ul T, \ul D](\bs \omega_0,\bs D_0,\bs z_0)=\mathds 1(\bs \omega_0=\bs \omega_{v_0})\mathds 1\Big(\bs D_0=\bs D_{v_0}=\bs D_{v_0}'\Big)\,\frac{\alpha_{v_0}}{ |S_{v_0}|!}
		\\
		&\qquad \times  \sum_{\substack{z: P_{v_0} \cup Q_{v_0}\to \Lambda_\infty\\ 
		\bs z_0= \bs z_{v_0}}}\hskip-.2truecm
		\fG_{T_{v_0},\infty}^{(h_{v_0})}(\bar\Psi_{v_1},\ldots,\bar\Psi_{v_{s_{v_0}}})
		\prod_{v \in S_{v_0}} K_{v,\infty}(\Psi_v), 	
	\end{split}
	\label{eq:W8def:1.2}
\end{equation}
where $\alpha_{v_0}=\alpha(\Psi_{v_0};\Psi_{v_1},\ldots,\Psi_{v_{s_{v_0}}})$, cf.\ \eqref{eq:BBF0a}, and 
we recall that, if $|S_{v_0}|=1$, then $T_{v_0}=\emptyset$. In this case, if $\Psi_{v_1}=\Psi_{v_0}$, then $\fG_{\emptyset,\infty}^{(h_{v_0})}(\emptyset)$ should be interpreted as being equal to $1$; this latter case is the one in which, graphically, $v_0$ is not dotted. In the second line of \eqref{eq:W8def:1.2}, if $h_{v_0}=1$, 
\begin{equation}
	K_{v,\infty}(\Psi_v)= \piecewise{ N_c(\Psi) & \text{if $v$ is of type}\ \tikzvertex{ctVertex} \ \text{and $v_0$ is undotted},\\
	N_m(\Psi_v) & \text{if $v$ is of type}\ \tikzvertex{ctVertex} \ \text{and $v_0$ is dotted},\\
		  Z^{-|\Psi_v|/2}V_\infty^{\rm int}(\Psi_v)& \text{if $v$ is of type}\ \tikzvertex{vertex},}\label{Kvinfty1}\end{equation}
while, if $h_{v_0}<1$, 
\begin{equation} 
	K_{v,\infty}(\Psi_v)
	:=	
	\piecewise{\upsilon_{h_{v_0}}\cdot F_\infty(\Psi_v) 
		  	& \text{if $v\in V_e(\tau)$ is of type \tikzvertex{ctVertex} and $h_v = h_{v_0}+1$}\\ 
			\cR_\infty N_c (\Psi_v) & \text{if $v\in V_e(\tau)$ is of type \tikzvertex{ctVertex} and $h_v = 2$}\\ 
			Z^{-|\Psi_v|/2}\cR_\infty V_\infty^{\rm int}(\Psi_v) & \text{if $v\in V_e(\tau)$ is of type \tikzvertex{vertex}} \\
			\lis W_\infty [\allct; \tau_v, \ul P_{v}, \ul T_{v}, \ul D_{v}](\Psi_v)
		  	& \text{if $v\in V_0(\tau)$,}\label{defKv4.1}
		}
\end{equation}
where, in the last line of \eqref{defKv4.1}, letting $\ul P_v$ (resp. $\ul T_v$, resp. $\ul D_v$) be the restriction of $\ul P$ (resp. $\ul T$, resp. $\ul D$) 
to the subtree $\tau_v$, and $\ul D_v':=\{D_{v}'\}\cup\{D_w\}_{w\in V(\tau): w>v_0}$ (here $D_{v}'$ is the map such that 
$\bs D_{v}':=\bigoplus_{w\in S_{v}}\bs D_w\big|_{P_{v}}$), we denoted 
\begin{equation}
\label{lisWinfty}\lis W_\infty [\allct; \tau_v, \ul P_{v}, \ul T_{v}, \ul D_{v}]:=\cR_\infty W_\infty [\allct; \tau_v, \ul P_{v}, \ul T_{v}, \ul D_{v}']\Big|_{\bs D_{v}},
\end{equation}
and we recall that the definition of $\cR_\infty V\big|_{\bs D}$ was given a few lines after \eqref{defcR}. 
The inductive proof that $V_\infty^{(h)}=V_\infty^{(h)}[\allct]$, as iteratively computed by \eqref{eq:BBF0a}-\eqref{eq:BBFh-1ab}, 
is equivalent to $\sum_{\tau\in\cT^{(h)}_\infty}W_\infty[\allct;\tau]$, with $W_\infty[\allct;\tau]$ as in 
\eqref{Wtauexp}, \eqref{eq:W8def:1.2}, is straightforward and left to the reader. 

\begin{Remark} Given $\tau\in \cT_\infty$ and $\ul P\in\cP(\tau)$, we say that $\ul D$ is `allowed' if $W_\infty [\allct; \tau, \ul P, \ul T, \ul D]\not\sim 0$.
With some abuse of notation, from now on we will re-define $\cD(\tau,\ul P)$ 
to be the set of allowed $\ul D$ for a given $\tau$ and $\ul P$. Of course, such a redefinition has no impact on the validity of \eqref{Wtauexp}. 
If $\ul D$ is allowed, then it must satisfy a number of constraints. In particular, 
if $v\in V_0(\tau)$, $w\in V_e(\tau)$ and $f\in P_v\cap P_w$, then $D_v(f)\ge D_w(f)$. Moreover, 
letting, for any $v\in V_0(\tau)$, $R_v:=\| \bs D_v\|_1-\sum_{w\in S_v}\| \bs D_w\big|_{P_v}\|_1\equiv \| \bs D_v\|_1-\| \bs D_v'\|_1$, one has
	\begin{equation}
		R_v
		=
		\piecewise{
			2, & 
				|P_v| = 2 \textup{ and } \left\| \bs D_v' \right\|_1 = 0 \\ 
			1, & 
				|P_v| = 2 \textup{ and } \left\| \bs D_v' \right\|_1 = 1 \textup { or } 
				|P_v| = 4 \textup{ and } \left\| \bs D_v' \right\|_1 = 0 \\ 
			0, & \textup{otherwise},
		}
		\label{eq:Rv_def}
	\end{equation}
	with the exception of $v_0$, for which $R_{v_0}\equiv 0$ (in other words $\bs D_{v_0}=\bs D_{v_0}'$, see \eqref{eq:W8def:1.2}). 
Finally,  the combination 
\begin{equation} \label{defscaldim} d(P_v,\bs D_v):=2-\frac{|P_v|}{2}-\|\bs D_v\|_1,\end{equation}
known as the {\bf scaling dimension} of $v$, is $\le -1$ for all $v\in V_0'(\tau)$, and for all $v\in V_e(\tau)$ such that $h_v=2$ and $h_{v'}<1$. 
As we shall see below, see in particular the statement of Lemma \ref{lm:W8:scaldim}, the fact that $d(P_v,\bs D_v)\le -1$ for all such vertices 
guarantees that the expansion in GN trees is convergent uniformly in $h_{v_0}$. 
\label{remcruc}
\end{Remark}

\begin{remark}
	With the other arguments fixed, the number of choices of $\ul D$ for which $W_\infty[\allct;\tau,\ul P, \ul T, \ul D]$ does not vanish is no more than 
	$10^{\left|V(\tau)\right|}$: there is a choice of at most 10 possible values for each endpoint\footnote{The value 10 bounds 
	the number of different terms that the operator $\cR_\infty$ produces when it acts non-trivially on an interaction endpoints. In fact, the cases in which $\cR_\infty$ acts non-trivially are those listed in the right side of \eqref{defcR} with $(n,p)=(2,2), (4,1)$. If $(n,p)=(2,2)$, the number of possible values taken by $D_v$ is 10 (one derivative in direction $i\in\{1,2\}$ on the first Grassmann field and one derivative in direction $j\in\{1,2\}$ on the second Grassmann field, etc.); if $(n,p)=(4,1)$, the number of possible values taken by $D_v$ is 8, which is smaller than 10 (one derivative in direction $i\in\{1,2\}$ on the $k$-the Grassmann field, with $k\in\{1,2,3,4\}$).}, and then the other values are fixed except for a choice of up to 10 possibilities each time that $\cR_\infty$ acts non trivially on a vertex $v\in V_0'(\tau)$, i.e., each time that, for such a vertex, $R_v>0$, see \eqref{eq:Rv_def}. 
	\label{rem:D_choice}
\end{remark}

\subsection{Bounds on the kernels of the full plane effective potentials}
\label{sec:formal_bounds}

In this subsection we show that the 
norm of the kernels $W_\infty[\allct;\tau,\ul P,\ul T,\ul D]$ is summable over $\tau,\ul P,\ul T,\ul D$, provided that the elements of the sequence $\allct$ are 
bounded and sufficiently small. We measure the size of $W_\infty[\allct;\tau,\ul P,\ul T,\ul D]$ in terms of the weighted norm 
\eqref{eq:weightnorm}, with $\kappa=\frac{c_0}22^{h_{v_0}}$, where $h_{v_0}$ is the scale of the root of $\tau$ and $c_0$ is the minimum between the 
constant $c$ in Proposition \ref{thm:g_decomposition} and half the constant $c$ in \eqref{eq:W_free_bound}, \eqref{eq:WL1bound}, \eqref{eq:WE_base_decayrenzy}, 
\eqref{eq:WE_base_decayletta}. With some abuse of notation, we let 
\begin{equation} \label{defnormh}
\|W_\infty[\allct;\tau,\ul P,\ul T,\ul D]\|_{h_{v_0}}:=
\|W_\infty[\allct;\tau,\ul P,\ul T,\ul D]\|_{(\frac{c_0}{2}2^{h_{v_0}})}.
\end{equation}
The first, basic, bound on the kernels $W_\infty[\allct;\tau,\ul P,\ul T,\ul D]$ is provided by the following proposition. 
We recall that we assumed once and for all that $Z\in U$, with $U=\{z\in\mathbb R: |z-1|\le 1/2\}$, and that $t_1^*,t_1,t_2\in K'$, with $K'$ the compact set defined before the statement of Theorem \ref{prop:main}. 

\begin{proposition} Let $W_\infty[\allct;\tau,\ul P,\ul T,\ul D]$ be inductively defined as in \eqref{eq:W8def:1.2}. There exist $C,\kappa,\lambda_0>0$ 
such that, for any $\tau\in \cT_\infty$, $\ul P\in \cP(\tau)$, $\ul T\in \cS(\tau,\ul P)$, $\ul D\in \cD(\tau,\ul P)$, and $|\lambda|\le \lambda_0$,
	\begin{eqnarray}
&& \|W_\infty[\allct;\tau,\ul P,\ul T,\ul D]\|_{h_{v_0}}\le     C^{\sum_{v\in V_e(\tau)}|P_v|}
	    \Big(
	      \prod_{v \in V_0(\tau)}
	      \frac{2^{(\tfrac12 |Q_v|+ \sum_{w\in S_v}\| \bs D_w |_{Q_v} \|_1 - R_v + 2 - 2 \left|S_v\right|)h_v}}{ |S_v|!}
	    \Big)\nonumber
	    \\ && \times \prod_{v \in V_e(\tau)}\piecewise{
	      2^{(h_v-1)(2 - \tfrac12 |P_v| - \left\|\bs D_v\right\|_1 )}\epsilon_{h_v-1} & \text{if $v$ is of type \tikzvertex{ctVertex}}\\
	     |\lambda|^{\max\{1,\kappa|P_v|\}} & \text{if $v$ is of type \tikzvertex{vertex}}}
		\label{eq:W8_primitive_bound}
	  \end{eqnarray}
	where $\epsilon_h:=\max\{|\nu_h|,|\zeta_h|,|\eta_h|\}$ if $h \le 0$ and 
	$\epsilon_1 := \max \left\{ \| N_c \|_2, \| N_m\|_2 \right\}$.
	\label{thm:W8:existence}
\end{proposition}

\begin{proof}
Let us first consider the case $h_{v_0}=1$, in which case, 
using \eqref{eq:W8def:1.2} and \eqref{Kvinfty1}, we find
\begin{eqnarray} &&
\|W_\infty[\allct;\tau,\ul P,\ul T,\ul 0]\|_{h_{v_0}}\le 
\frac{1}{|S_{v_0}|!}
\sum_{\substack{z: \cup_{v\in S_{v_0}}\!P_{v}\to \Lambda_\infty:\\ 
	z(f_1)\ \text{fixed}}}
e^{c_0\delta(\bs z_{v_0})} \big|\fG^{(1)}_{T_{v_0},\infty}(\bar\Psi_{v_1},\ldots,\bar\Psi_{v_{s_{v_0}}})\big|
\cdot \nonumber\\
&& \qquad \cdot \prod_{v \in S_{v_0}}
		\piecewise{| N_c(\Psi_v)| & \text{if $v$ is of type}\ \tikzvertex{ctVertex}\ \text{and $v_0$ is undotted},\\
		| N_m (\Psi_v)| & \text{if $v$ is of type}\ \tikzvertex{ctVertex}\ \text{and $v_0$ is dotted},\\
		 |Z|^{-|\Psi_v|/2} |V_\infty^{\rm int}(\Psi_v)|& \text{if $v$ is of type}\ \tikzvertex{vertex},}\label{ekk.k1}
\end{eqnarray}
where $f_1$ is the first element of $P_{v_0}$. By using the definition of $\fG^{(1)}_{T,\infty}$ and the property $({\rm Pf}M)^2=\det M$, valid
for any antisymmetrix matrix $M$, we find 
\begin{equation} \big|\fG^{(1)}_{T,\infty}(Q_1,\ldots, Q_s)\big|\le\Big(
\prod_{\ell\in T}|g_{\ell,\infty}^{(1)}| \Big)\sup_{\bs t}\sqrt{|\det G^{(1)}_{Q_1,\ldots, Q_s, T,\infty}(\bs t)|},\label{ekk.k0}
\end{equation}
so that, thanks to items 1,3,4 of \cref{thm:g_decomposition} (which apply to $\fg^{(1)}_\infty$ by \cref{rem:gramfgE}), and to the Gram-Hadamard inequality \cite[Appendix~A.3]{GM01}, which allows one to bound the determinant of any matrix $M$ with elements $M_{i,j}=(\gamma_i,\tilde\gamma_j)$ as 
$|\det M|\le \prod_{i}|\gamma_i|\, |\tilde \gamma_i|$, 
\begin{equation} 
	\begin{split}
		\big|\fG^{(1)}_{T,\infty}(Q_1,\ldots, Q_s)\big|
		& \le 
		C^{s-1} \Big(
		\prod_{(f,f')\in T} e^{-2c_0 \|z(f)-z(f')\|_1} \Big)\, \Big(\prod_{f\in \cup_i Q_i\setminus T} |\gamma_{\omega(f),0,z(f)}^{(1)}|\cdot |\tilde \gamma_{\omega(f),0,z(f)}^{(1)}|\Big)^{1/2}
		\\ & \le 
		(C')^{\sum_i |Q_i|} \Big(
		\prod_{(f,f')\in T} e^{-2c_0 \|z(f)-z(f')\|_1} \Big).\label{ekk.k2}
	\end{split}
\end{equation}
If we now note that 
\begin{equation}
		\delta(\bs z_0)
		\le
		\sum_{(f,f') \in T_{v_0}} \|z(f) - z(f')\|_1
		+
		\sum_{v \in S_{v_0}} \delta (\bs z_v)
		,\label{deltab}
	\end{equation}
and plug \cref{ekk.k2} into \eqref{ekk.k1}, we find
\begin{equation} 
	\begin{split}
		\|W_\infty[\allct;\tau,\ul P,\ul T,\ul 0]\|_{h_{v_0}}\le& 
		\frac{C^{\sum_{v\in S_{v_0}}|P_v|}}{|S_{v_0}|!} \Big(\sum_{z\in \Lambda_\infty}e^{-c_0\|z\|_1}\Big)^{|S_{v_0}|-1}
		\\ & \times
		\prod_{v \in S_{v_0}}
		\piecewise{
		\|(N_c)_{2}\|_{(c_0)}   & \text{if $v$ is of type}\ \tikzvertex{ctVertex}\ \text{and $v_0$ is undotted},\\
		\|(N_m)_{2}\|_{(c_0)}   & \text{if $v$ is of type}\ \tikzvertex{ctVertex}\ \text{and $v_0$ is undotted},\\
\|(V_\infty^{\rm int})_{|P_v|}\|_{(c_0)} & \text{if $v$ is of type}\ \tikzvertex{vertex},}
	\end{split}
		  \label{ekk-kk}
\end{equation}
which immediately implies the desired bound for $h_{v_0}=1$, because 
of the definition of $\epsilon_1$
and $ \|(V_\infty^{(1)})_{|P_v|}\|_{(c_0)}\le C^{|P_v|}|\lambda|^{\max\{1,\kappa|P_v|\}}$, see \eqref{eq:WL1bound} and recall the definition of $c_0$ at the beginning of this 
subsection.

\medskip

Next, we consider the case $h_{v_0}<1$, in which case $K_{v,\infty}$ is defined by \eqref{defKv4.1}. We proceed similarly: we start from \eqref{eq:W8def:1.2}  and use the analogue of \eqref{ekk.k0}-\eqref{ekk.k2}, namely
\begin{eqnarray} &&\big|\fG^{(h_{v_0})}_{T,\infty}(Q_1,\ldots, Q_s)\big|\le \Big(\prod_{(f,f')\in T} |g_{\ell,\infty}^{(h_{v_0})}| \Big)\, \Big(\prod_{f\in \cup_i Q_i\setminus T} |\gamma_{\omega(f),D(f),z(f)}^{(h_{v_0})}|\cdot |\tilde \gamma_{\omega(f),D(f),z(f)}^{(h_{v_0})}|\Big)^{1/2}
\nonumber\\
&&\qquad \le  (C2^{h_{v_0}})^{\frac12\sum_i |Q_i|} 2^{h_{v_0}\sum_{f\in \cup_i Q_i}\|D(f)\|_1}\Big(
\prod_{(f,f')\in T} e^{-c_0 2^{h_{v_0}} \|z(f)-z(f')\|_1} \Big) , \label{detb}
\end{eqnarray}
where in the second inequality we used again the bounds in items 1,4 of Prop.\ref{thm:g_decomposition}. Using \eqref{detb} and, again, \eqref{deltab}, we obtain the analogue of \eqref{ekk-kk}, 
\begin{eqnarray} 
&&\|W_\infty[\allct;\tau,\ul P,\ul T,\ul D]\|_{h_{v_0}}\le 
\frac{1}{|S_{v_0}|!} (C2^{h_{v_0}})^\frac{|Q_{v_0}|}2
2^{h_{v_0}\sum_{v\in S_{v_0}}\|\bs D_v|_{Q_{v_0}}\|_1}
\cdot \Big(\sum_{z\in \Lambda_\infty}e^{-\frac{c_0}2 2^{h_{v_0}}\|z\|_1}\Big)^{|S_{v_0}|-1}\cdot\nonumber\\
&&\cdot \prod_{v \in S_{v_0}}
		\piecewise{\|(\upsilon_{h_{v_0}}\cdot F_\infty)_{2,\|\bs D_v\|_1}\|_{h_{v_0}} 
		  	& \text{if $v\in V_e(\tau)$ is of type \tikzvertex{ctVertex} and $h_v = h_{v_0}+1$}\\ 
			\|(\cR_\infty N_c)_{2,\|\bs D_v\|_1}\|_{h_{v_0}} & \text{if $v\in V_e(\tau)$ is of type \tikzvertex{ctVertex} and $h_v \neq h_{v_0}+1$}\\ 
			|Z|^{-|\Psi_v|/2}\|(\cR_\infty V_\infty^{\rm int})_{|P_v|,\|\bs D_v\|_1}\|_{h_{v_0}} & \text{if $v\in V_e(\tau)$ is of type \tikzvertex{vertex}} \\
			\|(\lis W_\infty [\allct; \tau_v, \ul P_{v}, \ul T_{v}, \ul D_{v}])_{|P_v|,\|\bs D_v\|_1}\|_{h_{v_0}}
		  	& \text{if $v\in V_0(\tau)$,}
		}\label{basbih}
\end{eqnarray}
The terms in the second line can be bounded as follows: 
\begin{itemize}
	\item If $v\in V_e(\tau)$ is of type \tikzvertex{ctVertex} and $h_v = h_{v_0}+1$,
		\begin{eqnarray}\|(\upsilon_{h_{v_0}}\cdot F_\infty)_{2,\|\bs D_v\|_1}\|_{h_{v_0}}& \le& C \piecewise{2^{h_{v_0}}|\nu_{h_{v_0}}| & \text{if $\|\bs D_v\|_1=0$}\\
		\max\{|\zeta_{h_{v_0}}|, |\eta_{h_{v_0}}|\} & \text{if $\|\bs D_v\|_1=1$}}\\
		&\le& C' 2^{(h_{v}-1)(2-\frac{|P_v|}2-\|\bs D_v\|_1)}\epsilon_{h_v-1}\end{eqnarray}
		where we recall that $\epsilon_h=\max\{|\nu_h|,|\zeta_h|,|\eta_h|\}$ and, in passing from the first to the second line, we used the fact that $|P_v|=2$ and $h_{v_0}=h_{v}-1$, so that $2^{(h_{v}-1)(2-\frac{|P_v|}2-\|\bs D_v\|_1)}$ is equal to $2^{h_{v_0}}$, 
		if $\|\bs D_v\|_1=0$, and is equal to 1, if $\|\bs D_v\|_1=1$.
	\item If $v\in V_e(\tau)$ is of type \tikzvertex{ctVertex} and $h_v \neq h_{v_0}+1$,   using \cref{lem:R_bounds} and the definition of $\epsilon_1$ we have
		$$
		\|( \cR_\infty N_c)_{2,\|\bs D_v\|_1}\|_{h_{v_0}}
		\le 
		\|( \cR_\infty N_c)_{2,\|\bs D_v\|_1}\|_{0}
		\le 
		2 C \| N_c \|_1
		\le
		2 C \epsilon_1
		=
		C 2^{(h_{v}-1)(2-\frac{|P_v|}2-\|\bs D_v\|_1)}\epsilon_1,
		$$
		where the last equality holds trivially since we necessarily have $h_v=2$ and $|P_v|=2$, and $(\cR_\infty N_c)_{2,\|\bs D_v\|_1}$ vanishes unless $\|\bs D_v\|_1 = 2$.
	 \item If $v\in V_e(\tau)$ is of type \tikzvertex{vertex} (and, therefore, $h_v = 2$), then 
		$$\|(\cR_\infty V_\infty^{\rm int})_{|P_v|,\|\bs D_v\|_1}\|_{h_{v_0}}\le \|(\cR_\infty V_\infty^{\rm int})_{|P_v|,\|\bs D_v\|_1}\|_{0}\le C^{|P_v|}|\lambda|^{\max\{1,\kappa |P_v|\}},$$
		thanks to Lemma \ref{lem:R_bounds} and Eq.\eqref{eq:WL1bound}. 
		\item If $v\in V_0(\tau)$, recalling the definition of $\lis W_\infty[\allct; \tau_v, \ul P_{v}, \ul T_{v}, \ul D_{v}]$, see 
		\cref{lisWinfty}, and the bounds on the norm of $\cR_\infty$, see \eqref{ek.1}--\eqref{staz.ok}, we find
		\begin{eqnarray} && 
		\|(\lis W_\infty [\allct; \tau_v, \ul P_{v}, \ul T_{v}, \ul D_{v}])_{|P_v|,\|\bs D_v\|_1}\|_{h_{v_0}}\label{firsts}\\
		&&\le 
		C 2^{-h_v(\|\bs D_v\|_1-\|\bs D_v'\|_1)}\|(W_\infty [\allct; \tau_v, \ul P_{v}, \ul T_{v}, \ul D_{v}'])_{|P_v|,\|\bs D_v'\|_1}\|_{h_{v}}.\nonumber\end{eqnarray}
		Recalling the definition of $\bs D_v'$, that is $\bs D_v'=\bigoplus_{w\in S_v} \bs D_w\big|_{P_v}$, as well as the one of $R_v$, see the line before \eqref{eq:Rv_def}, we recognize that 
		$\|\bs D_v\|_1-\|\bs D_v'\|_1=R_v$. Note also that $W_\infty [\allct; \tau_v, \ul P_{v}, \ul T_{v}, \ul D_{v}']$ coincides with 
		its restriction $(W_\infty [\allct; \tau_v, \ul P_{v}, \ul T_{v}, \ul D_{v}'])_{|P_v|,\|\bs D_v'\|_1}$, so that the second line of \eqref{firsts}
		can be rewritten more compactly as $C 2^{-h_v R_v} \|W_\infty [\allct; \tau_v, \ul P_{v}, \ul T_{v}, \ul D_{v}']\|_{h_v}$. 
\end{itemize}
Plugging these bounds in \eqref{basbih}, noting that $\sum_{z\in \Lambda_\infty}e^{-\frac{c_0}2 2^{h_{v_0}}\|z\|_1}\le C 2^{-2h_{v_0}}$, we find
\begin{eqnarray} 
&&\|W_\infty[\allct;\tau,\ul P,\ul T,\ul D]\|_{h_{v_0}}\le 
\frac{1}{|S_{v_0}|!} C^{\frac{|Q_{v_0}|}2+|S_{v_0}|} 2^{h_{v_0}(\frac{|Q_{v_0}|}2+\sum_{v\in S_{v_0}}\|\bs D_v|_{Q_{v_0}}\|_1+2-2|S_{v_0}|)}\cdot\nonumber\\
&&\cdot \prod_{v \in S_{v_0}}
		\piecewise{2^{(h_{v}-1)(2-\frac{|P_v|}2-\|\bs D_v\|_1)}\epsilon_{h_{v}-1}& \text{if $v\in V_e(\tau)$ is of type \tikzvertex{ctVertex}}\\
			C^{|P_v|}|\lambda|^{\max\{1,\kappa |P_v|\}}& \text{if $v\in V_e(\tau)$ is of type \tikzvertex{vertex}} \\
			2^{-h_vR_v}\|W_\infty [\allct; \tau_v, \ul P_{v}, \ul T_{v}, \ul D_{v}']\|_{h_{v}}
		  	& \text{if $v\in V_0(\tau)$,}
		}\label{basbik}
\end{eqnarray}
Now, in the last line, for $v\in S_{v_0}\cap V_0(\tau)$, we iterate the bound, and we continue to do so until we reach all the endpoints. 
By doing so, recalling that $R_{v_0}=0$, we obtain the desired bound, \eqref{eq:W8_primitive_bound}, provided that 
\begin{equation}\label{easy} C^{\sum_{v\in V_0(\tau)}(\frac{|Q_v|}2+|S_v|)}\le (C')^{\sum_{v\in V_e(\tau)}|P_v|}.\end{equation}
In order to prove this, note that, for any dotted 
$v\in V_0(\tau)$, $|S_v|\le 1+\frac{|Q_v|}{2}\le |Q_v|$, because $|Q_v|\ge \max\{2,2(|S_v|-1)\}$; moreover, if $v_0$ is not dotted, $|S_{v_0}|=1$ and $|Q_{v_0}|=0$. Therefore, recalling also that $|Q_v|=\sum_{w\in S_v}|P_w|-|P_v|$, we find 
$$C^{\sum_{v\in V_0(\tau)}(\frac{|Q_v|}2+|S_v|)}\le C^{1+\frac32\sum_{v\in V_0}|Q_v|}\le C^{1+\frac32\sum_{v\in V_e(\tau)}|P_v|},$$
which implies \eqref{easy}.
\end{proof}

\medskip

Next, we rearrange \eqref{eq:W8_primitive_bound} in a different form, more suitable for summing over GN trees and their labels. 

\begin{lemma} Under the same assumptions as Proposition \ref{thm:W8:existence}, 
	\begin{align}
\label{eq:W8_reworked_bound}
& \|W_\infty[\allct;\tau,\ul P,\ul T,\ul D]\|_{h_{v_0}}\le     C^{\sum_{v\in V_e(\tau)}|P_v|}\frac1{|S_{v_0}|!}
2^{h_{v_0}d(P_{v_0}, \bs D_{v_0})}
\\
		& \qquad \times
		\Big(
	      \prod_{v \in V'(\tau)}\frac1{|S_{v}|!}
	      2^{(h_v-h_{v'})d(P_v,\bs D_v)}  \Big)\prod_{v \in V_e(\tau)}\piecewise{\epsilon_{h_v-1} & \text{if $v$ is of type \tikzvertex{ctVertex}}\\
	           |\lambda|^{\max\{1,\kappa|P_v|\}} & \text{if $v$ is of type \tikzvertex{vertex}}}
	  \end{align}
	where $d(P_v,\bs D_v)=2-\frac{|P_v|}{2}-\|\bs D_v\|_1$ is the scaling dimension of $v$, see \eqref{defscaldim}.
	\label{lm:W8:scaldim}
\end{lemma}

Note that, as observed in Remark \ref{remcruc}, the scaling dimensions appearing at exponent in the product over $v\in V'(\tau)$
are all negative, with the exception of the case that $v$ is an endpoint such that $h_{v'}=h_v-1$. Note, however, that in such a case
$2^{(h_v-h_{v'})d(P_v,\bs D_v)}\le 2$, which is a constant that can be reabsorbed in $ C^{\sum_{v\in V_e(\tau)}|P_v|}$, 
up to a redefinition of the constant $C$. 

\begin{proof} First note that, for all $v\in V_0(\tau)$, $|Q_v|=\sum_{w\in S_v}|P_w|-|P_v|$, so that, recalling that 
$R_v=\|\bs D_v\|_1-\sum_{w\in S_v}\|\bs D_w\big|_{P_v}\|_1$, we can rewrite the factor $\prod_{v \in V_0(\tau)}2^{(\tfrac12 |Q_v|+ \sum_{w\in S_v}\| \bs D_w |_{Q_v} \|_1 - R_v + 2 - 2 \left|S_v\right|)h_v}$ in \eqref{eq:W8_primitive_bound} as
\begin{eqnarray} &&
2^{h_{v_0}\sum_{v\in V_0(\tau)}(\frac12\sum_{w\in S_v}|P_w|-\frac12{|P_v|}+\sum_{w\in S_v}\|\bs D_w\|_1-\|\bs D_v\|_1-2(|S_{v}|-1))}\cdot\label{firstsecond}
\\
&& \cdot 2^{\sum_{v\in V_0(\tau)}(h_v-h_{v_0})(\frac12\sum_{w\in S_v}|P_w|-\frac12|P_v|+\sum_{w\in S_v}\|\bs D_w\|_1-\|\bs D_v\|_1-2(|S_{v}|-1))}.
\nonumber\end{eqnarray}
Now, the factor in the first line can be further rewritten by noting that: 
\begin{eqnarray} && 
\sum_{v\in V_0(\tau)}(\sum_{w\in S_v}|P_w|-|P_v|)=\sum_{v\in V_e(\tau)}|P_v|-|P_{v_0}|,\\
&&\sum_{v\in V_0(\tau)}(\sum_{w\in S_v}\|\bs D_v\|_1-\|\bs D_v\|_1)=\sum_{v\in V_e(\tau)}\|\bs D_v\|_1-\|\bs D_{v_0}\|_1,\\
&&\sum_{v\in V_0(\tau)}(|S_v|-1)=|V_e(\tau)|-1.\label{telsv}\end{eqnarray} 
(The first two identities are `obvious', due to the telescopic structure of the summand; the latter identity can be easily proved by induction.) Therefore, 
\begin{equation}
	\begin{split}
		2^{h_{v_0}\sum_{v\in V_0(\tau)}(\frac12\sum_{w\in S_v}|P_w|-\frac{|P_v|}2+\sum_{w\in S_v}\|\bs D_w\|_1-\|\bs D_v\|_1-2(|S_{v_0}|-1))}
		\label{4.1.77}
		\\ =2^{h_{v_0}(2-\frac{|P_{v_0}|}{2}-\|\bs D_{v_0}\|_1)}\Big(\prod_{v\in V_e(\tau)}2^{-h_{v_0}(2-\frac{|P_v|}{2}-\|\bs D_v\|_1)}\Big).
	\end{split}
\end{equation}
Similarly, the exponent of the factor in the second line of \eqref{firstsecond} can be rewritten as
\begin{eqnarray} && \sum_{v\in V_0(\tau)}\Big(\sum_{\substack{w\in V_0(\tau)\\ v_0<w\le v}}(h_w-h_{w'})\Big)\Big(\frac12\sum_{w\in S_v}|P_w|-\frac12|P_v|+\sum_{w\in S_v}\|\bs D_w\|_1-\|\bs D_v\|_1-2(|S_{v_0}|-1)\Big)\nonumber\\
&&=\sum_{\substack{w\in V_0(\tau)\\ w>v_0}}(h_w-h_{w'})\sum_{v\in V_0(\tau_w)}\Big(\frac12\sum_{w\in S_v}|P_w|-\frac12|P_v|+\sum_{w\in S_v}\|\bs D_w\|_1-\|\bs D_v\|_1-2(|S_{v_0}|-1)\Big)\nonumber\\
&&=\sum_{\substack{w\in V_0(\tau)\\ w>v_0}}(h_w-h_{w'})\Big(2-\frac{|P_w|}{2}-\|\bs D_w\|_1-\sum_{v\in V_e(\tau_w)}\big(2-\frac{|P_v|}{2}-\|\bs D_v\|_1\big)\Big).
\label{4.1.78}
\end{eqnarray}
Using \eqref{4.1.77} and \eqref{4.1.78}, and recalling that $2-\frac{|P_v|}{2}-\|\bs D_v\|_1\equiv d(P_v,\bs D_v)$, we rewrite \eqref{firstsecond} as
\begin{eqnarray}\eqref{firstsecond}&=&2^{h_{v_0}d(P_{v_0},\bs D_{v_0})}\Big(\prod_{v\in V_0'(\tau)}2^{(h_v-h_{v'})d(P_v,\bs D_v)}\Big)\Big(\prod_{v\in V_e(\tau)}2^{-(h_{v_0}+\sum_{w\in V_0(\tau)}^{v_0<w< v}(h_w-h_{w'}))d(P_v,\bs D_v)}\Big)\nonumber\\
&=&2^{h_{v_0}d(P_{v_0},\bs D_{v_0})}\Big(\prod_{v\in V_0'(\tau)}2^{(h_v-h_{v'})d(P_v,\bs D_v)}\Big)\Big(\prod_{v\in V_e(\tau)}2^{-h_{v'}d(P_v,\bs D_v)}\Big).
\label{eq:decay_reorg_middle}
\end{eqnarray}
By using this rewriting in \eqref{eq:W8_primitive_bound}, and noting that 
\begin{equation} 
	\begin{split}
		&\left(\prod_{v\in V_e(\tau)}2^{-h_{v'}d(P_v,\bs D_v)}\cdot \piecewise{
				2^{(h_v-1)(2 - \tfrac12 |P_v| - \left\|\bs D_v\right\|_1 )}\epsilon_{h_v-1} & \text{if $v$ is of type \tikzvertex{ctVertex}}\\
		|\lambda|^{\max\{1,\kappa|P_v|\}} & \text{if $v$ is of type \tikzvertex{vertex}}}\right) \\
		&\le\left( \prod_{v\in V_e(\tau)}2^{(h_v-h_{v'})d(P_v,\bs D_v)}\cdot \piecewise{
				\epsilon_{h_v-1} & \text{if $v$ is of type \tikzvertex{ctVertex}}\\
		2^{|P_v|} |\lambda|^{\max\{1,\kappa|P_v|\}} & \text{if $v$ is of type \tikzvertex{vertex}}}\right)
	\end{split}
	\label{eq:dacay_reorg_endpoints}
      \end{equation}
(since if $v$ is an endpoint of type \tikzvertex{vertex}, then $2^{h_v d(P_v,\bs D_v)}\ge 2^{-|P_v|}$), 
we readily obtain the desired estimate, \eqref{eq:W8_reworked_bound}.\end{proof}

As announced above, the bound \eqref{eq:W8_reworked_bound} is written in a form suitable for summing over the GN trees and their labels, as summarized in 
the following lemma. 

\begin{lemma} Under the same assumptions as Proposition \ref{thm:W8:existence}, for any $\vartheta \in (0,1)$, there exists $C_\vartheta>0$ such that, 
	letting $\cT^{(h)}_{\infty;(N,M)}$ denote the subset of $\cT^{(h)}_{\infty}$ whose trees have $N$ endpoints of type \tikzvertex{vertex} and $M$ endpoints of type \tikzvertex{ctVertex}, 
	\begin{equation} 
		\begin{split}
			2^{-\vartheta h}
			\sum_{\tau\in \cT^{(h)}_{\infty;(N,M)}} &
			\sum_{\substack{\ul P\in\cP(\tau)\\  |P_{v_0}|=n}}\
			\sum_{\ul T\in \cS(\tau,\ul P)}\ \sum_{\substack{\ul D\in\cD(\tau,\ul P)\\ \|\bs D_{v_0}\|_1=p}}\|W_\infty[\allct;\tau,\ul P,\ul T,\ul D]\|_{h+1}		
			\\ & \le 
			C_\vartheta^{N+M}|\lambda|^N
			\Big(\max_{h<h'\le 1} 2^{- \vartheta h'}\epsilon_{h'}\Big)^M 
			2^{h\cdot d(n,p)}.
		\end{split}
	\label{eq:treesum}
	\end{equation}
	\label{lem:treesum}
\end{lemma}

\begin{remark} {\bf (Short memory property).} 
	The fact that this estimate holds with the factors of $2^{-\vartheta h}$ included indicates that the contribution of trees covering a large range of scales are exponentially suppressed, a behavior referred to in previous works (e.g.\ \cite{GGM}) as the `short memory property'.
	As we shall see below, taking advantage of this there is a way to choose the free parameters $Z,t_1^*, \beta$ such that 
$|\epsilon_h|\le K_\vartheta |\lambda|2^{\vartheta h}$, see Propositions \ref{thm:ct} and \ref{thm:bare_to_renormalized} below. Under this condition, 
Lemma \ref{lem:treesum} implies that 
\begin{equation}
	\sum_{\tau\in \cT^{(h)}_{\infty;(N,M)}} \sum_{\substack{\ul P\in\cP(\tau)\\ |P_{v_0}|=n}}\
	\sum_{\ul T\in \cS(\tau,\ul P)}\ \sum_{\substack{\ul D\in\cD(\tau,\ul P)\\ \|\bs D_{v_0}\|_1=p}}\|W_\infty[\allct;\tau,\ul P,\ul T,\ul D]\|_{h+1}	
	\le
	C_\vartheta^{N+M}|\lambda|^{N+M}
	2^{h d(n,p)}2^{\vartheta h}
	,
\end{equation}
for all $\vartheta<1$ and $N+M>0$. The factor $2^{\vartheta h}$ in the right side is called the `short-memory factor'.

Summing over $N+M \ge 1$, this immediately implies
\begin{equation}
	\left\| \left( W_\infty \right)_{n,p}\right\|_{h+1}
	\le
	C_\vartheta
	|\lambda|
	2^{h d(n,p)}2^{\vartheta h}
	.
	\label{eq:W8_bound_final}
\end{equation}
\label{rem:short_memory}
\end{remark}

\begin{proof} Thanks to \eqref{eq:W8_reworked_bound}, the left side of \eqref{eq:treesum} can be bounded from above by 
\begin{equation} 
	\begin{split}
		2^{h\cdot d(n,p)} & |\lambda|^N
		\Big(\max_{h<h'\le 1} 2^{- \vartheta h'}\epsilon_{h'}\Big)^M 
		2^{-\vartheta h}
		\sum_{\tau\in \cT^{(h)}_{\infty;(N,M)}} 
		\sum_{\substack{\ul P\in\cP(\tau) \\ |P_{v_0}|=n}} C^{\sum_{v\in V_{ep}(\tau)}|P_v|} 
		\sum_{\ul T\in \cS(\tau,\ul P)}\ \sum_{\substack{\ul D\in\cD(\tau,\ul P) \ \|\bs D_{v_0}\|_1=p}} 
		\\ & \times \frac1{|S_{v_0}|!}\Big(
			\prod_{v \in V'(\tau)}\frac1{|S_{v}|!}
		2^{(h_v-h_{v'})d(P_v,\bs D_v)}  \Big)
		\prod_{\substack{v \in V_e(\tau) \\ \text{$v$ of type \tikzvertex{ctVertex}}}} 2^{\vartheta h_v}
		\prod_{\substack{v \in V_e(\tau) \\ \text{$v$ of type \tikzvertex{vertex}}}} |\lambda|^{[\kappa|P_v|-1]_+}
		,
	\end{split}
	  \label{eq:4.1.83}    
\end{equation}
where, in the last factor, $[\cdot]_+$ indicates the positive part. Note that, if either $v \in V_0'(\tau)$ or $v$ is an endpoint 
such that $h_{v'}<h_v-1$, then the scaling dimension of $v$ can be bounded 
uniformly in $\bs D_v$, i.e., $d(P_v,\bs D_v)\le \min\{-1, 2-|P_v|/2\}$; if $v$ is an endpoint such that $h_{v'}=h_v-1$, then 
the factor $2^{(h_v-h_{v'})d(P_v,\bs D_v)} $ is smaller than $2$ (and, therefore, it can be reabsorbed in $C^{\sum_{v\in V_{ep}(\tau)}|P_v|} $ up to a redefinition 
of the constant $C$). Moreover, 
recall that the number of elements of $\cD(\tau,\ul P)$ is bounded by $10^{|V(\tau)|}$,
see Remark \ref{rem:D_choice}. Finally, the number of elements of $\cS(\tau,\ul P)$ is bounded by 
\begin{equation}
	\label{numspantree}|\cS(\tau,\ul P)|\le C^{\sum_{v\in V_{ep}(\tau)}|P_v|}\prod_{v\in V_0(\tau)}|S_v|!, 
\end{equation}
see e.g.\ \cite[Lemma A.5]{GM01}. Therefore, putting these observations together, we see that \eqref{eq:4.1.83} can be bounded from above by 
\begin{eqnarray} 
&&  2^{h\cdot d(n,p)} 
|\lambda|^N
\Big(\max_{h<h'\le 1} 2^{- \vartheta h'}\epsilon_{h'}\Big)^M 
\sum_{\tau\in \cT^{(h)}_{\infty;(N,M)}} 2^{\vartheta(h^{max}_\tau-h)}\sum_{\ul P\in\cP(\tau)} (C')^{\sum_{v\in V_{ep}(\tau)}|P_v|} 
\times\nonumber\\
&&\times \Big(
	      \prod_{v \in V'(\tau)} 2^{(h_v-h_{v'})\min\{-1,2-\frac{|P_v|}{2}\}} \Big)
		\prod_{\substack{v \in V_e(\tau) \\ \text{$v$ of type \tikzvertex{ctVertex}}}} 2^{\vartheta h_v}
	      \prod_{\substack{v \in V_e(\tau)\\ \text{$v$ of type \tikzvertex{vertex}}}} |\lambda|^{[\kappa|P_v|-1]_+}.
\label{eq:4.1.84}    
\end{eqnarray}
We can cancel the factor of $2^{-\vartheta h}$ with a product of factors leading to an endpoint, and simplify the remaining bounds, to get
$$
	2^{- \vartheta h}
	\Big(
	\prod_{v \in V'(\tau)} 2^{(h_v-h_{v'})\min\{-1,2-\frac{|P_v|}{2}\}} \Big)
	\prod_{\substack{v \in V_e(\tau) \\ \text{$v$ of type \tikzvertex{ctVertex}}}} 2^{\vartheta h_v}
	\le 
	\Big(
	\prod_{v \in V'(\tau)} 2^{(h_v-h_{v'})(\vartheta-1)\frac{|P_v|}{6}} \Big)
	.
$$
This leaves a variety of exponential factors which make it possible to control the sum over $\ul P$, giving
$$\sum_{\tau\in \cT^{(h)}_{\infty;(N,M)}} \sum_{\ul P\in\cP(\tau)}
(C')^{\sum_{v\in V_{ep}(\tau)}|P_v|} 
\Big(
	      \prod_{v \in V'(\tau)} 2^{(h_v-h_{v'})(\vartheta-1)\frac{|P_v|}{6}} \Big)\prod_{\substack{v \in V_e(\tau)\\ \text{$v$ of type \tikzvertex{vertex}}}} |\lambda|^{[\kappa|P_v|-1]_+}\le (C'')^{N+M},$$
    	      see e.g. \cite[App.A.6.1]{GM01}, from which \eqref{eq:treesum} follows. 
\end{proof}
	
We conclude this subsection by noting that, in order for the right side of the bound \eqref{eq:treesum} to be summable over $N,M$ uniformly 
in $h$, we need that $\epsilon_h$ is bounded and small, uniformly in $h$. In view of Lemma \ref{lem:treesum}, this condition is sufficient for the whole sequence of kernels $V_\infty^{(h)}$, $h\le 1$, to be well defined. In the next subsection, we study the iterative definition of the 
running coupling constants, and prove that they in fact remain bounded and small, uniformly in $h$, provided that the counterterms $\upsilon_1$ 
are properly fixed.

\subsection{Beta function equation and choice of the counterterms}
\label{sec:fixed_point}

The definition of the running coupling constants, \eqref{rCC}, combined with the GN tree expansion for the effective potentials 
implies that the running coupling constants $\allct=\{(\nu_h,\zeta_h,\eta_h)\}_{h\le 0}$ satisfy the following equation, for all $h\le 0$: 
\begin{equation} 
	2^{h}\nu_h F_{\nu,\infty}+\zeta_h F_{\zeta,\infty}+\eta_h F_{\eta,\infty}
	=\sum_{\tau\in \cT^{(h)}_\infty}
	\cL_\infty W_\infty[\allct;\tau].
\end{equation}
More explicitly, using the definitions of $F_{\nu,\infty}$, $F_{\zeta,\infty}$, $F_{\eta,\infty}$ (see the lines after \eqref{1strem} and after \eqref{2ndrem}), for any 
$z\in\Lambda_\infty$, 
\begin{align} 
\nu_h &=2^{-h}(2\omega)\sum_{\tau\in \cT^{(h)}_\infty}
\cL_\infty W_\infty[\allct;\tau]
\big((\omega, 0,z),(-\omega,0,z)\big),\nonumber\\
\zeta_h & =4\omega\sum_{\tau\in \cT^{(h)}_\infty}
\cL_\infty W_\infty[\allct;\tau] 
\big((\omega, 0,z),(\omega,\hat e_1,z)\big),\label{betaf}\\
\eta_h &=4\sum_{\tau\in \cT^{(h)}_\infty}
\cL_\infty W_\infty[\allct;\tau] 
\big((\omega, 0,z),(-\omega,\hat e_2,z)\big).\nonumber\end{align}
In view of the iterative definition of the kernels $W_\infty[\allct;\tau]$, see 
\cref{Wtauexp,eq:W8def:1.2},
the right sides of these three equations can be naturally thought of as functions of $\allct$ or, better yet, of the restriction of $\allct$ to the scales larger than $h$. 
Therefore, these are recursive equations for the components of $\allct$: given $V_\infty^{(1)}$ (and, in particular, given $t_1^*,\beta, Z$), 
one can in principle construct the whole sequence $\allct$. 
More precisely, since the definition of the right sides of \cref{betaf} requires the summations over $\tau$ to be well defined, in view of \cref{lem:treesum}, 
these recursive equations 
allow one to construct the running coupling constants only for the scales $h$ such that $\max_{h'>h}\epsilon_{h'}$ is small enough. As we shall soon see, 
given a sufficiently small $\lambda$, the quantity $\max_{h'>h}\epsilon_{h'}$ stays small, uniformly in $h$, only for a special choice of 
the free parameters $t_1^*,\beta, Z$; to understand the appropriate choice of these parameters, 
it is helpful first to isolate those trees with only conterterm vertices, whose contribution in \eqref{betaf} is especilly simple:
\begin{enumerate}
	\item For $h < 0$, $\cT^{(h)}_{\infty;(0,1)}$ consists of two trees, one with a single counterterm endpoint $v$ on scale $h_v = 2 > h_{v_0}+1$, for which $\cL_\infty 
	W_\infty [\allct,\tau] = \cL_\infty \cR_\infty N_c = 0$, and one with a single counterterm endpoint on scale $h_v = h_{v_0}+1$, for which we have
		$$
		\cL_\infty W_\infty [\allct,\tau] =
		2^{h}\nu_h F_{\nu,\infty}+\zeta_h F_{\zeta,\infty}+\eta_h F_{\eta,\infty}
		,
		$$
	\item For $h=0$, $\cT^{(h)}_{\infty;(0,1)}$ consists of only a single tree $\tau$.
		This tree gives a contribution involving $N_c$ which can be calculated quite explicitly from \cref{eq:cV_1_def} and  \cref{eq:2.1.13} taking advantage of the fact that the latter is written in terms of the Fourier transformed fields:
		\begin{equation}\begin{split}
			\cL_\infty W_\infty [\allct,\tau] 
			&=
			\cL_\infty N_c
			\\
&=			\frac1{Z}\Big(t_2-\frac{1-t_1}{1+t_1}\Big) 
			F_{\nu,\infty}
			+\Big(
			\frac1{Z}\frac{t_1}{(1+t_1)^2}-\frac{t_1^*}{(1+t_1^*)^2} \Big)
			F_{\zeta,\infty}
			+\Big(
			\frac{t_2}{2Z}-\frac{t_2^*}2\Big) 
			F_{\eta,\infty}\end{split}
		\end{equation}
		(recall $t_2^* = (1- t_1^*)/(1+t_1^*)$)
		which we can put in the same form as the other scales by letting 
		\begin{equation}
			2 \nu_1 
			:=
			\frac1{Z}\Big(t_2-\frac{1-t_1}{1+t_1}\Big) 
			,
			\quad
			\zeta_1
			:=
			\frac1{Z}\frac{t_1}{(1+t_1)^2}-\frac{t_1^*}{(1+t_1^*)^2} 
			,
			\quad
			\eta_1
			:=
			\frac{t_2}{2Z}-\frac{t_2^*}2 
			.
			\label{eq:ct1_def}
		\end{equation}
\item The contributions from the trees in $\cT^{(h)}_{\infty,(0,M)}$ (i.e., those with no interaction endpoints and $M$ counterterm endpoints) vanishes
	for all $h \le 0$ and $M\ge 2$, as can be seen as follows.
In this case, each vertex must be assigned exactly 2 field labels, and for any endpoint $v$ the components of $\bs\omega_v$ must be both imaginary (corresponding to 
$N_m$, i.e.\ to $h_v = 2$, $h_{v'} = 1$) or both real (in all the other cases).  If any endpoints of the first type appear there is no allowed assignment of field labels to the corresponding tree (since the $\xi$ fields this case represents can neither be external fields nor be contracted with $\phi$ fields to form a connected diagram).
In any case, all of the internal field labels are included in spanning trees, and the local part of the contribution can be made quite explicit.
For example, denoting by $\big[\fg^{(h)}_{\infty}(z-z')\big]_{\omega,\omega'}$ the elements of the translation invariant infinite volume propagator $\fg^{(h)}_\infty(z-z')$, see
\eqref{eq:gN_def}, \eqref{eq:gh_def}, \eqref{ginftyeta}, 
the local part of the term with $\bs D_v = \bs 0$ for all vertices (which contributes to $\nu_h$) is a sum of terms proportional to 
 \begin{align} &
 \hskip-.3truecm 
\sum_{z_2,\dots,z_{M} \in \bZ^2}\big[
\fg_{\infty}^{(h_1)}(z_1-z_2)\big]_{-\omega_1,\omega_2} \big[\fg_{\infty}^{(h_2)}(z_2-z_3)\big]_{-\omega_2,\omega_3} \cdots 
\big[\fg_{\infty}^{(h_{M-1})}(z_{M-1}-z_M)\big]_{-\omega_{M-1},\omega_1}   \nonumber\\
&    
\qquad =  
\big[\hat \fg_{\infty}^{(h_1)}\big]_{-\omega_1,\omega_2}(0) 
\big[\hat \fg_{\infty}^{(h_2)}\big]_{-\omega_2,\omega_3}(0) \cdots \big[\hat \fg_{\infty}^{(h_{M-1})}\big]_{-\omega_{M-1},\omega_1}(0), 
	\label{eq:chain_props}
\end{align}
for some tuple $(\omega_1,\ldots,\omega_{M-1})\in\{+,-\}^{M-1}$; from \cref{ginftyeta,eq:gN_def,eq:gh_def}, 
the Fourier transform $\hat \fg^{(h)}$ in the right side satisfies
\begin{equation}
	\hat \fg^{(h)} (k)
	=
	\int_{2^{-2h-2}}^{2^{-2h}} 
	\hat \fg^{[\eta]}(k)\, \dd\eta
	,
	\quad
	\hat \fg^{[\eta]} (k)
	=
	e^{-\eta D(k)} D(k)
	\hat \fg (k)
\end{equation}
(except for $h=0$, where the lower limit of integration is 0), with $D$ defined in \cref{defDk1k2}, from which it is evident that $D(0) = 0$, and in consequence that the right hand 
side of \cref{eq:chain_props} is likewise zero.  The same argument holds, \emph{mutatis mutandis}, when $\bs D_v$ does not vanish for all vertices of the tree, with 
some discrete derivatives appearing in \cref{eq:chain_props}. 
\end{enumerate}
In view of these properties, we can write
\begin{equation}
\label{betafun}\nu_h=2\nu_{h+1}+B^\nu_{h+1}[\allct], \qquad \zeta_h=\zeta_{h+1}+B^\zeta_{h+1}[\allct], \qquad \eta_h=\eta_{h+1}+B^\eta_{h+1}[\allct], \end{equation}
for all $h \le 0$,
where the functions $B^\sharp_{h+1}[\allct]$ are given by the restrictions of the sums on the right-hand sides of \cref{betaf} to trees in some $\cT^{(h)}_{\infty,(M,N)}$ with $N \ge 1$.
The functions $B^\sharp_{h+1}[\allct]$, with $\sharp\in\{\nu,\zeta,\eta\}$, are called the components of the {\it beta function}, 
and \eqref{betafun} are called the {\it beta function flow equations} for the running coupling constants; note that, even if not explicitly indicated, the functions 
$B^\sharp_{h+1}[\allct]$, in addition to $\allct$, depend analytically upon $\lambda, t_1^*, \beta, Z$. For later reference, 
we let $B^\sharp_{h}[\allct;\tau]$ be the contribution to $B^\sharp_{h+1}[\allct]$ associated with the GN tree $\tau$. 
 In view of Lemma \ref{lem:treesum}, we find that, for any $|Z-1|\le 1/2$, any $|t_1^*|,|t_1|,|t_2|\in K'$, any 
$\vartheta\in (0,1)$ there exists $C_\vartheta>0$ such that, if $|\lambda|$ and $\max_{h'>h}\{\epsilon_{h'}\}$ are small enough, 
then 
\begin{equation}
\label{betafunb}\max_{\sharp\in\{\nu,\zeta,\eta\}}|B^\sharp_{h}[\allct]|\le C_\vartheta |\lambda| 2^{\vartheta h}, \qquad \forall \ h\le 0. \end{equation}

\begin{proposition} For any $\vartheta\in(0,1)$, there exist $K_\vartheta, \lambda_0(\vartheta)>0$ and functions $\nu_1(\lambda; t_1^*,\beta,Z)$, 
$\zeta_1(\lambda;  t_1^*,\beta,Z)$, $\eta_1(\lambda;  t_1^*,\beta,Z)$, $\allct(\lambda;  t_1^*,\beta,Z):=\{(\nu_h(\lambda;  t_1^*,\beta,Z),\zeta_h(\lambda;  t_1^*,\beta,Z),\eta_h(\lambda;  t_1^*,\beta,Z))\}_{h\le 0}$, analytic in $|\lambda|\le \lambda_0(\vartheta)$, $|Z-1|\le 1/2$, $|t_1^*|\in K'$ and $\beta$ such that 
$|t_1|, |t_2|\in K'$ (recall that $t_i=\tanh(\beta J_i)$), for all $\lambda, t_1^*, \beta,Z$ in such an analyticity domain:
  \begin{enumerate}
    \item $\nu_1(\lambda;  t_1^*,\beta,Z)$, $\zeta_1(\lambda;  t_1^*,\beta,Z)$, $\eta_1(\lambda;  t_1^*,\beta,Z)$ 
    and the components of $\allct(\lambda;  t_1^*,\beta,Z)$ satisfy \eqref{betafun}, for all $h\le 0$;
      \label{it:fp_beta}
    \item For all $h\le 1$, \begin{equation}
	\epsilon_h (\lambda;  t_1^*,\beta,Z):=\max\{|\nu_h(\lambda;  t_1^*,\beta,Z)|, |\zeta_h(\lambda;  t_1^*,\beta,Z)|, |\eta_h(\lambda;  t_1^*,\beta,Z)|\}
	\le	K_\vartheta 
	\left|\lambda\right|
	2^{\vartheta h}.
	\label{eq:ct:short_memory}
      \end{equation}
      \label{it:fp_short_memory}
  \end{enumerate}
  \label{thm:ct}
\end{proposition}
Note that, at this point, we do not prove that $\nu_1(\lambda; t_1^*,\beta,Z)$, 
$\zeta_1(\lambda; t_1^*,\beta,Z)$, $\eta_1(\lambda; t_1^*,\beta,Z)$ satisfy \cref{eq:ct1_def}; this comes later.

\begin{proof}
For simplicity, we do not track the dependence of the constants and of the norms upon $\vartheta$ which, 
within the course of this proof, we assume to be a fixed constant in $(0,1)$. 
In order to construct $\tilde\allct(\lambda; t_1^*,\beta,Z):=\{(\nu_h(\lambda; t_1^*,\beta,Z), 
\zeta_h(\lambda; t_1^*,\beta,Z), \eta_h(\lambda; t_1^*,\beta,Z))\}_{h\le 1}$, 
we first note that the equations for $\nu_{h}, \zeta_{h}, \eta_{h}$ in \eqref{betafun} imply that, for $k<h\le 1$,
$\nu_{h}=2^{k-h}\nu_{k}-\sum_{k<j\le h}2^{j-h-1}B^\nu_{j}[\allct]$,
$\zeta_{h}=\zeta_{k}-\sum_{k<j\le h}B^\zeta_{j}[\allct]$, and $\eta_{h}=\eta_k-\sum_{k<j\le h}B^\eta_{j}[\allct]$.
If we send $k\to-\infty$ and impose that $\epsilon_k:=\max\{|\nu_k|,|\zeta_k|,|\eta_k|\}\to 0$ as $k\to-\infty$, we get 
  \begin{equation}
\begin{cases} \nu_{h}=-\sum_{j\le h}2^{j-h-1}B^\nu_{j}[\allct],\\
 \zeta_{h}=-\sum_{j\le h}B^\zeta_{j}[\allct],\\
 \eta_{h}=-\sum_{j\le h}B^\eta_{j}[\allct],\end{cases}\label{eq:sys}
  \end{equation}
which we regard as a fixed point equation $\tilde\allct=\bs T[\tilde\allct]$ for a map $\bs T$ on the space of sequences $X_\varepsilon:=\{\tilde\allct=\{(\nu_h,\zeta_h,\eta_h)\}_{h\le 1}\,:\, \|\tilde\allct\|\le \varepsilon\}$, with $\|\tilde\allct\|=\sup_{h\le 1} \{2^{-\vartheta h} \epsilon_h\}$ and $\varepsilon$ a sufficiently small constant. 

We now intend to prove that $\bs T$ is a contraction on $X_\varepsilon$ and, more precisely, that: (1) the image of $X_\varepsilon$ under the action of $\bs T$ is contained in $X_\varepsilon$; (2) $\| \bs T[\tilde\allct]-\bs T[\tilde\allct']\|\le (1/2)\, \|\tilde\allct -\tilde\allct'\|$ for all $\tilde\allct,\tilde\allct'\in X_\varepsilon$.
Once $\bs T$ is proved to be a contraction, it follows that it admits a unique fixed point in $X_\varepsilon$, which corresponds to the desired sequence $\tilde\allct(\lambda;t_1^*,\beta,Z)$. The analyticity of $\tilde\allct(\lambda;t_1^*,\beta,Z)$ follows from the analyticity of the components of the beta function with respect to $\lambda,t_1^*,\beta,Z$ and $\allct$ that, in turn, follows from the absolute summability of its tree expansion, which is a power series in $\lambda$, with coefficients analytically depending upon $t_1^*,t_1,t_2,Z$, with $t_i=\tanh(\beta J_i)$. 

The fact that the image of $X_\varepsilon$ under the action of $\bs T$ is contained in $X_\varepsilon$ follows immediately from \cref{betafunb}. 
In order to prove that  $\| \bs T[\tilde\allct]-\bs T[\tilde\allct']\|\le (1/2)\, \|\tilde\allct -\tilde\allct'\|$, we rewrite the $\nu$-component of 
$\bs T[\tilde\allct]-\bs T[\tilde\allct']$ at scale $h$ as a linear interpolation
\begin{equation}
	\label{21_9:07}
	\nu_h-\nu_{h'}=-\sum_{j\le h} \sum_{\tau\in \cT^{(j-1)}_\infty} 2^{j-h-1}\int_0^1 \frac{d}{dt} B^\nu_j[\tilde\allct(t);\tau]\, dt,
\end{equation}
where $\tilde\allct(t)=\tilde\allct'+t(\tilde\allct-\tilde\allct')$, and similarly for the $\zeta$- and $\eta$-components. When the derivative with respect to $t$ acts on 
the tree value $B^\nu_j[\allct(t);\tau]$, it has the effect of replacing one of the factors $\nu_h(t)$, or $\zeta_h(t)$, or $\eta_h(t)$, associated with 
one of the counterterm endpoints, by $\frac{d}{dt}\nu_h(t)=\nu_h-\nu_h'$, or $\frac{d}{dt}\zeta_h(t)=\zeta_h-\zeta_h'$, or $\frac{d}{dt}\eta_h(t)=\eta_h-\eta_h'$, 
respectively. Therefore, we get the analogue of \cref{betafunb}: if $|\lambda|$ and $\|\tilde\allct(t)\|$ are sufficiently small, then 
\begin{equation}
\label{betafunbt}
\max_{\sharp\in\{\nu,\zeta,\eta\}}\sum_{\tau\in\cT^{(h-1)}_\infty}\Big|\frac{d}{dt}B^\sharp_{h}[\allct(t);\tau]\Big|\le C |\lambda| 2^{\vartheta h}
\max_{h'\ge h}\{|\nu_{h'}-\nu_{h'}'|, |\zeta_{h'}-\zeta_{h'}'|, |\eta_{h'}-\eta_{h'}'|\}
\end{equation}
for all $h \le 1$.
Plugging this estimate into \cref{21_9:07} and its analogues for $\zeta_h-\zeta_h'$ and $\eta_h-\eta_h'$, we readily obtain the desired estimate, 
$\| \bs T[\tilde\allct]-\bs T[\tilde\allct']\|\le (1/2)\, \|\tilde\allct -\tilde\allct'\|$, for $\lambda_0$ and $\varepsilon$ sufficiently small. \end{proof}

We now need to show that it is possible to choose the free parameters $\beta, Z, t_1^*$ in such a way that the functions $\nu_1(\lambda; t_1^*,\beta,Z)$, $\zeta_1(\lambda; t_1^*,\beta,Z)$, $\eta_1(\lambda; t_1^*,\beta,Z)$ constructed in \cref{thm:ct} satisfy \cref{eq:ct1_def};
that is, for given $J_1$, $J_2$ and $\lambda$, 
there exists a critical value of the inverse temperature $\beta$ for which the above expansion for the kernels of the infinite plane effective potentials
is convergent, with dressed parameters $t_1^*$ and $Z$. The desired result is summarized in the following proposition. 

\begin{proposition} For any $J_1,J_2$ satisfying the conditions of \cref{prop:main}, and any $\vartheta\in(0,1)$, there exist $\lambda_0(\vartheta)>0$ 
and functions $t_1^*(\lambda)$, $\beta_c(\lambda)$, $Z(\lambda)$, 
analytic in $|\lambda|\le \lambda_0(\vartheta)$, such that \cref{eq:ct1_def} holds, with $t_1=\tanh(\beta_c(\lambda)J_1)$, $t_2=\tanh(\beta_c(\lambda)J_2)$, and $(\nu_1,\zeta_1,\eta_1)=((\tilde\nu_1(\lambda),\tilde\zeta_1(\lambda), \tilde\eta_1(\lambda))$, 
with $\tilde\nu_1(\lambda)=\nu_1(\lambda;t_1^*(\lambda),$ $\beta_c(\lambda),Z(\lambda))$ (here $\nu_1(\lambda;t_1^*,\beta,Z)$ is the same as in Proposition \ref{thm:ct}), 
and similarly for $\tilde\zeta_1(\lambda)$ and $\tilde\eta_1(\lambda)$.  
Correspondingly, the flow of running coupling constants with initial datum 
$(\tilde\nu_1(\lambda),\tilde\zeta_1(\lambda), \tilde\eta_1(\lambda))$, generated by the flow equations \eqref{betafun}, is well defined for all $h\le 0$ 
and satisfies \eqref{eq:ct:short_memory}.	\label{thm:bare_to_renormalized}
\end{proposition}

\begin{proof} The result is a direct consequence of the analytic implicit function theorem: we intend to fix $t_1^*=t_1^*(\lambda)$, $\beta=\beta_c(\lambda)$, $Z=Z(\lambda)$ 
in such a way that \cref{eq:ct1_def} holds, with $\nu_1=\nu_1(\lambda;t_1^*,\beta,Z)$, $\zeta_1=\zeta_1(\lambda;t_1^*,\beta,Z)$, $\eta_1=\eta_1(\lambda;t_1^*,\beta,Z)$, the 
same functions as in Proposition \ref{thm:ct}. With this in mind, we recast the system of equations \eqref{eq:ct1_def} in the following form (recall once more that $t_i=\tanh(\beta J_i)$ and $t_2^*=(1-t_1^*)/(1+t_1^*)$): 
	\begin{align}
		2 Z \nu_1 (\lambda;t_1^*,\beta,Z)
		- 
		\tanh \beta J_2
		+
		e^{-2 \beta J_1}
		& = 0
		,
		\label{eq:implicit_nu}
		\\
		4 Z \zeta_1 (\lambda;t_1^*,\beta,Z)
		+
		\frac{4 Z t_1^*}{(1 + t_1^*)^2}
		+
		e^{- 4 \beta J_1}
		- 1
		& = 0
		,
		\label{eq:implicit_zeta}
		\\
		2Z \eta_1 (\lambda;t_1^*,\beta,Z)
		+
		Z \frac{1- t_1^*}{1+t_1^*}
		-
		\tanh \beta J_2
		& = 0.
		\label{eq:implicit_eta}
	\end{align}
Note that, by Proposition \ref{thm:ct}, everything appearing in \cref{eq:implicit_nu,eq:implicit_zeta,eq:implicit_eta} is analytic, so if we succeed in showing that 
the implicit function theorem applies, the resulting solution will be analytic in $\lambda$, as desired. Note that, at $\lambda=0$, we have 
$\nu_1 (0;t_1^*,\beta,Z) = \zeta (0;t_1^*,\beta,Z) =\eta_1 (0;t_1^*,\beta,Z) = 0$, so in this case the system \eqref{eq:implicit_nu}--\eqref{eq:implicit_eta}
is solved by $\beta = \beta_c (J_1,J_2)$ (with $\beta_c(J_1,J_2)$ the critical temperature of the nearest neighbor model, see the lines after \eqref{eq:truncated}), 
$t_1^* = \tanh \big(\beta_c (J_1,J_2) J_1\big)$ and $Z=1$. Note also that the partial derivatives of $\nu_1(\lambda;t_1^*,\beta,Z),\zeta_1(\lambda;t_1^*,\beta,Z),\eta_1(\lambda;t_1^*,\beta,Z)$ with respect to $t_1^*,\beta$ and $Z$ vanish at $\lambda=0$, so the determinant of the Jacobian of the system \eqref{eq:implicit_nu}--\eqref{eq:implicit_eta}
with respect to $t_1^*,\beta,Z$, computed at $\lambda=0$ and $(t_1^*,\beta,Z)=(\tanh\big(\beta_c (J_1,J_2) J_1\big), 
\beta_c(J_1,J_2),1)$, equals 
	\begin{equation}
		\begin{vmatrix}
			0&- J_2\sech^2(\beta_c J_2)
			- 2 J_1 e^{-2 \beta_c J_1}
			&
			0
			\\
			4 \frac{1- t_1^*}{(1+t_1^*)^3}
			&
			-4 J_1 e^{-4 \beta_c J_2}
			&
			\frac{4 t_1^*}{(1 + t_1^*)^2}
			\\
			-\frac{2}{(1+t_1^*)^2}
			&
			- J_2 \sech^2(\beta_c J_2)
			&
			\frac{1 - t_1^*}{1+t_1^*}
		\end{vmatrix}
		=
		-
		4
		(
			J_2 \sech^2(\beta_c J_2) 
			+ 
			2 J_1 e^{-2 \beta_c J_1}
		)
		\frac{1 + (t_1^*)^2}{(1 + t_1^*)^4}
	\end{equation}
	with $\beta_c = \beta_c (J_1,J_2)$ and $t_1^* = \tanh(\beta_c (J_1,J_2) J_1)$;  the right hand side is evidently nonzero, and, therefore, the analytic implicit function theorem applies, implying the desired claim. 
\end{proof}

This concludes the construction of the sequence of effective potentials in the infinite volume limit, uniformly in the scale label, 
with optimal bounds on the speed at which, after proper rescaling, such effective potentials go to zero as $h\to-\infty$. This result, and the methods introduced 
to prove it, is a key ingredient in the proof of Theorem \ref{prop:main}, for whose completion we refer the reader to \cite{AGG_part1}. 

\appendix

\section{Diagonalization of the matrix $A_c$}
\label{app:diagonalization}
In this section we compute the 
propagator of $\phi$
and the Gaussian integral associated with $\cS_c$.
For this purpose, we first need to block diagonalize the 
coefficient matrix, which we do by using a 
transformation which can be thought of as a Fourier sine transformation with modified frequencies. We write  
\begin{equation}  
\label{eq:cS_c}
\begin{split}
	\cS_c(\phi)=& \frac1{2L}\sum_{k_1\in \mathcal D_L} \Big( \begin{bmatrix}\phi_+(-k_1) \\ \phi_-(-k_1)\end{bmatrix},\ \begin{bmatrix} -i\Delta(k_1) & - b(k_1) +t_2\tau \\
	 b(k_1)-t_2\tau^T &  i\Delta(k_1)\end{bmatrix} \begin{bmatrix} \phi_+(k_1) \\ \phi_-(k_1)\end{bmatrix}\Big)\\
	=:&  \frac1{2L}\sum_{k_1\in \mathcal D_L}\Big( \begin{bmatrix}\phi_+(-k_1) \\ \phi_-(-k_1)\end{bmatrix},\ \tilde A_c(k_1) \begin{bmatrix} \phi_+(k_1) \\ \phi_-(k_1)\end{bmatrix}\Big),
\end{split}
\end{equation}
where $\phi_\omega(k_1)$, with $\omega=\pm$, is the column vector whose components are $\phi_{\omega,z_2}(k_1)$ with $z_2=1,\ldots,M$, and $\tau$ is the $M\times M$ shift matrix $\tau_{z_2,z'_2}:=\delta_{z_2+1,z_2'}$, that is, 
$$\tau = \begin{bmatrix} 0 & 1 & 0 \\
 0 & 0  & 1 & \ddots \\
0 & \ddots & \ddots  & \ddots & \ddots \\
&\ddots & \ddots & \ddots & 1 &0\\
&& \ddots  & 0 & 0  & 1\\
&&& 0 & 0 & 0 \end{bmatrix}. $$
For brevity, we will write $\tilde A_c = \tilde A_c(k_1)$, $\Delta = \Delta (k_1)$, $ b =  b(k_1)$ since dependence on $k_1$ plays no role in the next several pages.
It is helpful to begin by diagonalizing the real symmetric matrix
$$\tilde A_c^2 = \begin{bmatrix} \tilde B_c^+ & 0 \\ 0 & \tilde B_c^- \end{bmatrix},$$
where $\tilde B_c^+$ is the $M\times M$ tri-diagonal matrix
$$\tilde B_c^+=\begin{bmatrix}
-\Delta^2 - b^2-t_2^2 &   bt_2 \\
   bt_2 & -\Delta^2 - b^2-t_2^2 &   b t_2 \\
 &  bt_2 & -\Delta^2 - b^2-t_2^2 &  \ddots \\
& & \ddots & \ddots  & \ddots &  \\
&&&  bt_2  & -\Delta^2 - b^2 -t_2^2 &  b t_2\\
&&&&  bt_2 & -\Delta^2 - b^2
\end{bmatrix}.
$$
Note that all the diagonal entries are equal to $-\Delta^2 - b^2-t_2^2$ apart from the last one, which equals $-\Delta^2 - b^2$. The block $\tilde B^-_c$ is obtained from $\tilde B^+_c$
by reversing the order of the rows and columns.
$\tilde B_c^\omega$, with $\omega=\pm$, can each be thought of as a discrete Laplacian with mixed boundary conditions, which suggests the ansatz
\begin{equation}
v_{k_2}=\begin{bmatrix}
\alpha_{k_2} e^{ik_2} + \beta_{k_2} e^{-ik_2} \\
\alpha_{k_2} e^{i2k_2} + \beta_{k_2} e^{-i2k_2} \\
\vdots
\end{bmatrix},\label{eq:vdef}
\end{equation}
for their eigenvectors. In fact, we see that $v_{k_2}$ is an eigenvector of $\tilde B_c^+$ iff the system of equations
\begin{gather}
(-\Delta^2- b^2-t_2^2)(\alpha_{k_2} e^{ik_2} + \beta_{k_2} e^{-ik_2}) +  bt_2 (\alpha_{k_2} e^{i2k_2}+ \beta_{k_2} e^{-i2k_2}) = \lambda_{k_2} (\alpha_{k_2} e^{ik_2} + \beta_{k_2} e^{-ik_2})\\[2ex]
\begin{split}
  bt_2 (\alpha_{k_2} e^{i(z_2-1)k_2} +&\beta_{k_2} e^{-ik_2(z_2-1)})+ (-\Delta^2- b^2-t_2^2)(\alpha_{k_2} e^{ik_2z_2} + \beta_{k_2} e^{-ik_2z_2})  \\
  +  bt_2 &(\alpha_{k_2} e^{ik_2(z_2+1)} + \beta_{k_2} e^{-ik_2(z_2+1)})  \\
   &= \lambda_{k_2} (\alpha_{k_2} e^{ik_2z_2} + \beta_{k_2} e^{-ik_2z_2}), \qquad  1<z_2<M \label{eq:bulk_eigen}
\end{split}\\[2ex]
\begin{split}
  bt_2 (\alpha_{k_2} e^{ik_2(M-1)} +\beta_{k_2} e^{-ik_2(M-1)})+ (-\Delta^2- b^2)(\alpha_{k_2} e^{ik_2M} + \beta_{k_2} e^{-ik_2M})  \\ = \lambda_{k_2} (\alpha_{k_2} e^{ik_2M} + \beta_{k_2} e^{-ik_2M})
\end{split}
\end{gather}
are all satisfied. Equation \eqref{eq:bulk_eigen} is solved by choosing $$\lambda_{k_2} =  bt_2(e^{ik_2}+e^{-ik_2}) + (-\Delta^2- b ^2-t_2^2),$$ which reduces the other two conditions to 
\begin{gather}
 bt_2  (\alpha_{k_2}  + \beta_{k_2} ) = 0 ,\\
 bt_2 (\alpha_{k_2} e^{ik_2(M+1)}+\beta_{k_2} e^{-ik_2(M+1)}) - t_2^2 (\alpha_{k_2} e^{ik_2M}+ \beta_{k_2} e^{-ik_2M}) = 0 . \label{eq:bd_eig}
\end{gather}
The first condition implies $\beta_{k_2} = -\alpha_{k_2}$, 
which can be used to rewrite Equation~\eqref{eq:bd_eig} as
\begin{equation}
	\sin k_2(M+1)=B(k_1)\sin k_2M
\label{eq:q_condition_recall}
\end{equation}
where for brevity we have introduced
\begin{equation}
B(k_1):= \frac{t_2}{ b(k_1)}=t_2\frac{|1+t_1 e^{ik_1}|^2}{1-t_1^2}, 
\end{equation}
(cf.\ Equation~\eqref{eq:B_def}).
We call $\mathcal{Q}_M^+(k_1)$ the set of the solutions of \eqref{eq:q_condition} with $\Re k_2 \in [0,\pi]$. 

Restricting to the critical case \eqref{eq:anisotropic_critical_condition},
\begin{equation}
  B(k_1) = 1- \frac{2 t_1t_2}{1-t_1^2}(1- \cos k_1) =: 1-\kappa (1-\cos k_1) \label{e:Bk}
  ,
\end{equation}
{so that} $0{<} B(k_1){<} 1$ {for $k_1\in \mathcal D_L$}.
Equation~\eqref{eq:q_condition} is equivalent to 
\begin{equation}
	\tan k_2 (M+1)
	=
	\frac{B(k_1) \sin k_2}{B(k_1) \cos k_2 - 1}
	,
\end{equation}
which for ${0<}B(k_1)<1$ has a unique real solution in each interval $I_n:=\frac{\pi}{M+1}(n+\frac12,n+1)$, $n=0,\ldots,M-1$, 
since the left hand side increases monotonically from $-\infty$ to $0$ while the right hand side is negative and decreasing.
Thus all $M$ eigenvectors of $\tilde B_c^+$ (and, as a consequence, of $\tilde B^-_c$) correspond to real solutions of this form
by
\begin{equation}
	u_{k_2}^+ = \sqrt{\frac{2}{N_M(k_1, k_2)}}
\begin{bmatrix} \sin k_2 \\ \sin 2k_2 \\ \vdots \\ \sin k_2M \end{bmatrix}\qquad ({\text{and, respectively,}} \quad  u_{k_2}^- = \sqrt{\frac{2}{N_M(k_1, k_2)}} 
\begin{bmatrix} \sin k_2M \\ \sin k_2(M-1) \\ \vdots \\ \sin k_2 \end{bmatrix}),
\end{equation}
where 
\begin{equation}
	N_M(k_1,k_2)
	:=
	2
	\sum_{x=1}^M
	\sin^2 k_2x
	=
	M+\frac{1}{{2}}
	-
	\frac12 \frac{\sin (2M+1)k_2}{\sin k_2}
\label{eq:ell_def}
\end{equation}
so that the eigenvectors are normalized.\\
To obtain Equation~\eqref{eq:N_M_def}, we note that since $q$ satisfies~\eqref{eq:q_condition_recall}, 
we have 
\begin{equation*}
\begin{split}
\sin \left[ k_2(M+1) \pm k_2M \right] &= \sin k_2(M+1) \cos k_2 M \pm \sin k_2M \cos k_2(M+1)\\ &= \left[ B(k_1) \cos k_2 M \pm \cos k_2 (M+1) \right] \sin k_2M,
\end{split}
\end{equation*}
and so we can rewrite $N_M(k_1, k_2)$ as
\begin{equation}\label{eq:N_M}
	N_M(k_1, k_2)=\frac{B(k_1)M\cos k_2M -(M+1)\cos k_2(M+1)}{B(k_1)\cos k_2M-\cos k_2(M+1)}=\frac{\frac{d}{dk_2}\left(B(k_1)\sin k_2M-\sin k_2(M+1) \right)}{B(k_1)\cos k_2M-\cos k_2(M+1)}.
\end{equation}

We now return to $\tilde A_c$. 
Equation~\eqref{eq:q_condition} is equivalent to 
\begin{equation}
 b e^{ik_2(M+1)} - t_2 e^{ik_2M} = be^{-ik_2(M+1)} - t_2e^{-ik_2M}
,
\label{eq:symmetry_g_barvv}
\end{equation}
and therefore
 $$ ( b-t_2\tau^T)u^+_{k_2}=-( b e^{ik_2(M+1)}-t_2e^{ik_2M})u^-_{k_2},\qquad \text{and}\qquad  (- b+t_2\tau)u^-_{k_2}=( b e^{ik_2(M+1)}-t_2e^{ik_2M})u^+_{k_2}$$
 whenever $k_2 \in \cQ_M(k_1)$.
 Combining this with the definition of $\tilde A_c(k_1)$ in Equation~\eqref{eq:cS_c}, we see that
\begin{equation}
\tilde A_c \begin{bmatrix} u_{k_2}^+ \\ 0\end{bmatrix} =\begin{bmatrix} -i\Delta u^+_{k_2} \\ -( b e^{ik_2(M+1)}-t_2e^{ik_2M})u^-_{k_2}\end{bmatrix},\qquad 
\text{and}\qquad \tilde A_c \begin{bmatrix} 0 \\ u^-_{k_2}\end{bmatrix} =\begin{bmatrix} ( b e^{ik_2(M+1)}-t_2e^{ik_2M})u^+_{k_2} \\ i\Delta u^-_{k_2} \end{bmatrix},
\end{equation}
or in other words, the change of variables induced by $\begin{bmatrix} u_{k_2}^+ \\ 0 \end{bmatrix}$ and $\begin{bmatrix}  0 \\ u_{k_2}^- \end{bmatrix}$
puts $\tilde A_c$ in block-diagonal form, with $2\times 2$ blocks
\begin{equation}
	\begin{split}
		& \tilde {\mathfrak g}^{-1}(k_1, k_2)
		:=
		\begin{bmatrix}
			-i\Delta & e^{-ik_2(M+1)} ( b - t_2 e^{ik_2}) \\
			-e^{ik_2(M+1)} ( b-t_2e^{-ik_2}) & i\Delta
		\end{bmatrix}
		\\ & \quad = 
		\frac1{\left|1 + t_1 e^{ik_1}\right|^2}
		\begin{bmatrix}
			-2 i t_1 \sin  k_1 &
			e^{-ik_2(M+1)} (1- t_1^2) (1 - B(k_1) e^{ik_2}) \\
			- e^{ik_2(M+1)} (1- t_1^2) (1 - B(k_1) e^{-ik_2}) &
			2 i t_1 \sin k_1
		\end{bmatrix},
	\end{split}
\label{eq:blocks}
\end{equation}
recalling the definitions~\eqref{eq:b_def} and~\eqref{eq:B_def}.
This block-diagonalization implies that 
\begin{equation} {\rm Pf} A_c= \prod_{k_1\in\mathcal D_L}\prod_{k_2\in\mathcal Q^+_M(k_1)} \sqrt{\det \tilde {\mathfrak g}^{-1}(k_1, k_2)},
\end{equation}
where explicitly, using the criticality condition \eqref{eq:anisotropic_critical_condition},
\begin{equation} 
	\det \tilde {\mathfrak g}^{-1}(k_1, k_2)
	=
	\frac{
	2(1 - t_2)^2(1 - \cos k_1)+2(1 - t_1)^2(1 - \cos k_2)
	}{ \left|1 + t_1 e^{ik_1}\right|^2}.
	\label{eqD}
\end{equation}
Note that this determinant vanishes iff $k_1=k_2=0$ mod $2\pi$ {(in particular, it is positive if $k_1\in \mathcal D_L$)}. 

Concerning the propagator, denoting the 
inverse of \eqref{eq:blocks} by
\begin{equation}
	\begin{split}
		\tilde {\mathfrak g} (k_1, k_2)
		:=&
		\begin{bmatrix}
			\tilde {g}_{++} (k_1, k_2) & 
			\tilde {g}_{+-}(k_1, k_2) \\
			\tilde {g}_{-+}(k_1, k_2) &
			\tilde {g}_{--}(k_1, k_2) 
		\end{bmatrix}
		\\
		:= &
		\frac{1}{D(k_1, k_2)}
		\begin{bmatrix}
			 2 i t_1 \sin k_1 &
			- e^{-ik_2(M+1)} (1- t_1^2) (1 - B(k_1) e^{ik_2}) \\
			e^{ik_2(M+1)} (1- t_1^2) (1 - B(k_1) e^{-ik_2}) &
			 -2 i t_1 \sin k_1 
		\end{bmatrix}
		,
	\end{split}
\end{equation}
where $D(k_1, k_2)$ is defined as in \eqref{defDk1k2} 
and, letting $$\tilde \phi_{k_2,\omega}(k_1):=\sum_{z_2=1}^M \phi_{\omega,x_2}(k_1) u^\omega_{k_2}(z_2),$$ we have, for $k_1,k_1'\in\mathcal D_L$ , 
$k_2,k_2'\in\mathcal Q^+_M$, and $\omega,\omega'\in\{\pm\}$,
\begin{equation} \langle \tilde \phi_{k_2,\omega}(k_1)\tilde \phi_{k_2',\omega'}(k_1')\rangle=-L\delta_{k_1,-k_1'}\delta_{k_2,k_2'}\tilde {g}_{\omega\omega'} (k_1,k_2),
\end{equation}
so that, in terms of $\phi_{\omega,z}=\frac1{L}\sum_{k_1\in \mathcal D_L}e^{-ik_1 z_1} \phi_{\omega,z_2}(k_1)$, 
\begin{equation} \langle \phi_{\omega,z} \phi_{\omega',z'}\rangle=-\frac1{L}\sum_{k_1\in\mathcal D_L}\sum_{k_2\in\mathcal Q_M^+(k_1)}\tilde {g}_{\omega\omega'} (k_1, k_2) 
e^{-ik_1(z_1-z_1')}u^\omega_{k_2}(z_2)u^{\omega'}_{k_2}(z_2')\equiv g_{\omega\omega'}(z, z').\label{prop:before}
\end{equation}
Then recalling the definition of $u^{\pm}_{k_2} (z_2)$ {and the identity \eqref{eq:symmetry_g_barvv}}, Equations~\eqref{eq:g_cyl_matrix} and~\eqref{eq:g_finite_cyl} 
follow by writing out the sines in terms of complex exponentials and relabelling the sum in terms of $\cQ_M (k_1) := \cQ_M^+(k_1) \sqcup (- \cQ_M^+(k_1))$.

\section{Proof of Proposition~\ref{thm:g_decomposition}}\label{sec:proofthm2.3}

For the proof of
items~\ref{it:gBh_bounds} and~\ref{it:gEh_bounds}, it is convenient to start by proving their analogues for the infinite volume limit propagators.

\paragraph{Decay bounds on $\fg^{[\eta]}_\infty$, $\fg^{(h)}_\infty$.} Recall that 
$\fg^{[\eta]}_\infty$ was defined in \eqref{ginftyeta}. We intend to prove that,  for all $x\in\mathbb Z^2$, 
 \begin{equation}\label{ccases} \|\partial_1^r\partial_2^s\fg^{[\eta]}_\infty(z)\|\le C^{1+r+s} \times \begin{cases}r! s! \eta^{-\frac{3+r+s}2}e^{-\eta^{-1/2}|z|_1}& \text{if}\qquad \eta\ge 1,\\
 e^{-|z|_1} & \text{if}\qquad 0\le \eta\le 1,\end{cases}\end{equation}
where $\partial_j$ is the discrete derivative with respect to the $j$-th coordinate. By using \eqref{ccases} in the definition of 
$\fg^{(h)}_\infty$, namely
\begin{gather}
  \mathfrak{g}_{\infty}^{(h)} (z)=\begin{cases}
  \int_{0}^1  \mathfrak{g}_{\infty}^{[\eta]}(z)
  \ d \eta, & \text{if} \quad h=0,   \\
  \int_{2^{-2h-2}}^{2^{-2h}} 
  \mathfrak{g}_{\infty}^{[\eta]}(z)
  \ d \eta, & \text{if}\quad h < 0,\end{cases}
\label{eee111}  
\end{gather}
we obtain the analogue of \eqref{eq:gBh_bounds}, 
\begin{gather}
\| \partial_1^r\partial_2^s{\mathfrak g}^{(h)}_\infty (z)\|\le C^{1+r+s}r! s!	2^{(1+r+s)h} e^{-c 2^h |z|_1}, 
\label{eq:g8h_bound}
\end{gather}
for all $z\in\mathbb Z^2$ and $h\le 0$.

In order to prove \eqref{ccases}, we start from the explicit expression of the function in the left side, 
\begin{equation} 
	\partial_1^r\partial_2^s\fg^{[\eta]}_\infty(z)=\int\limits_{[-\pi,\pi]^2} e^{-i(k_1z_1+k_2z_2)}(e^{-ik_1}-1)^r(e^{-ik_2}-1)^s e^{-\eta D(k_1, k_2)}\,M(k_1, k_2)\frac{\dd k_1\dd k_2}{(2\pi)^2},
	\label{eq:g8_partial_integral}
\end{equation}
where $D(k_1, k_2)$ is as in \cref{defDk1k2} (cf.\ \cref{eq:fa_def}) and 
\begin{equation} 
	M(k_1, k_2):= 
	D(k_1, k_2) \hat{\fg}(k_1, k_2)
	=
	\begin{bmatrix}
			-2 i t_1 \sin k_1 &
		 -(1-t_1^2)[1 - e^{-ik_2}B(k_1)] \\
		  (1-t_1^2)[1 - e^{ik_2}B(k_1)] &
			 2 i t_1 \sin k_1 
		\end{bmatrix}
		.
	\end{equation}
Note that the integrand in Equation~\eqref{eq:g8_partial_integral} is periodic with period $2\pi$ both in $k_1$ and in $k_2$, and is entire in both arguments. Therefore, 
we can shift both variables in the complex plane, $k_1\to k_1-i a$ 
and $k_2\to k_2 -i b$, 
to obtain
\begin{equation}  
	\begin{split}
		\partial_1^r\partial_2^s\fg^{[\eta]}_\infty(z)
		=
		e^{-a z_1-b z_2}
		\int\limits_{[-\pi,\pi]^2} & e^{-i(k_1x_1+k_2x_2)}(e^{-ik_1-a}-1)^r(e^{-ik_2-b}-1)^s 
		\\
		& \times e^{-\eta D(k_1-ia,k_2-ib)}\,M(k_1-ia,k_2-ib)\frac{\dd k_1\dd k_2}{(2\pi)^2}.
	\end{split}
\end{equation}
We now pick $a=\alpha\,\text{sign}(z_1)$ and $b=\alpha\,\text{sign}(z_2)$, with $\alpha=\min\{\eta^{-1/2},1\}$, and take the absolute value, thus getting 
\begin{equation}  
	\begin{split}
		\|\partial_1^r\partial_2^s\fg^{[\eta]}_\infty(z)\|
		\le & \ e^{-\alpha|z|_1}e^{-2\eta[(1-t_1)^2+(1-t_2)^2](\cosh\alpha-1)} \\
		& \times \int\limits_{[-\pi,\pi]^2} e^{-\eta D(k_1, k_2)}(\alpha e^{\alpha}+|e^{-ik_1}-1|)^r
		(\alpha e^{\alpha}+|e^{-ik_2}-1|)^s M_0(k_1, k_2,\alpha)\frac{\dd k_1\dd k_2}{(2\pi)^2},
	\end{split}
	\label{2.2.30}
\end{equation}
where $M_0(k_1, k_2,\alpha)=\max_{|a|=|b|=\alpha}\| M(k_1-ia,k_2-ib)\|$ 
and we used the fact that, if $|a|=|b|=\alpha$, then 
\begin{equation}
	|D(k_1-ia,k_2-ib)|\ge D(k_1, k_2)+2[(1-t_1)^2+(1-t_2)^2](\cosh\alpha-1)
\end{equation}
and
\begin{equation} 
|e^{-ik_1-a}-1|\le \alpha e^{\alpha}+|e^{-ik_1}-1|,\qquad |e^{-ik_2-b}-1|\le \alpha e^{\alpha}+|e^{-ik_2}-1|.
\end{equation}
Now, if $\eta \le 1$ and, therefore, $\alpha=1$, then \eqref{2.2.30} immediately implies that $\|\partial_1^r\partial_2^s\fg^{[\eta]}_\infty(z)\|\le C^{1+r+s} e^{-|z|_1}$, as desired. 
If $\eta\ge 1$ and, therefore, $\alpha=\eta^{-1/2}$, we make the following observations: (i) the factor $e^{-2\eta[(1-t_1)^2+(1-t_2)^2](\cosh\eta^{-1/2}-1)}$ is bounded from above uniformly in $\eta$;
(ii) if  $-\pi\le k_1, k_2\le \pi$, then $D(k_1, k_2)\ge c(k_1^2+k_2^2)$, $|e^{-ik_1}-1|\le C|k_1|$, $|e^{-ik_2}-1|\le C|k_2|$ and $M_0(k_1, k_2,\eta^{-1/2})\le C(\eta^{-1/2}+|k_1|+|k_2|)$. By using these inequalities in \eqref{2.2.30}, we find 
\begin{equation}  
	\begin{split}
		\|\partial_1^r\partial_2^s\fg^{[\eta]}_\infty(z)\|\le&  C^{1+r+s}e^{-\eta^{-1/2}|z|_1} \\
		& \times \int_{\mathbb R^2} e^{-c \eta(k_1^2+k_2^2)} (\eta^{-1/2}+|k_1|)^r (\eta^{-1/2}+|k_2|)^s (\eta^{-1/2}+|k_1|+|k_2|){\dd k_1\dd k_2},
	\end{split}
\label{eq:g8eta_bound}
\end{equation}
and expanding the powers in the integrand we obtain a sum of Gaussian integrals which reduce to~\eqref{ccases} for this case as well. 

\paragraph{Decay bounds on $\fg^{[\eta]}_{{\rm E}}$, $\fg^{(h)}_{{\rm E}}$, and 
proof of items \ref{it:gBh_bounds} and \ref{it:gEh_bounds}.} 
Recall that $\fg^{[\eta]}_{\E}(z,z')=\fg^{[\eta]}(z,z')-\fg^{[\eta]}_{\B}(z,z')$, with $\fg^{[\eta]}$ as in \eqref{eq:ga_def} and $\fg^{[\eta]}_{\B}$ as in 
\eqref{eq:gaB_def}. We focus on the case that $z_1-z_1'\neq \pm L/2$ (recall that in our conventions $z_1,z_1'\in\{1,\ldots, L\}$), in which
the function $s_L$ in \eqref{eq:gaB_def} is equal to $\pm1$; the complementary case, 
$z_1-z_1'= \pm L/2$, in which $\fg^{[\eta]}_{\B}(z,z')=0$, can be treated in a way analogous to the discussion below, and is left to the reader. 
If $z_1-z_1'\neq \pm L/2$, by using the anti-periodicity of the propagator in the horizontal direction, we can reduce without loss of generality to 
the case $z_1-z'_1=\per_L(z_1-z_1')$ (i.e., $-L/2<z_1-z_1'<L/2$), in which $\fg^{[\eta]}_{\B}(z,z')=\fg^{[\eta]}_{\infty}(z-z')$, and we shall do so in the following. 
Therefore, in this case,
\begin{equation} \label{getastrip}
 \fg^{[\eta]}_{\E}(z, z')=\sum_{k_1\in\mathcal D_L}\sum_{k_2 \in \mathcal{Q}_M(k_1)}
		\frac{1}{2L N_M(k_1, k_2)} \Big[G_+^{[\eta]}(k_1, k_2;z, z' )-G_-^{[\eta]}(k_1, k_2; z, z' )\Big]- \fg^{[\eta]}_\infty (z- z'),\end{equation}
where, recalling the definitions of $f_\eta$ in \eqref{eq:fa_def} and of $\hat \fg(k_1,k_2)$ and $\hat g_{\omega\omega'}(k_1,k_2)$ in \eqref{eq:ghat}, 
\begin{eqnarray}  \label{eq:G_eta_def}
    G_+^{[\eta]}(k_1, k_2;z, z' )
    &:=&
    e^{-ik_1(z_1-z_1')}
    e^{-ik_2(z_2-z_2')} 
    f_\eta(k_1, k_2)
    \hat{\fg}(k_1, k_2)
    ,
    \\
    G_-^{[\eta]}(k_1, k_2; z, z' )
    &:=&	
    e^{-ik_1(z_1-z_1')} e^{-ik_2(z_2+z_2')}f_\eta(k_1, k_2)
    \begin{bmatrix}
      \hat{g}_{++}(k_1, k_2) &
      \hat{g}_{+-}(k_1,-k_2) \\
      \hat{g}_{-+}(k_1, k_2) & e^{2ik_2(M+1)} 
      \hat{g}_{--}(k_1, k_2)
    \end{bmatrix},\nonumber
\end{eqnarray}
which are entire functions of $k_1, k_2$, and $2\pi$-periodic in both variables. We intend to prove that,  for all $z, z'\in\Lambda$,
\begin{equation}\label{ccases_E} 
\|\bs\partial^{\bs r}\fg^{[\eta]}_{\E}(z, z')\|\le C^{1+|\bs r|_1}  \times \begin{cases}\bs r! {\eta^{-\frac{3+|\bs r|_1}2}e^{-c\eta^{-1/2} d_\E(z,z')}}
   & \text{if}\qquad 1\le \eta\le 2^{-2h^*},\\
 e^{-c\, d_\E(z, z')} & \text{if}\qquad 0\le \eta\le 1,\end{cases}
\end{equation}
where we recall that $h^*=-\lfloor \log_2\min\{L,M\}\rfloor$. 
Recalling also the relationship between $\fg_{\E}^{(h)}$ and $\fg_{\E}^{[\eta]}$, this implies 
\begin{eqnarray}
	\| \bs \partial^{\bs r}{\mathfrak g}^{(h)}_{\E}(z, z')\|
	&\le& C^{1+\left|r\right|_1} \bs r ! 
	e^{-c 2^h d_\E (z, z')} 
	\int_{2^{-2h-2}}^{2^{-2h}} \eta^{-(3+|\bs r|_1)/2} \dd \eta\nonumber	\\ 
	&\le&
	(C')^{1+|\bs r|_1}\bs r! 2^{(1+|\bs r|_1)h} e^{-c 2^h d_\E(z, z')}
	, 
	\label{eq:gEstrip_h_bound}
\end{eqnarray}
for $h^* < h < 0$, and $\| \bs \partial^{\bs r}{\mathfrak g}^{(0)}_{\E}(z, z')\|\le C^{1+\left|r\right|_1} e^{-c \, d_\E (z, z')}$ for $h=0$.

Recalling that $\fg^{(h)} = \fg_\infty^{(h)} + \fg_{\E}^{(h)}$ and noting that $d_\E(z, z') \ge \|z - z'\|_1$, inequalities~\eqref{eq:g8h_bound} and~\eqref{eq:gEstrip_h_bound} 
also imply that $\fg^{(h)}$ satisfies a bound of the form \eqref{eq:gBh_bounds}.

\medskip

In order to prove \eqref{ccases_E}, we start from \eqref{getastrip}. Recalling that
$\mathcal Q_M(k_1)$ is the set of roots of $B(k_1)\sin k_2M-\sin k_2(M+1)$, with $k_2\in (-\pi, \pi]$, that  
$N_M(k_1, k_2)$ is given by \eqref{eq:N_M}, and that $G_\sharp^{[\eta]}(k_1, k_2; z, z')$ are entire and $2\pi$-periodic, for $\sharp\in\{\pm\}$, we can rewrite
\begin{equation}
\sum_{k_2\in\mathcal Q_M(k_1)}\frac{1}{2N_M(k_1, k_2)}G_\sharp^{[\eta]}(k_1, k_2;z, z')=\frac 12\oint_{\mathcal C} \frac{B(k_1)\cos k_2M - \cos k_2(M+1)}{B(k_1)\sin k_2M -\sin k_2(M+1)}G_\sharp^{[\eta]}(k_1, k_2; z, z')\frac{\dd k_2}{2\pi i},\label{eccola}\end{equation}
where $\mathcal C$ is the boundary of the rectangle in the complex plane of vertices $-\pi-i b, \pi-i b, \pi+i b, -\pi+ib$, $b>0$, traversed 
counterclockwise. 
We rewrite
\begin{equation}
\frac{B(k_1)\cos k_2M - \cos k_2(M+1)}{B(k_1)\sin k_2M -\sin k_2(M+1)}=\mp i\left[1+2 \frac{R_\pm (k_1, k_2) e^{\pm 2ik_2(M+1)}}{1-R_\pm (k_1, k_2) e^{\pm 2ik_2(M+1)}}\right],
\end{equation}
where
\begin{equation}\label{eq:defRpm}
R_\pm(k_1, k_2):=\frac{1- B(k_1)e^{\mp ik_2}}{1-B(k_1)e^{\pm ik_2}},
\end{equation}
and using this and noting that the contributions of the left and right sides of $\cC$ in the integral in Equation~\eqref{eccola} cancel by the periodicity of the integrand, we then have
\begin{equation}
 \sum_{k_2\in\mathcal Q_M(k_1)}\frac{1}{2N_M(k_1, k_2)}G_\sharp^{[\eta]}(k_1, k_2; z, z')=
\sum_{\sigma=0,\pm1} \int_{-\pi+i\sigma b}^{\pi+i\sigma b} A_\sigma (k_1,k_2) G_\sharp^{[\eta]}(k_1, k_2; z, z')\frac{\dd k_2}{2\pi } \label{eq:B.18}
\end{equation}
where $A_0(k_1,k_2)\equiv 1$ and, if $\sigma=\pm1$, 
\begin{equation}\label{Asigmadef}
A_\sigma(k_1,k_2):=\frac{R_\sigma (k_1, k_2) e^{ 2i\sigma k_2(M+1)}}{1-R_\sigma (k_1, k_2) e^{ 2i\sigma k_2(M+1)}}.
\end{equation}
We now need to sum \eqref{eq:B.18} over $k_1 \in \mathcal D_L$. Notice that the right side of \eqref{eq:B.18} is analytic in $k_1$ in a sufficiently small 
strip around the real axis (quantitative bounds on the width of the analyticity strip will follow), and is $2\pi$-periodic in $k_1$. Given 
any function $F(k_1)$ that is $2\pi$-periodic and analytic in a strip of width $2b>0$ around the real axis,
we have
\begin{equation} \frac1L\sum_{k_1\in\mathcal D_L}F(k_1)=\oint_{\mathcal C}\frac{d k_1}{2\pi}\frac{F(k_1)}{1+e^{-ik_1L}}=\sum_{\sigma'=0,\pm1}
\int_{-\pi+i\sigma' b}^{\pi+i\sigma' b} A'_{\sigma'}(k_1)\, F(k_1) \frac{d k_1}{2\pi},\label{eq:B.20}\end{equation}
where $A'_0(k_1)\equiv 1$ and, if $\sigma'=\pm1$, $A'_{\sigma'}(k_1)=-e^{i\sigma' k_1L}/(1+e^{i\sigma' k_1L})$; moreover $\mathcal C$ is the same contour defined after 
\eqref{eccola}. 
Using \eqref{getastrip}, \eqref{eq:B.18} and \eqref{eq:B.20}, we obtain
\begin{eqnarray}\label{eq:g_E_inside_proof} &&
\bs \partial^{\bs r}\fg^{[\eta]}_{\E} (z, z') = - \iint\limits_{[-\pi,\pi]^2}\frac{\dd k_1\dd k_2}{(2\pi)^2}
\Big[\bs \partial^{\bs r}G_-^{[\eta]}(k_1, k_2; z, z')\\
&&\quad +\sum_{\sharp=\pm}\sharp\sum_{\sigma,\sigma'=0,\pm}^* 
\, A'_{\sigma'}(k_1+i\sigma' b)\, A_\sigma(k_1+i\sigma'b ,k_2+i\sigma b)\, 
\bs \partial^{\bs r}G_\sharp^{[\eta]}(k_1+i\sigma'b , k_2+i\sigma b; z, z')\Big],\nonumber
\end{eqnarray}
where the $*$ on the sum indicates the constraint that $(\sigma,\sigma')\neq(0,0)$, and 
\begin{equation}\bs \partial^{\bs r}G_\sharp^{[\eta]}(k_1, k_2; z, z')=(e^{-ik_1}-1)^{r_{1,1}}(e^{ik_1}-1)^{r_{2,1}}(e^{-ik_2}-1)^{r_{1,2}}(e^{\sharp ik_2}-1)^{r_{2,2}}
G_\sharp^{[\eta]}(k_1, k_2; z, z').\label{eq:partialrG}\end{equation}
Now, by using the definition of $G_-^{[\eta]}$ and by proceeding as in the proof of \eqref{ccases}, we see that 
the first term in the right side of \eqref{eq:g_E_inside_proof} admits the same bound as $\fg^{[\eta]}_\infty(z- z')$, see \eqref{ccases}, 
with the only difference that $|z- z'|_1$ should be replaced by $|z_1-z_1'|+\min\{z_2+z_2',2(M+1)-z_2-z_2'\} \le d_\E (z, z')$. Therefore, 
the first term in the right side of \eqref{eq:g_E_inside_proof} satisfies \eqref{ccases_E} as desired. 

\medskip

Let us now prove that the contribution to $[\bs \partial^{\bs r}\fg^{[\eta]}_{\E} (z, z')]_{\omega\omega'}$ from 
the second line of \eqref{eq:g_E_inside_proof} satisfy \eqref{ccases_E}. For this purpose, if $\sigma'=0$, we shift $k_1$ in the complex plane 
as $k_1\to k_1-i\, b\, {\rm sign}(z_1-z_1')$. If $\sigma=0$, we shift $k_2$ in the complex plane as $k_2\to k_2-i \tau b$, with 
$\tau=\pm$, its specific valued depending on $\sharp$ and on the matrix element $(\omega,\omega')$ we are looking at; more precisely, $\tau=\tau_{\sharp,(\omega\omega')}$, with 
\begin{equation}\tau_{\sharp,(\omega\omega')}:=\begin{cases} -{\rm sign}(z_2-z_2') & \text{if $\sharp=+$}\\
-1 & \text{if 
$\sharp=-$ and $(\omega,\omega')\neq(-,-)$}\\
+1 & \text{if $\sharp=-$ and $(\omega,\omega')=(-,-)$.}\end{cases}\label{tausharp}\end{equation}
Once these complex shifts are performed, we bound the contribution to $[\bs \partial^{\bs r}\fg^{[\eta]}_{\E} (z, z')]_{\omega\omega'}$ from the 
second line of \eqref{eq:g_E_inside_proof} by the sum over $\sharp$ and over $\sigma,\sigma'$ (with $(\sigma,\sigma')\neq(0,0)$) of:
\begin{equation} \label{eq:B.24}
\iint\limits_{[-\pi,\pi]^2}\frac{\dd k_1\, \dd k_2}{(2\pi)^2}\,\big|A'_{\sigma'}(k_1+i\tilde \sigma' b)\big|\, \big|A_\sigma(k_1+i\tilde \sigma'b ,k_2+i\tilde \sigma b)\big|\,
 \Big|[\bs \partial^{\bs r}G_\sharp^{[\eta]}(k_1+i\tilde \sigma'b , k_2+i\tilde \sigma b; z, z')]_{\omega\omega'}\Big|,\end{equation}
where $\tilde \sigma=\begin{cases} \sigma & \text{if $\sigma\neq0$}\\
\tau_{\sharp,(\omega\omega')} & \text{if $\sigma=0$}\end{cases}$, and $\tilde \sigma'=\begin{cases} \sigma' & \text{if $\sigma'\neq0$}\\
-{\rm sign}(z_1-z_1') & \text{if $\sigma'=0$}\end{cases}$. Note that, if $\sigma'\neq 0$, then 
\begin{equation}\label{boundA'}|A'_{\sigma'}(k_1+i\tilde \sigma' b)|\le \frac{e^{-bL}}{1-e^{-bL}}.\end{equation} 
If $\sigma\neq0$, we recall that $A_{\sigma}(k_1,k_2)$ is given by \eqref{Asigmadef}, with 
$R_\sigma(k_1,k_2)$ as in \eqref{eq:defRpm}. We claim that, if $b\le c_0$, with $c_0$ sufficiently small, and $k_1,k_2$ real, then 
\begin{equation} 
\big|R_\sigma(k_1+i\tilde \sigma' b,k_2+i\tilde \sigma b)\big|\le e^{Cb}.\label{eq:RsigmaCb}\end{equation} 
for some $C>0$; this will be proved momentarily, after \eqref{eq:2.2.54}.
We now pick $b=c_0\min\{1,\eta^{-1/2}\}$, so that, using \eqref{eq:RsigmaCb}, for $\sigma\neq0$, $M$ sufficiently large, and $k_1,k_2$ real,  
\begin{equation} |A_\sigma(k_1+i\tilde \sigma' b,k_2+i\tilde \sigma b)|\le \frac{e^{C b}}{1-e^{-bM}}\, e^{-2b(M+1)}. %
\label{boundA}\end{equation}
If we now use \eqref{boundA'} and \eqref{boundA} in the second line of \eqref{eq:g_E_inside_proof} and we estimate the integral of 
$\big|\bs \partial^{\bs r}G_\sharp^{[\eta]}(k_1+i\tilde\sigma'b, k_2+i\tilde \sigma b; z, z')\big|$ via the same strategy used in the 
proof of \eqref{ccases}, see eqs.~\eqref{2.2.30} to \eqref{eq:g8eta_bound}, we find that, for $\eta\le 1$, 
\begin{equation} \eqref{eq:B.24}\le  C^{1+|\bs r|_1}e^{-c_0 d_\E(z,z')},
\end{equation}
where $d_\E$ was defined after \eqref{eq:gEh_bounds}, while, if $1\le \eta\le 2^{-2h^*}$, 
\begin{eqnarray}  \eqref{eq:B.24}&\le&
 C^{1+|\bs r|_1}e^{-c_0\eta^{-1/2}d_\E(z, z')} 
 int_{\mathbb R^2} \dd k_1\dd k_2 e^{-c\eta (k_1^2+ k_2^2)}(|k_1|+|k_2|+\eta^{-1/2})^{1+|\bs r|_1}\nonumber\\
&\le & (C')^{1+|\bs r|_1} 
e^{-c_0\eta^{-1/2}d_\E(z, z')} 
{\bs r}! \eta^{-\frac{3+|\bs r|_1}{2}}.\label{eq:2.2.54}
\end{eqnarray}
This completes the proof that the terms  in the second line of \eqref{eq:g_E_inside_proof} 
satisfy \eqref{ccases_E}, {with $c=c_0$}, provided that the bound \eqref{eq:RsigmaCb} holds.

\medskip

{\it Proof of \eqref{eq:RsigmaCb}.}
We will prove the following version of \eqref{eq:RsigmaCb}: if {$\sigma\neq0$, $b\le c_0$ with $c_0$ sufficiently small and
$a:=\tilde\sigma'b$ with $\tilde\sigma'=\pm1$,} then 
\begin{equation}|R_\sigma(k_1{+}ia, k_2+i\sigma b)|^2\le e^{2b} (1+C_0 b),\label{eq:app.a1}\end{equation}
where $C_0$ can be chosen 
\begin{equation} \label{eq:appC0}C_0=\max\big\{ 8\kappa(1+\pi^2), 128\pi^4\kappa(1-2\kappa)^{-2}\big\},\end{equation}
and $\kappa$ is the same as in \eqref{e:Bk}.
By using the definition of $R_\sigma$, \eqref{eq:defRpm}, one sees that \eqref{eq:app.a1} is equivalent to 
\begin{equation}\label{eq:app.a2}
e^{-2b}+\rho^2-2\rho e^{-b}\cos(\sigma k_2-\beta)\le (1+C_0b)\big[1+\rho^2e^{-2b}-2\rho e^{-b}\cos(\sigma k_2+\beta)\big],
\end{equation}
where $\rho=\rho(k_1,a):=|B(k_1{+}ia)|$ and $\beta=\beta(k_1,a):=\text{Arg}(B(k_1{+}ia))$. As shown below, if $-\pi\le k_1\le \pi$ and $|a|{=b}\le c_0$ with $c_0$ sufficiently small, then 
\begin{equation} \label{appa.bounds}\rho\le 1+\kappa[(1+\pi^2){b}^2-k_1^2/\pi^2]
\,, \qquad |\beta|\le \frac{4\kappa}{1-2\kappa}|k_1|\cdot{b} \,.
\end{equation}
By rearranging the terms in the two sides, one sees that \eqref{eq:app.a2} is equivalent to
\begin{equation}\label{eq:app.a3}
(\rho^2-1)(1-e^{-2b})\le 4 \rho e^{-b}\sin(\sigma k_2)\sin\beta+C_0 b[(1-\rho e^{-b})^2+2\rho e^{-b}\big(1-\cos(\sigma k_2+\beta)\big)].
\end{equation}
By using the first bound in \eqref{appa.bounds}, the fact that $|a|=b\le c_0$ with $c_0$ sufficiently small, and the bound $1-\cos(\sigma k_2+\beta)\ge (\sigma k_2+\beta)^2/\pi^2$ valid for $\beta$ sufficiently small, it is straightforward to check that the 
left side of \eqref{eq:app.a3} is smaller or equal than $4\kappa b[(1+\pi^2)b^2-k_1^2/\pi^2]$, while the right side is greater or equal than 
$$-4\rho e^{-b} |k_2|\cdot|\beta|+C_0b\Big[\frac{b^2}2+\frac2{\pi^2}\rho e^{-b}(|k_2|-|\beta|)^2\Big].$$ 
Therefore, \eqref{eq:app.a3} is a consequence of
\begin{equation}\label{eq:app.a4}
4\kappa (1+\pi^2)b^3+4\rho e^{-b} |k_2|\cdot|\beta|
\le 
4\kappa bk_1^2/\pi^2+C_0\frac{b^3}2+C_0b\frac2{\pi^2}\rho e^{-b}(|k_2|-|\beta|)^2.
\end{equation}
Now, the first term in the left side of \eqref{eq:app.a4} is smaller than the third term in the right side, $C_0b^3/2$, because $C_0\ge 8\kappa(1+\pi^2)$, 
see \eqref{eq:appC0}. By using the second bound in \eqref{appa.bounds}
and the fact that $\rho e^{-b}\le 1$ for $b$ small enough (thanks to the first bound in \eqref{appa.bounds}), 
we see that, if $|k_2|\le \frac{1-2\kappa}{4\pi^2}|k_1|$, then the second term in the left side of \eqref{eq:app.a4} is smaller 
than the first term in the right side. In the complementary case, $|k_2|\ge\frac{1-2\kappa}{4\pi^2}|k_1|$ (which implies, in particular, that
$|k_2|\ge 2|\beta|$, thanks to the second bound in \eqref{appa.bounds}), then the second term in the left side of \eqref{eq:app.a4} is bounded from above 
by $\frac{16\kappa}{1-2\kappa}\rho e^{-b} b |k_2|\cdot|k_1|$, while the last term in the right side is bounded from below by $\frac{C_0b}{2\pi^2}\rho e^{-b}|k_2|^2$;
now, 
$$\frac{16\kappa}{1-2\kappa}\rho e^{-b} b |k_2|\cdot|k_1|\le \frac{C_0b}{2\pi^2}\rho e^{-b}|k_2|^2\quad \Leftrightarrow \quad |k_2|\ge \frac{32\pi^2\kappa}{C_0(1-2\kappa)} |k_1|,
$$
which is verified for $|k_2|\ge\frac{1-2\kappa}{4\pi^2}|k_1|$, because $C_0\ge 128\pi^4\kappa(1-2\kappa)^{-2}$, see \eqref{eq:appC0}. In conclusion, 
\eqref{eq:app.a4} is always verified and, as a consequence, \eqref{eq:app.a3} (and, therefore, \eqref{eq:app.a1}) is, as desired. 

\medskip

We are left with proving the validity of \eqref{appa.bounds} for $|a|{=b}$ small enough. By definition
\begin{equation}\rho e^{i\beta} = 1-\kappa+\kappa\cos k_1\cosh a{-}i\kappa \sin k_1 \sinh a,\label{app.bast}\end{equation}
so that, using $1-\cos k_1\ge 2k_1^2/\pi^2$ and the fact that, for $|a|{=b}$ small, $\cosh a\le 1+{b}^2$ and $|\sinh a|\le 2{b}$, 
$$\rho\le 1+\kappa ({b}^2-2k_1^2/\pi^2+2{b}|k_1|).$$
Using $2{b}|k_1|\le \pi^2 {b}^2+k_1^2/\pi^2$ we get the first of \eqref{appa.bounds}. Finally, from \eqref{app.bast}, we find
$$|\beta|\le \frac{\kappa|\sin k_1\sinh a|}{1-\kappa+\kappa\cos k_1\cosh a}\le \frac{2\kappa{b}|k_1|}{1-\kappa(2+{b}^2)}.$$
Now, picking ${b}^2$ smaller than $(1-2\kappa)/2\kappa$, we find that $\beta$ satisfies the second of \eqref{appa.bounds}. 

\medskip

\paragraph{Gram representation: proof of items \ref{it:gh_gram} and \ref{it:gh_gram_bounds}.}
Recall that 
\begin{eqnarray} &&\hskip-.4truecm \bs\partial^{(\bs s,\bs s')}g^{(h)}_{\omega\omega'}(z, z')=\label{ghcyl}\\
&&\hskip-.2truecm =\int_{I_h} \dd\eta
\sum_{k_1\in\mathcal D_L}\sum_{k_2\in\mathcal Q_M(k)}\frac1{2LN_M(k_1, k_2)} \Big[\bs\partial^{(\bs s,\bs s')}G_{+,\omega\omega'}^{[\eta]}(k_1, k_2; z, z')-
\bs\partial^{(\bs s,\bs s')}G_{-,\omega\omega'}^{[\eta]}(k_1, k_2; z, z')\Big],\nonumber
\end{eqnarray}
where $I_0=[0,1)$, and $I_h=[2^{-2h-2},2^{-2h})$ for all $h^*\le h<0$. We recall that $\bs\partial^{(\bs s,\bs s')}G_\sharp^{[\eta]}$ is given by 
\eqref{eq:partialrG}, with $\bs s$ playing the role of $(r_{1,1},r_{1,2})$ and $\bs s'$ playing the role of $(r_{2,1},r_{2,2})$. In the following, we 
will exhibit a Gram decomposition separately for the two terms in \eqref{ghcyl} corresponding to $G_+^{[\eta]}$ and $G_-^{[\eta]}$, which will 
immediately imply a Gram decomposition for the combination of the two. 

We rewrite 
\begin{equation} \label{eq:Gsharpgamma}\mathds 1_{I_h}(\eta)\, \bs\partial^{(\bs s,\bs s')}G_{\sharp,\omega\omega'}^{[\eta]}(k_1, k_2; z, z')= \sum_{\sigma=1}^4 \big[
\tilde\gamma_{\sharp,\omega,\bs s, z}^{(h)}(k_1, k_2,\eta)\big]^*_\sigma\big[\gamma_{\sharp,\omega',\bs s', z'}^{(h)}(k_1, k_2,\eta)\big]_\sigma,
\end{equation}
where $\mathds 1_{I_h}$ is the characteristic function of the interval $I_h$, and 
$\big[\gamma_{\sharp,\omega,\bs s, z}^{(h)}(k_1, k_2,\eta)\big]_\sigma, \big[\tilde \gamma_{\sharp,\omega,\bs s, z}^{(h)}(k_1, k_2,\eta)\big]_\sigma$ are the 
components of the following 4-vectors:
\begin{eqnarray} 
&& \tilde \gamma_{\sharp,+,\bs s, z}^{(h)}(k_1, k_2,\eta)= \mathds 1_{I_h}(\eta)\,e^{ik_1z_1+ik_2 z_2}(e^{ik_1}-1)^{s_1}(e^{ik_2}-1)^{s_2}\sqrt{f_\eta(k_1, k_2)}
\begin{bmatrix} \left[ \sqrt{\hat g_{\sharp,++}(k_1, k_2)}\right]^*\phantom{.} \\ \left[ \sqrt{\hat g_{\sharp,+-}(k_1, k_2)}\right]^* \phantom{.}\\ 0 \\ 0 \end{bmatrix},\nonumber\\
&& \tilde \gamma_{\sharp,-,\bs s, z}^{(h)}(k_1, k_2,\eta)= \mathds 1_{I_h}(\eta)\, e^{ik_1z_1+ik_2 z_2}(e^{ik_1}-1)^{s_1}(e^{ik_2}-1)^{s_2}\sqrt{f_\eta(k_1, k_2)}
\begin{bmatrix}0\\ 0\\ \left[ \sqrt{\hat g_{\sharp,-+}(k_1, k_2)}\right]^*\phantom{.} \\ \left[ \sqrt{\hat g_{\sharp,--}(k_1, k_2)}\right]^*\phantom{.}  \end{bmatrix},\nonumber\\
&& \gamma_{\sharp,+,\bs s, z}^{(h)}(k_1, k_2,\eta)= \mathds 1_{I_h}(\eta)\, e^{ik_1z_1+\sharp ik_2 z_2}(e^{ik_1}-1)^{s_1}(e^{\sharp ik_2}-1)^{s_2}\sqrt{f_\eta(k_1, k_2)}
\begin{bmatrix} \sqrt{\hat g_{\sharp,++}(k_1, k_2)}\phantom{.} \\ 0 \\ \sqrt{\hat g_{\sharp,-+}(k,q)} \\ 0  \end{bmatrix},\nonumber\\
&& \gamma_{\sharp,-,\bs s, z}^{(h)}(k_1, k_2,\eta)= \mathds 1_{I_h}(\eta)\, e^{ik_1z_1+\sharp ik_2 z_2}(e^{ik_1}-1)^{s_1}(e^{\sharp ik_2}-1)^{s_2}\sqrt{f_\eta(k_1, k_2)}
\begin{bmatrix} 0 \\ \sqrt{\hat g_{\sharp,+-}(k_1, k_2)}\phantom{.} \\ 0 \\ \sqrt{\hat g_{\sharp,--}(k_1, k_2)}  \end{bmatrix},\nonumber
\end{eqnarray}
where: if $\sharp=+$, then $\hat g_{+,\omega\omega'}(k_1, k_2)$, with $\omega,\omega'\in\{\pm\}$,
are the components of the $2\times 2$ matrix $\hat \fg_+(k_1, k_2)\equiv\hat \fg(k_1, k_2)$, see 
\eqref{eq:ghat}; if $\sharp=-$, then $\hat g_{\sharp,\omega\omega'}(k_1, k_2)$, with $\omega,\omega'\in\{\pm\}$,
are the components of 
$$\hat\fg_-(k_1, k_2)\equiv  \begin{bmatrix}
		  \hat{g}_{++}(k_1, k_2) &
		 \hat{g}_{+-}(k_1,-k_2) \\
		  \hat{g}_{-+}(k_1, k_2) & e^{2ik_2(M+1)} 
		  \hat{g}_{--}(k_1, k_2)\end{bmatrix},
		$$
		cf.~\eqref{eq:G_scal_def}. The square roots $\sqrt{\fg_{\sharp,\omega\omega'}(k_1, k_2)}$ of the complex numbers 
$\fg_{\sharp,\omega\omega'}(k_1, k_2)$ are all defined by the same (arbitrarily chosen) branch. 

\medskip

In conclusion, in light of \eqref{eq:Gsharpgamma}, \eqref{ghcyl} can be rewritten as
\begin{eqnarray} &&
 \bs\partial^{(\bs s,\bs s')}\fg^{(h)}_{\omega\omega'}(z, z')=\\
 &&\quad =\int_0^\infty \dd\eta \sum_{k_1\in\mathcal D_L}\sum_{k_2\in\mathcal Q_M(k_1)}\frac1{2LN_M(k_1, k_2)} 
 \sum_{\sigma=1}^4\sum_{\sharp=\pm}\sharp 
\big[
\tilde\gamma_{\sharp,\omega,\bs s, z}^{(h)}(k_1, k_2,\eta)\big]^*_\sigma\big[\gamma_{\sharp,\omega',\bs s', z'}^{(h)}(k_1, k_2,\eta)\big]_\sigma=:
\nonumber\\
&&\quad =: \Big( \tilde\gamma_{+,\omega,\bs s, z}^{(h)}\otimes \hat e_1 + \tilde\gamma_{-,\omega,\bs s, z}^{(h)}\otimes  \hat e_2\  ,\
\gamma_{+,\omega',\bs s', z'}^{(h)}\otimes \hat e_1 - \gamma_{-,\omega',\bs s', z'}^{(h)}\otimes \hat e_2\Big)\equiv 
\Big( \tilde \gamma_{\omega,\bs s, z}^{(h)}, \gamma_{\omega',\bs s', z'}^{(h)}\Big),\nonumber
\end{eqnarray}
where in the last line $\hat e_1, \hat e_2$ are the elements of the standard Euclidean basis of $\mathbb R^2$. 
We can adapt all of the preceding discussion to $\fg^{(\le h)}$ simply by replacing $I_h$ with $[2^{-2h-2},\infty)$;
this concludes the proof of item \ref{it:gh_gram}.

\medskip

In order to prove the bounds  in item \ref{it:gh_gram_bounds}, we first note that the definitions given above for $ \tilde \gamma_{\omega,\bs s, z}^{(h)}$,
$ \gamma_{\omega,\bs s, z}^{(h)}$ immediately imply 
\begin{equation} \begin{split}\Big|  \tilde \gamma_{\omega,\bs s, z}^{(h)}\Big|^2, \Big|  \gamma_{\omega,\bs s, z}^{(h)}\Big|^2&\le 
\int_{I_h}\dd\eta \sum_{k\in\mathcal D_L}\sum_{k_2\in\mathcal Q_M(k)}\frac1{LN_M(k_1, k_2)}  
\big|e^{ik_1}-1\big|^{2s_1}\big|e^{ik_2}-1\big|^{2s_2}\cdot\\
&\cdot f_\eta(k_1, k_2) \sum_{\omega,\omega'=\pm} |g_{\omega\omega'}(k_1, k_2)|.\end{split}\label{grambbb}\end{equation}
Now, recall that the set $\mathcal D_L$ consists of points in $[-\pi,\pi]$ that are equi-spaced at a mutual distance $2\pi/L$, and that the 
set $\mathcal Q_M(k)$ consists of points in $[-\pi,\pi]$ that are {\it almost} equi-spaced at a mutual distance $\pi/(M+1)$
(more precisely, recall that there is exactly one point of $\mathcal Q_M(k)$ 
in every interval $\frac{\pi}{M+1}(n+\frac12,n+1)$, $n=0,\ldots,M-1$, and exactly one point in every interval 
$\frac{\pi}{M+1}(-n-1,-n-\frac12)$, $n=0,\ldots,M-1$). Note also that $N_M(k_1, k_2)\ge M$\footnote{{In fact, by Equation~\eqref{eq:ell_def} and the definition of $\mathcal Q_M(k_1)$, one has $N_M(k_1, k_2)- M=(1-B(k_1)\cos k_2)/(B^2(k_1)-2B(k_1)\cos k_2+1)\ge 0$.}}, 
and that the summand in \eqref{grambbb} is continuous, so we can bound \eqref{grambbb} by a Riemann sum and obtain
\begin{equation}
  \begin{split}
		\Big|  \tilde \gamma_{\omega,\bs s, z}^{(h)}\Big|^2, \Big|  \gamma_{\omega,\bs s, z}^{(h)}\Big|^2
		\le& C
\int_{I_h}\dd\eta \iint_{[-\pi,\pi]^2} \dd k_1 \dd k_2 \big|e^{ik_1}-1\big|^{2s_1}\big|e^{ik_2}-1\big|^{2s_2} f_\eta(k_1, k_2) \sum_{\omega,\omega'=\pm} |g_{\omega\omega'}(k_1, k_2)|
   \\ 
		\le& (C')^{1+s_1+s_2}
		\int_{I_h}\dd\eta 
		\iint_{[-\pi,\pi]^2} \dd k_1 \dd k_2 
		\min(1,|k_1|^{2s_1}|k_2|^{2s_2} e^{-c\eta(k_1^2+k_2^2)}(|k_1|+|k_2|))
    \\
		\le& (C'')^{1+s_1+s_2}
		s_1! s_2!
		\int_{I_h} \min(1,\eta^{-s_1 - s_2 - 3/2}) \dd \eta
		\\
		\le& (C''')^{1+2s_1+2s_2}s_1!s_2! \,2^{h(1+2s_1+2s_2)}.
  \end{split}
\end{equation}
Similarly
\begin{equation}
  \begin{split}
		\Big|  \tilde \gamma_{\omega,\bs s, z}^{(\le h)}\Big|^2, \Big|  \gamma_{\omega,\bs s, z}^{(\le h)}\Big|^2
		\le& (C'')^{1+s_1+s_2}
		s_1! s_2!
		\int_{2^{-2h-2}}^\infty \eta^{-s_1 - s_2 - 3/2} \dd \eta
		\\
		\le& (C''')^{1+2s_1+2s_2}s_1!s_2! \,2^{h(1+2s_1+2s_2)},
  \end{split}
\end{equation}
and these bounds constitute item \ref{it:gh_gram_bounds}.

\section{Proof of Proposition~\ref{thm:g_scaling}}\label{sec:proof_2.9}

Recall that $\Lambda$ is the discrete cylinder of sides $L=2\lfloor a^{-1}\ell_1/2\rfloor$ and $M+1=\lfloor a^{-1}\ell_2\rfloor+1$ and 
$\fg_\scal(z,z')$ the scaling limit propagator \eqref{eq:g_scal_cyl_def} in the continuum cylinder $\Lambda_{\ell_1,\ell_2}$ of sides $\ell_1,\ell_2$.
In order to 
emphasize its dependence upon the sides of the cylinder, let us denote the scaling limit propagator in $\Lambda_{\ell_1,\ell_2}$ by $\fg_\scal(\ell_1,\ell_2;z,z')$.
Note that, upon rescaling by $\xi>0$, this propagator satisfies:
\begin{equation}\xi\fg_\scal(\ell_1,\ell_2; \xi\,z,\xi\,z')=\fg_\scal(\xi^{-1}\ell_1,\xi^{-1}\ell_2;z,z').\label{appresc.prop}\end{equation}
We will prove that, for any $z,z'\in\Lambda$ such that $z\neq z'$, 
and any $w,w'\in \Lambda_{a^{-1}\ell_1,a^{-1}\ell_2}$ such that $w\neq w'$, $\|w-z\|\le \sqrt2$ and $\|w'-z'\|\le\sqrt2$,  
\begin{equation} \big\|\fg_\cc(z,z')-a\fg_\scal(\ell_1,\ell_2;aw,aw')\big\|\le C(\min\{L,M,\|z-z'\|\})^{-2},\label{app.this}\end{equation}
provided that $\min\{L,M,\|z-z'\|\}$ is sufficiently large. 
Proposition~\ref{thm:g_scaling} readily follows from \eqref{app.this}, simply by rescaling by $a^{-1}$. 
Note that, thanks to \eqref{appresc.prop}, 
$a\fg_\scal(\ell_1,\ell_2;aw,aw')=\fg_\scal(a^{-1}\ell_1,a^{-1}\ell_2;w,w')$; note also that $|a^{-1}\ell_1-L|\le 2\sqrt2$ and $|a^{-1\ell_2-M}|\le \sqrt2$. 
By using the explicit expression of the scaling limit propagator \eqref{eq:g_scal_cyl_def} and the fact that $\|w-z\|\le \sqrt2$ and $\|w'-z'\|\le\sqrt2$, 
we find that $\|\fg_\scal(a^{-1}\ell_1,a^{-1}\ell_2;w,w')-\fg_\scal(L,M+1;z,z')\|\le C(\min\{L,M,\|z-z'\|\})^{-2}$. Therefore, in order to prove \eqref{app.this}, it is enough to 
show that, for $\min\{L,M,\|z-z'\|\}$ large, 
\begin{equation} \big\|\fg_\cc(z,z')-\fg_\scal(L,M+1;z,z')\big\|\le C(\min\{L,M,\|z-z'\|\})^{-2},\label{app.this.2}\end{equation}
which is what we will prove in the rest of this appendix. 

Recall that $\fg_\cc (z, z') = \int_0^\infty  \fg^{[\eta]}(z, z')\, d\eta$, with 
\begin{equation}
 {\mathfrak g}^{[\eta]} (z, z') =
    \sum_{\sharp=\pm}\sharp
    \sum_{k_1 \in \mathcal D_L}\sum_{k_2 \in \mathcal{Q}_M(k_1)}
    \frac{1}{2LN_M(k_1, k_2)}     G_{\sharp}^{[\eta]}(k_1, k_2; z, z'),
\end{equation}
where $G_{\sharp}^{[\eta]}$ were defined in \eqref{eq:G_eta_def}. Similarly, $\fg_\scal (L,M+1;z, z') = \int_0^\infty \fg_\scal^{[\eta]} (L,M+1;z, z')\, d\eta$, with
\begin{equation}
  \fg_\scal^{[\eta]}(L,M+1;z, z')
  :=\sum_{\sharp=\pm}\sharp
    \sum_{k_1 \in \frac\pi{L} (2 \bZ + 1) }\
    \sum_{k_2 \in \frac{\pi}{2(M+1)}(2 \bZ + 1)}    \frac{1}{2L(M+1)}
 G_{\scal;\sharp}^{[\eta]}(k_1,k_2;z,z'),
  \label{eq:g_eta_scal_def}
\end{equation}
where $d(k_1, k_2):=(1-t_2)^2k_1^2+(1-t_1)^2k_2^2$ and 
\begin{equation}
  \begin{split}
    G_{\scal;+}^{[\eta]}(k_1, k_2; z, z')
    := &
    e^{-ik_1(z_1-z_1')-ik_2(z_2-z_2')-\eta d(k_1, k_2)}
    \begin{bmatrix} -2it_1 k_1 & -(1-t_1^2)ik_2 \\
     -(1-t_1^2)ik_2 & 2it_1k_1\end{bmatrix}, 
   \\ 
    G_{\scal;-}^{[\eta]}(k_1, k_2; z, z')
    := &
     e^{-ik_1(z_1-z_1')-ik_2(z_2+z_2')-\eta d(k_1, k_2)}
     \begin{bmatrix} -2it_1 k_1 & (1-t_1^2)ik_2 \\
     -(1-t_1^2)ik_2 & e^{2ik_2(M+1)}2it_1 k_1\end{bmatrix}.
        \label{eq:G_scal_def}
      \end{split}
\end{equation}
We rewrite
\begin{equation}
    \fg_{\scal}^{[\eta]}(L,M+1;z, z') - {\mathfrak g}^{[\eta]} (z, z') =  R_1^{[\eta]} (z, z') + R_2^{[\eta]} (z, z') 
    ,\end{equation}
    where
\begin{equation}
R_1^{[\eta]}(z,z'):=\sum_{\sharp=\pm}\sharp\frac1{2L(M+1)}\Big[\sum_{(k_1,k_2)\in\mathcal B_{L,M}}
    G_{\scal;\sharp}(k_1,k_2;z,z')- \sum_{(k_1,k_2)\in\mathcal D_{L,M}} 
    G_{\sharp}^{[\eta]}(k_1, k_2; z, z')\Big],\end{equation}
with $\mathcal B_{L,M}:= \frac\pi{L}(2\mathbb Z+1)\times \frac\pi{2(M+1)}(2\mathbb Z+1)$, 
and $\mathcal D_{L,M}:=\mathcal D_L\times\mathcal D_{2(M+1)}$. Moreover, 
\begin{equation} R_2^{[\eta]}(z,z'):=\sum_{\sharp=\pm}\sharp
\frac1{L}\sum_{k_1\in\mathcal D_L}\Big[\sum_{k_2\in\mathcal D_{2(M+1)}}\frac1{2(M+1)}  -
\sum_{k_2\in\mathcal Q_M(k_1,k_2)}
\frac1{2N_M(k_1,k_2)}\Big]G_{\sharp}^{[\eta]}(k_1, k_2; z, z').\end{equation}

\paragraph{The first remainder term.}
We consider the contribution to $\fg_{\scal} (L,M+1;z, z')-\fg_c(z,z')$ from $R_1^{[\eta]}$ first.
Examining the definitions \eqref{eq:G_eta_def} and \eqref{eq:G_scal_def}, we see that, if $(k_1,k_2)\in \mathcal D_{L,M}$, 
each matrix element of  $G_{\scal;\sharp}^{[\eta]}(k_1, k_2; z, z')-G_{\sharp}^{[\eta]}(k_1, k_2; z, z')$ is bounded in absolute value by 
$C|k|^2e^{-c \eta |k|^2}$, with $|k|^2=k_1^2+k_2^2$, for some $C,c>0$; if $(k_1,k_2)\in \mathcal B_{L,M}\setminus \mathcal D_{L,M}$, 
each matrix element of  $G_{\scal;\sharp}^{[\eta]}(k_1, k_2; z, z')$ is bounded in absolute value by $C|k|e^{-c \eta |k|^2}$. These bounds are sufficient for performing the integral of $\|R_1^{[\eta]}(z,z')\|$ over $\eta\ge (\min\{L,M\})^2$. In fact, for such values of $\eta$, 
\begin{equation}\frac1{2L(M+1)}\sum_{(k_1,k_2)\in\mathcal B_{L,M}} e^{-c \eta |k|^2}\cdot\begin{cases}|k|^2 & \text{if $\max\{|k_1|,|k_2|\}<\pi$}\\
|k| & \text{otherwise}\end{cases} \le C\eta^{-2}, \end{equation}
for some $C>0$, so that 
\begin{equation} \int_{(\min\{L,M\})^2}^\infty \|R_1^{[\eta]}(z,z')\|d\eta\le C \int_{(\min\{L,M\})^2}^\infty \eta^{-2}\,d\eta\le 
C' (\min\{L,M\})^{-2}. \label{R1firstr}\end{equation}
In order to bound the contribution from the integral of $\|R_1^{[\eta]}(z,z')\|$ over $\eta\le (\min\{L,M\})^2$, we need to rewrite $R_1^{[\eta]}(z,z')$ as a suitable integral in the complex plane,
in analogy with what we did in Section~\ref{sec:proofthm2.3}. More precisely, by using \eqref{eq:B.20} and its analogue for the sums over $k_2$, 
we find that the matrix elements of $R_1^{[\eta]}(z, z')$ can be rewritten as: 
\begin{eqnarray} \big[R_1^{[\eta]}(z, z')\big]_{\omega\omega'}&=&\sum_{\sharp=\pm}\sharp\sum_{\sigma_1,\sigma_2=0,\pm}\cdot \label{complexR1}\\
&\cdot& \Big\{
\int_{-\infty+i\sigma_1 b}^{\infty+i\sigma_1b}\frac{dk_1}{2\pi}
\int_{-\infty+i\sigma_2 b}^{\infty+i\sigma_2 b}\frac{dk_2}{2\pi}\, A'_{\sigma_1}(k_1)\, A''_{\sigma_2}(k_2)\, \big[G_{\scal;\sharp}(k_1,k_2;z,z')\big]_{\omega\omega'}\nonumber\\
&&-\int_{-\pi+i\sigma_1 b}^{\pi+i\sigma_1b}\frac{dk_1}{2\pi}
\int_{-\pi+i\sigma_2 b}^{\pi+i\sigma_2 b}\frac{dk_2}{2\pi}\, A'_{\sigma_1}(k_1)\, A''_{\sigma_2}(k_2)\, \big[G_\sharp^{[\eta]}(k_1, k_2; z, z' )\big]_{\omega\omega'}\Big\},\nonumber\end{eqnarray}
where $b$ will be conveniently fixed below, $A'_\sigma(k)$ was defined after \eqref{eq:B.20} and $A''_\sigma$ is defined analogously: $A_0''(k)\equiv 1$ and, if $\sigma=\pm$, 
$A_\sigma''(k)=-e^{2i\sigma k(M+1)}/(1+e^{2i\sigma k(M+1)})$. We now proceed as described before and after \eqref{tausharp}: if $\sigma_1=0$, we shift $k_1$ in the complex plane as $k_1\to k_1-ib\,{\rm sign}(z_1-z_1')$;
if $\sigma_2=0$, depending on the values of $\sharp,\omega,\omega'$, we shift $k_2\to k_2-i\tau_{\sharp,(\omega\omega')}b$, with $\tau_{\sharp,(\omega\omega')}$ as in \eqref{tausharp}. Next, we combine the third line of \eqref{complexR1} with the contribution to the integral in the second line from the region $\max\{|\Re k_1|,|\Re k_2|\}\le \pi$. 
After these manipulations, we find that $\big[R_1^{[\eta]}(z,z')\big]_{\omega\omega'}$ can be further rewritten as
\begin{eqnarray} &&\hskip-.3truecm\big[R_1^{[\eta]}(z, z')\big]_{\omega\omega'}=\sum_{\sharp=\pm}\sharp\sum_{\sigma_1,\sigma_2=0,\pm} \int_{-\infty+i\tilde\sigma_1 b}^{\infty+i\tilde\sigma_1b}\frac{dk_1}{2\pi}
\int_{-\infty+i\tilde\sigma_2 b}^{\infty+i\tilde\sigma_2 b}\frac{dk_2}{2\pi}\, A'_{\sigma_1}(k_1)\, A''_{\sigma_2}(k_2)\cdot \\
&&\quad \cdot \begin{cases} 
\big[G_{\scal;\sharp}(k_1,k_2;z,z')\big]_{\omega\omega'}-\big[G_\sharp^{[\eta]}(k_1, k_2; z, z' )\big]_{\omega\omega'} & \text{if $\max\{|\Re k_1|,|\Re k_2|\}\le \pi$}\\
\big[G_{\scal;\sharp}(k_1,k_2;z,z')\big]_{\omega\omega'} & \text{if $\max\{|\Re k_1|,|\Re k_2|\}>\pi$}\end{cases}\nonumber\end{eqnarray}
where $\tilde\sigma_1= -{\rm sign}(z_1-z_1')$, if $\sigma_1=0$, and $\tilde\sigma_1=\sigma_1$, otherwise; and $\tilde\sigma_2=\tau_{\sharp,(\omega\omega')}$,  if $\sigma_2=0$, and $\tilde\sigma_2=\sigma_2$, otherwise. 
We now pick $b=\eta^{-1/2}$ and notice that, if $\Im k_1=\tilde\sigma_1 \eta^{-1/2}$ and $\Im k_2=\tilde \sigma_2 \eta^{-1/2}$, with $\eta\le(\min\{L,M\})^2$, 
the integrand in the right side of \eqref{complexR1} is bounded in absolute value by 
$$Ce^{\eta^{-1/2}}e^{-\eta^{-1/2}\|z-z'\|-c\eta |k|^2}\cdot\begin{cases} |k|^2+\eta^{-1} & \text{if $\max\{|\Re k_1|,|\Re k_2|\}\le \pi$}\\
|k|+\eta^{-1/2} & \text{if $\max\{|\Re k_1|,|\Re k_2|\}> \pi$}\end{cases} $$
for some $C,c>0$. Therefore, recalling that $\|z-z'\|\gg1$, 
\begin{eqnarray}&&\|R_1^{[\eta]}(z,z')\|\le C e^{-\frac12\eta^{-1/2}\|z-z'\|}\Biggl[\ \int\limits_{[-\pi,\pi]^2} dk \,  e^{-c\eta|k|^2}(|k|^2+\eta^{-1})\\
&&\quad +\int\limits_{\mathbb R^2\setminus [-\pi,\pi]^2} dk \,  e^{-c\eta|k|^2}(|k|+\eta^{-1/2})\Biggr]\le C' e^{-\frac12\eta^{-1/2}\|z-z'\|}(\eta^{-2}+e^{-c'\eta}\eta^{-3/2}).\nonumber
\end{eqnarray}
Note that $e^{-\frac14\eta^{-1/2}\|z-z'\|-c'\eta}\le Ce^{-c''\|z-z'\|^{2/3}}$, so that, by integrating over $\eta\le (\min\{L,M\})^2$, we find: 
\begin{equation}\int_0^{(\min\{L,M\})^2}\|R_1^{[\eta]}(z,z')\|d\eta\le C ( \|z-z'\|^{-2}+\|z-z'\|^{-1}e^{-c''\|z-z'\|^{2/3}})\le C' \|z-z'\|^{-2}.\label{eq:C.15}
\end{equation}
Combining this with \eqref{R1firstr}, we find that $\int_0^{\infty}\|R_1^{[\eta]}(z,z')\|d\eta\le C  (\min\{L,M,\|z-z'\|\})^{-2}$. 

\paragraph{The second remainder term.}
Let us now consider the contribution to $\fg_{\scal} (L,M+1;z, z')-\fg_{\cc}(z,z')$ from $R_2^{[\eta]}$. By using \eqref{eq:B.18} and the analogue of \eqref{eq:B.20} for the sums over $k_2$ in $\mathcal D_{2(M+1)}$, we rewrite $R_2^{[\eta]}(z,z')$ as
\begin{equation} \label{form1R2}R_2^{[\eta]}(z,z'):=\sum_{\sharp=\pm}\sharp
\frac1{L}\sum_{k_1\in\mathcal D_L}\sum_{\sigma=\pm}\int_{-\pi+i\sigma b}^{\pi+i\sigma b}\frac{dk_2}{2\pi} \big[A''_{\sigma}(k_2)-A_{\sigma}(k_1,k_2)\big]
G_{\sharp}^{[\eta]}(k_1, k_2; z, z'),\end{equation}
where $A''_\sigma$ and $A_\sigma$ were defined after \eqref{complexR1} and in \eqref{Asigmadef}, respectively. Note that, for $\sigma=\pm$, 
\begin{eqnarray} A''_{\sigma}(k_2)-A_{\sigma}(k_1,k_2)&=&
-\frac{e^{2i\sigma k(M+1)}}{1+e^{2i\sigma k(M+1)}}-\frac{R_\sigma (k_1, k_2) e^{ 2i\sigma k_2(M+1)}}{1-R_\sigma (k_1, k_2) e^{ 2i\sigma k_2(M+1)}}\label{eq:freq_shift_difference}
\\
&=&
-\frac{2(1+R_\sigma (k_1, k_2)) e^{ 2i\sigma k_2(M+1)}}{
(1+ e^{ 2i\sigma k_2(M+1)})(1-R_\sigma (k_1, k_2) e^{ 2i\sigma k_2(M+1)})}.
\nonumber\end{eqnarray}
Recalling the definition \eqref{eq:defRpm} of $R_\sigma$, %
we have
\begin{equation}
  |1+R_\sigma(k_1, k_2)|
  =
  2
  \left|
  \frac{1-B(k_1)\cos (k_2)}{1-B(k_1)e^{i\sigma k_2}}
  \right|
  \le
  C
  \frac{|k_1|^2 + |k_2|^2+b^2}{b}
  \label{eq:freq_diff_num_bound}
\end{equation}
for $|\Im k_1|\le \sigma \Im k_2 =b$ positive and sufficiently small; and recalling the bound \eqref{eq:RsigmaCb} on $R_\sigma$, which remains valid with $i\tilde\sigma b$ replaced by 
$i\tilde\sigma b'$, $|b'|\le b$, the denominator of \eqref{eq:freq_shift_difference} 
is bounded from below as
\begin{equation}
	|1-R_\sigma (k_1, k_2) e^{ 2i\sigma k_2(M+1)}||1+ e^{ 2i\sigma k_2(M+1)}|
	\ge (1-e^{-b(M+1))^2}
	\label{eq:freq_denom_badcase}
\end{equation}
for $M$ larger than some constant and $|\Im k_1|\le \sigma \Im k_2 =b$ positive and sufficiently small. We now proceed slightly differently, depending on whether 
$\eta$ is larger or smaller than $(\ell(z,z'))^2$, with $\ell(z,z'):=\max\{\min\{L,M\},\|z-z'\|\}$. 

\medskip

{\it The case of $\eta$ smaller than $(\ell(z,z'))^2$.} In this case, we rewrite \eqref{form1R2} by using \eqref{eq:B.20}; if $\sigma'=0$, we perform the complex shift $k_1\to k_1-ib\,{\rm{sign}}(z_1-z_1')$, thus getting (letting $\tilde\sigma'=\sigma'$, if $\sigma'=\pm$, and $\tilde\sigma'=-{\rm sign}(z_1-z_1')$, if $\sigma'=0$)
\begin{equation} R_2^{[\eta]}(z,z'):=\sum_{\sharp=\pm}\sharp\sum_{\sigma'=0,\pm}\sum_{\sigma=\pm}\int_{-\pi+i\tilde\sigma'b}^{\pi+i\tilde\sigma'b}\frac{dk_1}{2\pi}\int_{-\pi+i\sigma b}^{\pi+i\sigma b} \frac{dk_2}{2\pi}\, A'_{\sigma'}(k_1)\big[A''_{\sigma}(k_2)-A_{\sigma}(k_1,k_2)\big]
G_{\sharp}^{[\eta]}(k_1, k_2; z, z').\end{equation}
We now bound the integrand by its absolute value, by using, in particular, \eqref{eq:freq_shift_difference}-\eqref{eq:freq_diff_num_bound}, and by estimating the matrix elements 
of $G_{\sharp}^{[\eta]}$ in the same way as we did several times above and in Appendix \ref{sec:proofthm2.3}. We thus get, for $b=c_0\min\{1,\eta^{-1/2}\}$ with $c_0$ sufficiently small,
\begin{equation}
		\left\|R_2^{[\eta]} (z, z') \right\|
		\le 
		\frac{Ce^{-b \| z- z'\|}}{(1-e^{-bL})(1-e^{-bM})^2}
		\int_{[-\pi,\pi]^2} d k\, \frac{(|k|+b)^3}b
		e^{- \eta |k|^2}	\le
		\frac{Cb^4e^{-b \| z -z'\|}}{(1-e^{-bL})(1-e^{-bM})^2}.
\end{equation}
If we now integrate this inequality with respect to $\eta$, for $0\le \eta\le \ell(z,z')$, recalling that $\ell(z,z')=\max\{\min\{L,M\},\|z-z'\|\}$, we find, for $\|z-z'\time \gg1$, 
\begin{eqnarray}\label{eq:C.22}
\int_0^{(\ell(z,z'))^2}
		\left\|R_2^{[\eta]} (z, z') \right\|
		\dd \eta
		&\le& C\int_0^1 e^{-c_0\|z-z'\|} d\eta+C\int_1^{(\min\{L,M\})^2} \eta^{-2} e^{-c_0\eta^{-1/2}\|z-z'\|}d\eta\\
		&+&\mathds 1(\|z-z'\|>\min\{L,M\})\, \frac{C}{LM^2}
		\int_{(\min\{L,M\})^2}^{\|z-z'\|^2}\eta^{-1/2}e^{-c_0\eta^{-1/2}\|z-z'\|}d\eta\nonumber\\
&\le & C'\Big(e^{-c_0\|z-z'\|}+\frac1{\|z-z'\|^2}+ \frac{\|z-z'\|}{LM^2}\Big)\le \frac{C''}{(\min\{L,M,\|z-z'\|\})^2}.\nonumber
\end{eqnarray}

\medskip

{\it The case of $\eta$ larger than $\ell(z,z')$.} In this case we go back to the representation \eqref{form1R2} (no rewriting of the sum over $k_1$ in terms of an integral in the complex plane). 

We proceed slightly differently for the diagonal and off-diagonal elements of $R_2^{[\eta]}$.
Let us begin with the diagonal terms.  Note that the diagonal elements of $G_\sharp^{[\eta]}$ have the form
\begin{equation}
  \pm 2i t_1 
  e^{- i k_1 (z_1 - z_1') }
  e^{-i k_2 Z_2}
  e^{-\eta D(k_1, k_2)}  \sin k_1 
  \nonumber
\end{equation}
where $Z_2$ is either $z_2 - z_2'$, $z_2 + z_2'$, or $z_2 + z_2' - 2M -2$. 
We thus see that each diagonal element of $R_2^{[\eta]}$ is given by a sum of four terms (due to the sums over $\sharp$ and $\sigma$) of the form
\begin{equation}
  \begin{split}
    &
    \pm
    \frac{2 i t_1 }{L}
    \sum_{k_1 \in \mathcal D_L}
    \int_{-\pi +i\sigma b}^{\pi+i\sigma b}\,  \big[A''_{\sigma}(k_2)-A_{\sigma}(k_1,k_2)\big]
      e^{- i k_1 (z_1 - z_1') }
      e^{-i k_2 Z_2}
      e^{-\eta D(k_1, k_2)}
      \sin k_1 
      \frac{\dd k_2}{2\pi}
      \\ & = 
      \pm
      \frac{2 t_1 }{L}
      \sum_{k_1 \in \mathcal D_L}
      \int_{-\pi +i\sigma b}^{\pi+i\sigma b}
      \frac{2(1+R_\sigma (k_1, k_2)) e^{i \sigma k_2 Z_{2,\sigma}'}	\sin k_1 (z_1 - z_1')     e^{-\eta D(k_1, k_2)}
	\sin k_1 }{
	(1-R_\sigma (k_1, k_2) e^{ 2i\sigma k_2(M+1)})(1+ e^{ 2i\sigma k_2(M+1)})}
  	\frac{\dd k_2}{2\pi},
      \end{split}
      \label{eq:various_sines}
\end{equation}
where in passing from the first to the second line we used \eqref{eq:freq_shift_difference} and the fact that $R_\sigma(k_1,k_2)$ is even in $k_1$. Moreover,
in the second line, $Z'_{2,\sigma}$ is  either $2(M+1)-\sigma(z_2 - z_2')$, $2(M+1)-\sigma(z_2 + z_2')$, or $2(M+1)(1+\sigma)-\sigma(z_2 + z_2')$; in any case, 
$Z'_{2,\sigma}\ge 2$. 
We can then use this, together with Inequalities~\eqref{eq:freq_diff_num_bound} and~\eqref{eq:freq_denom_badcase}
	and the observation that $| \sin k_1 (z_1 - z_1') | \le |k_1| \cdot \|z-z'\| $
	to obtain, for $\omega=\pm$ and $b=c_0\eta^{-1/2}$, with $\eta\ge \ell(z,z')$, 
\begin{equation}
    \Big|\big[R_2^{[\eta]} (z, z')\big]_{\omega\omega}
    \Big|
    \le
    \frac{C\|z-z'\|}{L}\sum_{k_1\in\mathcal D_L} \int_{-\pi}^\pi dk_2\, 
    \frac{ \eta^{1/2} \left( k_1^2 + k_2^2 + \eta^{-1} \right)}{ (1-e^{-c_0\eta^{-1/2}M})^2}
    |k_1|^2  e^{- c\eta (k_1^2 + k_2^2)}\\
    \le \frac{C\|z-z'\|\eta^{-3/2}}{M^2},
\end{equation}
and thus, recalling that $\ell(z,z')=\max\{\min\{L,M\},\|z-z'\|\}$, 
\begin{equation}\label{eq:C.25}
	\left| 
		\int_{(\ell(z,z'))^2}^\infty 
		    \big[R_2^{[\eta]} (z, z')\big]_{\omega\omega}
		    \dd \eta	\right|
    \le 
	\frac{C \|z-z'\|}{M^2\, \ell(x,y)}\le 	\frac{C}{M^2},
\end{equation}
which is of the desired order. 

The off-diagonal elements of $G_\sharp^{[\eta]}$ are equal, up to a sign, to 
\begin{equation}
 2 (1-t_1^2)
  e^{- i k_1 (z_1 - z_1') }
  e^{-i k_2 Z_2} e^{-\eta D(k_1,k_2)}
  (1 - B(k_1) e^{\pm i k_2}),
  \nonumber
\end{equation}
where as before $Z_2$ is either $z_2 - z_2'$, $z_2 + z_2'$, or $z_2 + z_2' - 2M -2$.
Noting that $R_\pm (k_1, k_2) = R_\mp (k_1,-k_2)$, we rewrite
\begin{equation}
  \begin{split}
    \int_{-\pi -i b}^{\pi -i b}
    &
    \big[ A''_-(k_2)-A_-(k_1,k_2)\big] G_\sharp^{[\eta]}(k_1, k_2; z, z')
    \frac{\dd k_2}{2\pi}
    \\ & =
    \int_{-\pi +i b}^{\pi +i b}
    \big[ A''_+(k_2)-A_+(k_1,k_2)\big] G_\sharp^{[\eta]}(k_1, -k_2; z, z')
    \frac{\dd k_2}{2\pi}
  \end{split}
\end{equation}
and so each off-diagonal element of $R_2^{[\eta]}$ can be written as a sum of two terms (due to the sum over $\sharp$) of the form
\begin{equation}
  \begin{split}
    \frac{1}{L}
    \sum_{k_1 \in \mathcal D_L}
    \int_{-\pi +i b}^{\pi+ib}
    &
    \frac{2(1+R_+ (k_1, k_2)) e^{ 2i k_2(M+1)}}{
      (1-R_+ (k_1, k_2) e^{ 2i k_2 (M+1)})(1+ e^{ 2ik_2(M+1)})}
      e^{-ik_1 (z_1-z_1')}
      e^{-\eta D(k_1, k_2)}
      \\ & \times
      \left[ 
	e^{-ik_2 Z_2 } (1 - B(k_1)e^{\pm ik_2})
	+
	e^{i k_2 Z_2 } (1- B(k_1) e^{\mp ik_2})
      \right]
      \frac{\dd k_2}{2 \pi}
    \end{split}
    \nonumber
\end{equation}
up to uninteresting coefficients; then noting that
\begin{equation*}
  \begin{split}
    &
    \tfrac12\left[ 
      e^{-ik_2 Z_2 } (1 - B(k_1)e^{\pm ik_2})
      +
      e^{i k_2 Z_2 } (1- B(k_1) e^{\mp ik_2})
    \right]
    =
    \cos k_2Z_2 - B(k_1) \cos k_2 (Z_2 \mp 1)
    \\ & =
    (1 - \cos k_2) \cos k_2 Z_2 \mp \sin k_2 \sin k_2 Z_2 + [1-B(k_1)] \cos k_2 (Z_2 \mp 1)
    ,
  \end{split}
\end{equation*}
we obtain, for $\omega=\pm$ and $b=c_0\eta^{-1/2}$ with $\eta\ge \ell(z,z')$, noting also that $|Z_2|\le 2M$,  
\begin{equation}
    \left| 
    \big[R_2^{[\eta]} (z,z')\big]_{\omega,-\omega}
    \right|
    \le
    \frac{C|Z_2|}L\sum_{k_1\in\mathcal D_L}
    \int_{-\pi}^\pi dk_2\,
    \frac{\eta^{1/2} \left( k_1^2 + k_2^2 + \eta^{-1} \right)^2}{  (1-e^{-c_0\eta^{-1/2}M})^2}
    e^{- c\eta(k_1^2+k_2^2)} 
    \le
    \frac{C}{M} \eta^{-3/2},
\end{equation}
from which
\begin{equation}
  \left|
  \int_{(\ell(z,z'))^2}^\infty
    \left[ 
      R_2^{[\eta]} (z, z')
    \right]_{\omega,-\omega}
  \dd \eta
  \right|
  \le \frac{ C}{M\ell(z,z')},
\end{equation}
which is again of the desired order. Combining this with \eqref{eq:C.25} and \eqref{eq:C.22}, we find that $\int_0^{\infty}\big\|R_2^{[\eta]} (z, z') \big\|\,d\eta\le C(\min\{L,M,\|z-z'\|\})^{-2}$.
Together with the bound on $R_1^{[\eta]}$, see the line after \eqref{eq:C.15}, this concludes the proof of \eqref{app.this.2} and, therefore, of Proposition \ref{thm:g_scaling}. 

\section{Noninteracting correlation functions in the scaling limit}\label{Sec:pfaffian_proof}

In this appendix, we explain how to express the scaling limit of the non-interacting correlation function appearing in \cref{prop:main} in terms of the propagators studied in \cref{sec:proof_2.9}, thus proving, in particular, \cite[Eq.(\xref{scalinglim.free.1})]{AGG_part1}. 
For notational simplicity, in this appendix we let $t_1^*(\lambda)\equiv t_1$ and $t_2^*(\lambda)\equiv t_2=(1-t_1)/(1+t_1)$. 
By using the Grassmann representation of Proposition \ref{thm:Ising_to_Grassman} in the case $\lambda=0$, we find that, for the lattice of unit mesh 
and any $m$-tuple of distinct edges $x_1,\ldots,x_m$, with $m\ge 2$, 
\begin{equation}\media{\epsilon_{x_1};\cdots;\epsilon_{x_m}}_{0,t_1,t_2;\Lambda}=
\frac{\partial^m}{\partial A_{x_1}\cdots\partial A_{x_m}}\log \int \mathcal D\Phi e^{\mathcal S_{t_1,t_2}(\Phi)+
\sum_{x\in \fB_\Lambda}(1-t_{j(x)}^2)E_x A_x} \Big|_{\bf A=0}
		\end{equation}
(note that the expectation in the left side is the {\it truncated} one). 
Introducing the rescaled energy observable $\varepsilon^a_\ell (z) := a^{-1} \sigma_z \sigma_{z + a \hat e_\ell}$,
rescaling the lattice by a factor of $a$
and passing over to the
non-truncated expectation, we obtain
\begin{eqnarray} && \media{\varepsilon^a_{l_1}(z_1)\cdots\varepsilon^a_{l_m}(z_m)}_{0,t_1,t_2;\Lambda^a}=
\label{appD.2}
\\
&& \qquad =a^{-m}(1-t_1^2)^{m_1}(1-t_2^2)^{m_2}
\media{\left(E_{x(z_1,l_1)}-\media{E_{x(z_1,l_1)}}\right)\cdots \left(E_{x(z_m,l_m)}-\media{E_{x(z_m,l_m)}}\right)},
\nonumber
		\end{eqnarray}
where, in the right side: 
$E_{x(z,1)}=\lis H_z H_{z+a\hat e_1}$ and $E_{x(z,2)}=\lis V_z V_{z+a\hat e_2}$; the symbol $\media{(\cdot)}$ indicates normalized Grassmann measure $\frac{\int\mathcal D\Phi e^{\mathcal S^a_{t_1,t_2}(\Phi)}(\cdot)}{\int\mathcal D\Phi e^{\mathcal S^a_{t_1,t_2}(\Phi)}}$, with $\mathcal S^a_{t_1,t_2}$ the same as \eqref{eq:cS_def} on the rescaled lattice $\Lambda^a$. 

Recall the transformation \eqref{eq:sc_2} relating the variables $\{\lis H_z, H_z, \lis V_z, V_z\}_{z\in\Lambda}$ to $\{\phi_{\omega,z},$ $\xi_{\omega,z}\}_{z\in\Lambda}^{\omega\in\{\pm\}}$, 
from which we see that, if $x$ is a vertical edge of endpoints $z,z+a\hat e_2$, then $E_x=\phi_{+,z}\phi_{-,z+a\hat e_2}$, while, 
if $x$ is a horizontal  edge of endpoints $z,z+a\hat e_1$, then (with obvious notation)
$E_x=\big[s_+*(\phi_{+}-\phi_-)(z)\big]\, \big[s_-*(\phi_{+}+\phi_-)(z+a\hat e_1)\big]$ plus terms involving the `massive' variables $\{\xi_{\omega,z}\}_{z\in\Lambda,\omega\in\{\pm\}}$.

The reader can convince herself that, for the purpose of computing the limit $a\to 0^+$ of \eqref{appD.2}, in the right side of \eqref{appD.2} we can freely replace $E_x$ 
by the following local expressions in the Grassmann `massless' variables: $\phi_{+,z}\phi_{-,z}$, if $x$ is a vertical edge of endpoints $z,z+a\hat e_2$
(note that $\phi_{+,z}\phi_{-,z}$ is obtained from $\phi_{+,z}\phi_{-,z+a\hat e_2}$ by `localizing' the second field at the same position of the first one); and 
$(1+t_1)^{-2}(\phi_{+,z}-\phi_{-,z})(\phi_{+,z}+\phi_{-,z})=2(1+t_1)^{-2}\phi_{+,z}\phi_{-,z}$, if $x$ is a vertical edge of endpoints $z,z+a\hat e_1$ 
(note that $(1+t_1)^{-2}(\phi_{+,z}-\phi_{-,z})(\phi_{+,z}+\phi_{-,z})$ is obtained from $\big[s_+*(\phi_{+}-\phi_-)(z)\big]\, \big[s_-*(\phi_{+}+\phi_-)(z+a\hat e_1)\big]$
by localizing $\big[s_-*(\phi_{+}+\phi_-)(z+a\hat e_1)\big]$ at $z$, and by replacing the non-local, exponentially decaying, kernels $s_\pm(z_1)$ by their local counterparts, namely 
$c_0\delta_{z_1,0}$, with $c_0=\lim_{L\to\infty}\sum_{y=1}^L s_{\pm}(y)=(1+t_1)^{-1}$). It is, in fact, easy to check that the difference between the exact expression of $E_x$ 
and such a `local approximations' is of higher order in $a$ and its contribution to the correlation function vanishes in the limit $a\to 0$. Therefore, 
\begin{eqnarray}
&& \lim_{a\to 0^+}
\media{\varepsilon^a_{l_1}(z_1) \cdots \varepsilon^a_{l_m}(z_m)}_{0,t_1,t_2; \Lambda^a} =\label{eq:Pfaffian_explicit_factor}\\
&&\qquad = \lim_{a\to 0^+}a^{-m}\left(\frac{2(1-t_1^2)}{(1+t_1)^2}\right)^{m_1}(1-t_2^2)^{m_2}\,\media{:\!\phi_{+,z_1}\phi_{-,z_1}\!:\, \cdots\, :\!\phi_{+,z_m}\phi_{-,z_m}\!:},
\nonumber
\end{eqnarray}
where $:\phi_{+,z}\phi_{-,z}:$ denotes the difference 
$\phi_{+,z}\phi_{-,z}-\media{\phi_{+,z}\phi_{-,z}}$. Note that $\frac{2(1-t_1^2)}{(1+t_1)^2}=2t_2$. 
The Grassmann average in the right side of \eqref{eq:Pfaffian_explicit_factor} can be expressed in terms of the fermionic
Wick rule or, equivalently, in terms of the Pfaffian of the $2m\times 2m$ anti-symmetric matrix $\mathcal M^a(\bs z)$, whose elements, labelled by the indices $(1,+), (1,-), \ldots, (m,+),(m,-)$, are equal to 
$$\big[\mathcal M^a(\bs z)\big]_{(i,\omega)(j,\omega')}=\begin{cases}\media{\phi_{\omega,z_i}\phi_{\omega',z_j}} &\text{if $i\neq j$,}\\
0 &\text{otherwise.}\end{cases}$$
In view of Proposition \ref{thm:g_scaling}, $\lim_{a\to 0}a^{-1}\media{\phi_{\omega,z}\phi_{\omega',z'}}=\big[\fg_\scal (z, z')\big]_{\omega\omega'}$ and, therefore, 
\begin{equation}
 \lim_{a\to 0^+}
\media{\varepsilon^a_{l_1}(z_1) \cdots \varepsilon^a_{l_m}(z_m)}_{0,t_1,t_2; \Lambda^a} = (2t_2)^{m_1}(1-t_2^2)^{m_2}\,{\rm Pf}(\mathcal M(\bs z)),
\label{eq:Pfaffian_v2}
\end{equation}
with 
\begin{equation} 
	\big[\mathcal M(\bs z)\big]_{(i,\omega),(j,\omega')}=\begin{cases} \big[\fg_\scal(z_i,z_j)\big]_{\omega\omega'} & \text{if $i\neq j$,}\\
0 & \text{otherwise,}\end{cases} 
\end{equation} 
Since $\fg_\scal$ is covariant under rescaling, see \eqref{appresc.prop}, the scaling limit \eqref{eq:Pfaffian_v2} is, as well. 
Note that rescalings are, together with translations and parity, the only conformal transformations from finite cylinders to finite cylinders or, equivalently, from a finite circular annulus to a finite circular annulus: in fact, it is well known 
\cite{AIMO.Schottky,Schottky}
that an annulus $\{z\in\mathbb C\ :\ r<|z|<R\}$ can be conformally mapped to another one only if the two annuli have the same modulus 
$\frac1{2\pi}\log(R/r)$; moreover, any automorphism of the annulus $\{z\in\mathbb C\ : \ r<|z|<R\}$ is either a rotation 
$z\to ze^{i \theta}$ or a rotation followed by an inversion $z\to Rr/z$. Equivalently,  in terms of finite cylinders, this classical result of complex analysis implies that the 
only conformal transformations from finite cylinders to finite cylinders are uniform rescaling, translations and parity. 

\subsection*{Acknowledgements}
We thank Hugo Duminil-Copin for several inspiring discussions. 
This work has been supported by the European Research Council (ERC) under the European Union's Horizon 2020 research and innovation programme (ERC CoG UniCoSM, grant agreement No.\ 724939 for all three authors and also ERC StG MaMBoQ, grant agreement No.\ 802901 for R.L.G.). 
G.A.\ acknowledges financial support from the Swiss Fonds National.  
A.G.\ acknowledges financial support from MIUR, PRIN 2017 project MaQuMA PRIN201719VMAST01.

\printbibliography[heading=bibintoc]
\end{document}